\documentclass[pra,aps,amsfonts,eqsecnum,superscriptaddress,nofootinbib,notitlepage, longbibliography]{revtex4-1}

\usepackage[a4paper, left=1.5cm, right=1.5cm, top=2cm, bottom=2cm]{geometry}

\usepackage[english]{babel}
\usepackage[utf8]{inputenc}
\usepackage[T1]{fontenc}
\usepackage{amsthm}
\newtheorem{thm}{Theorem}[section]
\newtheorem{hyp}{Hypothesis}[section]

\newtheorem{fact}[thm]{Fact}
\newtheorem{lem}[thm]{Lemma}
\newtheorem{cor}{Corollary}[thm]

\newtheorem{prop}{Proposition}[thm]
\newtheorem{eg}{Example}[thm]
\newtheorem{constraint}{Constraint}[thm]

\usepackage{comment}

\usepackage{amsmath}
\usepackage{graphicx}
\usepackage[colorinlistoftodos]{todonotes}
\usepackage[colorlinks=true, urlcolor=violet, linkcolor=blue, citecolor=red, hyperindex=true, linktocpage=true]{hyperref}

\usepackage{mathtools}
\mathtoolsset{showonlyrefs=true}
\usepackage{bbm,microtype,mathrsfs,amsmath,amssymb,color,amsthm,graphicx,bm} 
\usepackage{tikz-cd}
\usepackage{xcolor}
\usepackage{braket}

\renewcommand{\vec}{\bm}

\newcommand{\CB}{\mathcal{B}}

\newcommand{\BC}{\mathbb{C}}

\newcommand{\CE}{\mathcal{E}}
\newcommand{\BE}{\mathbb{E}}
\newcommand{\CF}{\mathcal{F}}

\newcommand{\CH}{\mathcal{H}}

\newcommand{\CL}{\mathcal{L}}
\newcommand{\CM}{\mathcal{M}}
\newcommand{\CN}{\mathcal{N}}

\newcommand{\CO}{\mathcal{O}}

\newcommand{\CT}{\mathcal{T}}

\newcommand{\vA}{\bm{A}}
\newcommand{\va}{\bm{a}}
\newcommand{\vB}{\bm{B}}
\newcommand{\vC}{\bm{C}}

\newcommand{\vCE}{\bm{\CE}}
\newcommand{\vF}{\bm{F}}

\newcommand{\vH}{\bm{H}}
\newcommand{\vI}{\bm{I}}

\newcommand{\vK}{\bm{K}}

\newcommand{\vO}{\bm{O}}
\newcommand{\vP}{\bm{P}}

\newcommand{\vS}{\bm{S}}

\newcommand{\vU}{\bm{U}}

\newcommand{\vX}{\bm{X}}
\newcommand{\vY}{\bm{Y}}
\newcommand{\vZ}{\bm{Z}}
\newcommand{\vsigma}{\bm{\sigma}}
\newcommand{\vrho}{\bm{\rho}}
\renewcommand{\L}{\left}
\newcommand{\R}{\right}

\newcommand*{\tr}{\mathrm{Tr}}
\newcommand*{\btr}{\bar{\mathrm{Tr}}}
\newcommand*{\poly}{\mathrm{Poly}}

\newcommand*{\Supp}{\mathrm{Supp}}

\newcommand{\lV}{\lVert}
\newcommand{\rV}{\rVert}
\newcommand{\vertiii}[1]{{\left\vert\kern-0.25ex\left\vert\kern-0.25ex\left\vert #1 
    \right\vert\kern-0.25ex\right\vert\kern-0.25ex\right\vert}}
\newcommand{\vertiiiNoLR}[1]{{\bigg\vert\kern-0.25ex\bigg\vert\kern-0.25ex\bigg\vert #1 
    \bigg\vert\kern-0.25ex\bigg\vert\kern-0.25ex\bigg\vert}}
\newcommand{\lvertiii}{\bigg\vert\kern-0.25ex\bigg\vert\kern-0.25ex\bigg\vert }
\newcommand{\rvertiii}{\bigg\vert\kern-0.25ex\bigg\vert\kern-0.25ex\bigg\vert }
\newcommand{\norm}[1]{\Vert {#1} \Vert}
\newcommand{\lnorm}[1]{\left\Vert {#1} \right\Vert}

\newcommand{\normp}[2]{\norm{#1}_{#2}}
\newcommand{\lnormp}[2]{\lnorm{#1}_{#2}}

\newcommand{\lexp}[1]{\exp\L( #1\R)}

\newcommand{\labs}[1]{\left\vert {#1} \right\vert}
\newcommand{\e}{\mathrm{e}}
\newcommand{\iunit}{\mathrm{i}}
\newcommand{\ri}{\mathrm{i}}

\newcommand{\indicator}{\mathbbm{1}}

\begin{document}
\title{Average-case Speedup for Product Formulas}

\author{Chi-Fang (Anthony) Chen}
\email{chifang@caltech.edu}
\affiliation{Institute for Quantum Information and Matter,
California Institute of Technology, Pasadena, CA, USA}
\author{Fernando G.S.L. Brand\~ao}
\affiliation{Institute for Quantum Information and Matter,
California Institute of Technology, Pasadena, CA, USA}
\affiliation{AWS Center for Quantum Computing, Pasadena, CA}
\date{\today}		

\begin{abstract}
     Quantum simulation is a promising application of future quantum computers. Product formulas, or Trotterization, are the oldest and still remain an appealing method to simulate quantum systems. For an accurate product formula approximation, the state-of-the-art gate complexity depends on the number of terms in the Hamiltonian and a local energy estimate. In this work, we give evidence that product formulas, in practice, may work much better than expected. We prove that the Trotter error exhibits a qualitatively better scaling for the vast majority of input states while the existing estimate is for the worst states. For general $k$-local Hamiltonians and higher-order product formulas, we obtain gate count estimates for input states drawn from any orthogonal basis. The gate complexity significantly improves over the worst case for systems with large connectivity. Our typical-case results generalize to Hamiltonians with Fermionic terms, with input states drawn from a fixed-particle number subspace, and with Gaussian coefficients (e.g., the SYK models). Technically, we employ a family of simple but versatile inequalities from non-commutative martingales called \textit{uniform smoothness}, which leads to \textit{Hypercontractivity}, namely $p$-norm estimates for $k$-local operators. This delivers concentration bounds via Markov's inequality. For optimality, we give analytic and numerical examples that simultaneously match our typical-case estimates and the existing worst-case estimates. Therefore, our improvement is due to asking a qualitatively different question, and our results open doors to the study of quantum algorithms in the average case.
\end{abstract}
\maketitle
{ \hypersetup{hidelinks} \tableofcontents }.


\section{Introduction}

A promising application of future quantum computers is to simulate properties of physical systems~\cite{lloyd1996universal,Childs2017TowardTF,babbush2018low,kivlichan2018quantum,2021_Microsoft_catalysis, mcardle2020quantum,chamberland2020building, Davoudi2022QuantumSF}. As a fundamental quantum algorithm subroutine, \textit{Hamiltonian simulation} seeks to efficiently approximate the time evolution operator 
$
    \e^{\ri \vH t}
$
using elementary building blocks, such as a universal gate set or whichever experimentally available operations. Despite the simplicity of the problem statement, developing quantum algorithms that minimize the required resources (e.g., the gate complexity) has drawn tremendous effort~\cite{Berry2005EfficientQA, Low_2019_qubitize, LCU, campbell2019random, thy_trotter_error}, especially given the current limited experimental capability of quantum simulators.

The main Hamiltonian simulation method we study is \textit{product formulas}, or \textit{Trotterization}. As an old idea, it simply approximates the exponential of a sum by products of individual exponentials
\begin{align}
    \e^{\ri (\vH_1+\vH_2 )t }= \e^{\ri\vH_1t } \e^{\ri \vH_2 t } +\CO(t^2).
\end{align}
Constructions such as the Lie-Trotter-Suzuki~\cite{suzuki1991general,lloyd1996universal} formulas generalize to Hamiltonians with many terms and to a higher-order approximation $\CO( t^{\ell+1} )$. However, the \textit{Trotter error}, as hidden in $\CO(t^2)$, had been challenging to analyze, and for a while product formulas were under the shadow of more advanced quantum algorithms based on quantum walks and quantum signal processing~\cite{low2017optimal,Low_2019_qubitize}.

Nevertheless, product formulas have recently resurfaced as a strong candidate for Hamiltonian simulation for experimental, numerical, and theoretical reasons. In the near-term or early-fault-tolerant regime with severe restriction to the number of qubits, depth, and connectivity, its simple prescription without controlled ancilla appears attractive. Despite its simplicity, numerical case studies~\cite{Childs2017TowardTF} suggest product formula may outperform more advanced methods. These reasons further fueled theoretical analysis of Trotter errors where sharper and sharper theoretical guarantees continue to reduce the cost by exploiting the structure of the problem, such as initial state knowledge~\cite{low_energy2020,Tong2021ProvablyAS} and spatial locality of the model~\cite{Haah_2021}.
\begin{figure}[t]
    \centering
    \includegraphics[width=0.7\textwidth]{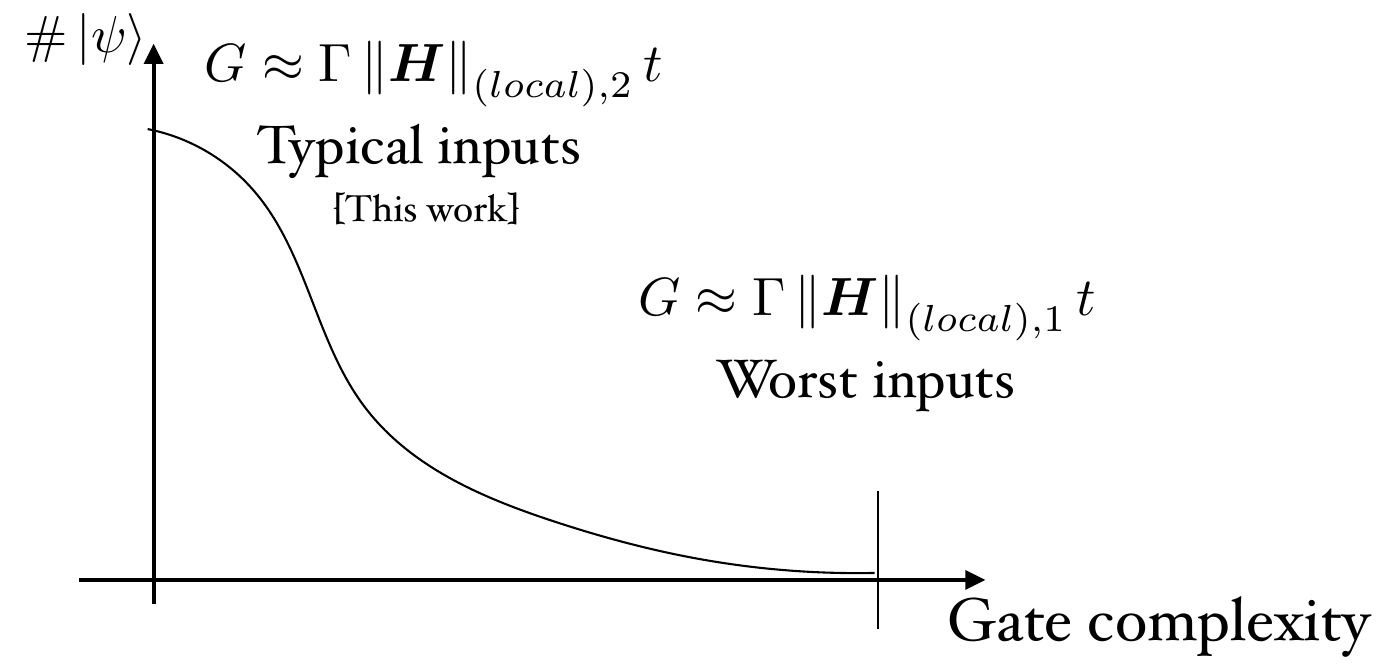}
    \caption{Concentration of gate complexity distribution for product formula for states drawn from any fixed orthogonal basis. The vast majority of states are controlled by our typical case results (Theorem~\ref{thm:Trotter_non_random_maintext}), while extremal states may require the worst case guarantees~\eqref{eq:1_norm_main}~\cite{thy_trotter_error}. The two gate complexities coexist and differ because the Trotter error is a high-dimensional object.}
    \label{fig:worst_typical}
\end{figure}

Especially, the seminal work~\cite{thy_trotter_error} puts together an analytic framework that exploits commutation relations. Consider the general class of \textit{k-local} Hamiltonians $\vH = \sum_{\gamma=1}^{\Gamma} \vH_{\gamma}$ (i.e., sum over few-body Pauli strings $\vsigma^x_1, \vsigma^y_1 \vsigma^y_2,\cdots$). It was shown that using higher-order Suzuki formulas, the gate complexity
\begin{align}
    G\approx \Gamma \lnormp{\vH}{(local),1} t \quad \text{ where}\quad \lnormp{\vH}{(local),1} &:= \max_{\text{site } i} \sum_{\gamma: i \in \gamma}\norm{\vH_\gamma} \label{eq:1_norm_main}
\end{align}
suffices to approximate the unitary evolution for \textit{any} input state. The bound depends on the number of terms $\Gamma$ in the Hamiltonian and a local energy estimate $\lnormp{\vH}{(local),1}$ (Figure~\ref{fig:local_quantity}). This local quantity sums over terms $\vH_\gamma$ overlapping with a site $i$ and takes the maximum over sites; it tends to be much smaller than the global sum $\sum_{\gamma} \norm{\vH_{\gamma}}$. This theoretical guarantee renders product formulas among the strongest candidates for simulating physical systems (Table~\ref{table:main}). 

In light of the developments, we may ask: what remains to be known for Trotter error? In some other contexts, the folklore~\cite{2021_Microsoft_catalysis} suggests errors in quantum computing might, in practice, add up \textit{incoherently}, which can be significantly smaller than coherent noise \cite{hastings2016turning,campbell_mixing16,2017_Temme_error_mitigation}. Intuitively, different scaling occurs whether the noises are ``pointing at the same direction''. For a minimal example, consider a sum over $m$ numbers taking values $\pm1$. In the worst possible scenario, they could all share the same sign and add up coherently. However, if the numbers have random signs independent of each other, the total strength is usually much smaller
\begin{align}
    1+1+\cdots 1 &= m \quad \text{(coherent error)}\\
    1-1+\cdots -1 &\sim 0 \pm \CO(\sqrt{m}) \quad \text{(incoherent error)}.
\end{align}
Curiously, the existing gate complexity, as a manifestation of the Trotter error, exhibits the coherent scaling where terms are added up linearly~\eqref{eq:1_norm_main}. Could Trotter error and the gate complexity, in practice, enjoy the much milder incoherent scaling?

This work presents the incoherent aspects of Trotter error that exhibit qualitatively different scaling from the state-of-the-art estimates~\cite{thy_trotter_error}. Pictorially, the Trotter error is a high-dimensional object that cannot be summarized in a single bound. Instead, there is a \textit{distribution} of Trotter error over input states (Figure~\ref{fig:worst_typical}). The existing estimate~\eqref{eq:1_norm_main} accounts for the \textit{worst-case} inputs that may not be practically relevant; the vast majority of inputs enjoy a much better scaling. More precisely, we show that, \textit{with high probability}, the gate complexity exhibits a \textit{root-mean-square}, or \textit{2-norm} scaling for inputs drawn from any orthogonal basis
\begin{align}
    G\approx \Gamma \lnormp{\vH}{(local),2} t \quad \text{ where}\quad \lnormp{\vH}{(local),2} &:= \max_{\text{site } i} \sqrt{\sum_{S: i \subset S } b^2_S} \quad\text{and}\quad b_{S}:=\sum_{\gamma\sim S} \norm{\vH_{\gamma}}.\label{eq:2_norm_main}
\end{align}
The local quantity $\lnormp{\vH}{(local),2}$ is now a sum-of-squares over sets $S$ overlapping with a site $i$ (the scalar $b_S$ sums over all terms $\vH_{\gamma}$ with support being the set $S$). Our estimate yields substantial improvements over~\eqref{eq:1_norm_main} when the Hamiltonian has large connectivity (such as with long-range interactions, see Table~\ref{table:main}), which directly leads to resource reduction for various quantum simulation tasks. Further, motivated by quantum chaos and the SYK models~\cite{maldacena2016remarks,Sachdev_1993}, we show that when the Hamiltonian itself has random coefficients, even the worst input states enjoy a 2-norm scaling for Trotter error. 

To reiterate, our results give evidence that, in practice, product formulas may generically work even better than expected. This improvement is due to framing a qualitatively different question from the existing worse-case results. Indeed, we provide analytic and numerical evidence that the average-case~\eqref{eq:2_norm_main} and worst-case~\eqref{eq:1_norm_main} estimates can be \textit{simultaneously} tight in their respective contexts. More broadly, our findings open doors to the average-case study of quantum algorithms, which is relatively unexplored yet could greatly improve the feasibility of quantum computing applications.

To derive our average-case results, we combine matrix concentration inequalities (uniform smoothness and hypercontractivity) with the commutator expansion of exponential products~\cite{thy_trotter_error}. The matrix analysis framework is simple and robust, and we expect further applications in quantum information (See, e.g.,~\cite{chen2020quantum,chen2021concentration,chen2021optimal}). 

When this work was completed, we became aware of the work~\cite{Qi_2021_Hamiltonian_simulation_random} which also studies Hamiltonian simulation for random inputs, and we briefly highlight the differences. First,~\cite{Qi_2021_Hamiltonian_simulation_random} studies only the variance of Trotter error, while we show a stronger sense of typicality where the 2-norm scaling holds for all but \textit{exponentially rare} inputs. This utilizes matrix concentration inequalities for the higher moments. Second, our gate complexity is asymptotically tighter for non-spatially local models and is accompanied by analytic and numerical evidence for optimality. This roots from diving deeply into the combinatorics of nested commutators. Third, in addition to random inputs, we also study random Hamiltonians and show the corresponding typical-case results. 

The main text is organized as follows: we summarize results for arbitrary k-local Hamiltonians in Section~\ref{sec:main_non-random} and random Hamiltonians in Section~\ref{sec:main_random}. The gate complexities are compared in Table~\ref{table:main}. We then introduce the proof ingredients in Section~\ref{sec:main_proof_ingredients}.

\begin{figure}[t]
    \centering
    \includegraphics[width=0.4\textwidth]{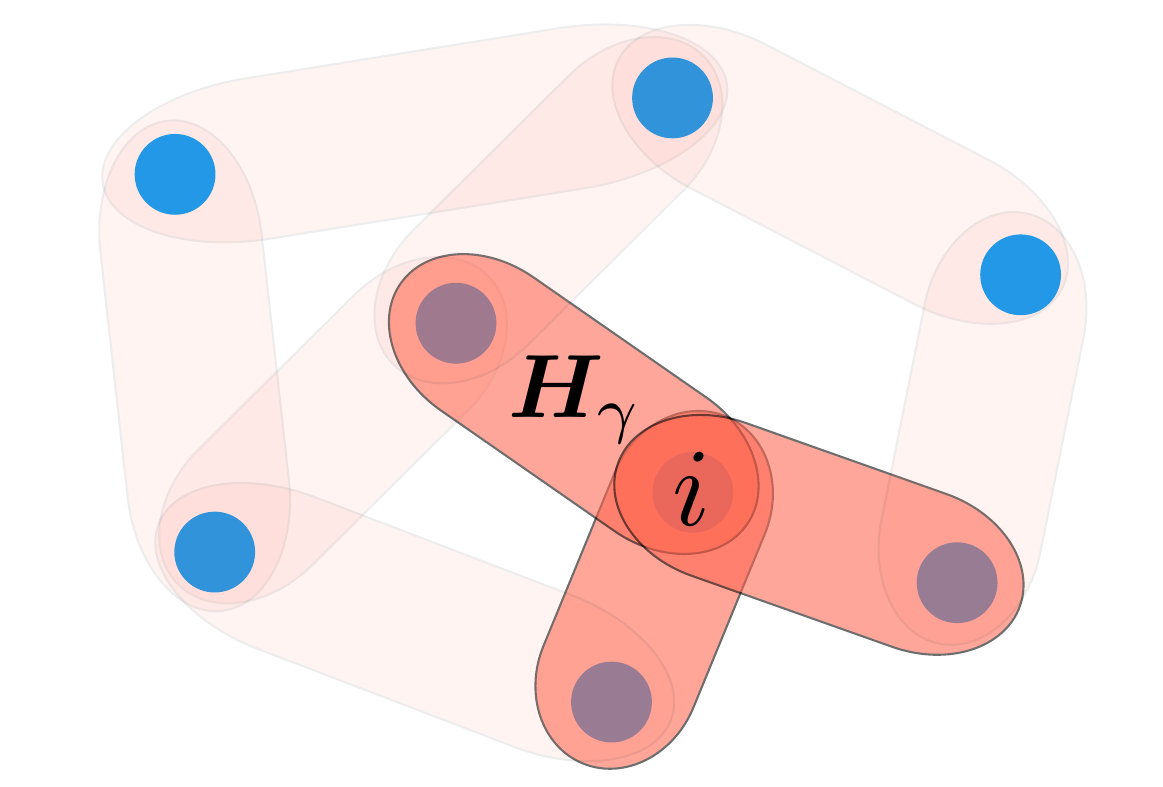}
    \caption{ The local energy estimates $\norm{\vH}_{(local),1}$ and $\norm{\vH}_{(local),2}$ sum over terms overlapping with a site $i$, and maximize over the sites. This is usually smaller than the global Hamiltonian.
    }
    \label{fig:local_quantity}
\end{figure}
\subsection{Summary of Results}
In this section, we present our main results regarding the performances of product formulas. Especially, consider the first-order Lie-Trotter formula and the second-order Suzuki formula 
\begin{align}
     \vec{S}_1(\tau) := \prod^\Gamma_{\gamma=1}\exp(\iunit \tau \vH_\gamma )\quad\text{and}\quad    \vec{S}_2(\tau) := \prod^1_{\gamma=\Gamma}\exp(\iunit(\tau/2) \vH_\gamma )\cdot \prod^\Gamma_{\gamma=1}\exp(\iunit(\tau/2) \vH_\gamma ),
\end{align}
and the higher-order ($\ell = 4, 6,\dots, 2p,\dots $) Suzuki~\cite{suzuki1991general} formulas constructed recursively 
\begin{equation}
\vec{S}_{2p}(\tau) := \vec{S}_{2p-2}(q_p\tau)^2 \cdot \vec{S}_{2p-2}((1-4q_p)\tau)\cdot \vec{S}_{2p-2}(q_p\tau)^2 \quad \text{where}\quad q_p:=1/(4-4^{1/(2p-1)})\label{eq:higher_order_suzuki}.
\end{equation}

\subsubsection{Non-random Hamiltonians}\label{sec:main_non-random}

Here, we consider a $k$-local (i.e., a sum of Pauli strings of length $k$) Hamiltonian on $n$-qubits with $\Gamma$ terms
   $\vH = \sum_{\gamma=1}^\Gamma \vH_\gamma.$
To present our main results, define the normalized Schatten $p$-norms 
$
    \lnormp{\vO}{\bar{p}}:= \frac{\lnormp{\vO}{p}}{\lnormp{\vI}{p}}
$, the vector 2-norm $\norm{\ket{\psi}}_{\ell_2} = \sqrt{\langle \psi|\psi\rangle }$, and a global energy estimate in \textit{2-norm}
\begin{align}
  \lnormp{\vH}{(global), 2} := \sqrt{\sum_{S} b_S^2} \quad\text{and}\quad b_{S}:=\sum_{\gamma\sim S} \norm{\vH_{\gamma}}.
\end{align}

\begin{thm}[Trotter error in $k$-local models] \label{thm:Trotter_non_random_maintext}
To simulate a $k$-local Hamiltonian using $\ell$-th order Suzuki formula, the gate count
\begin{align}
        G =\Omega\L( \left (\frac{p^{k/2}\lnormp{\vH}{(global),2} t}{\epsilon} \right)^{1/\ell} \Gamma p^{(k-1)/2}\lnormp{\vH}{(local),2} t \R)\ \ &\textrm{ensures}\ \  \lnormp{\e^{\iunit \vH t}- \vec{S}_{\ell}(t/r)^r}{\bar{p}} \le \epsilon.\label{eq:pnorm_maintext}
\end{align}
\end{thm}
The $p$-norm estimate implies concentration for typical input states via Markov's inequality.
\begin{cor} 
Draw $\ket{\psi}$ from a state 1-design ensemble such that $\BE[ \ket{\psi}\bra{\psi}] =\vI/\tr[\vI]$ (e.g., an orthonormal basis), then with high probability, the gate count
\begin{align}
    G \approx \left(\frac{\lnormp{\vH}{(global),2} t}{\epsilon} \right)^{1/\ell}  \Gamma \lnormp{\vH}{(local),2} t  \ \ &\textrm{ensures}\ \  \lnormp{(\e^{\iunit \vH t}- \vec{S}(t/r)^r)\ket{\psi}}{\ell_2} \le \epsilon.\label{eq:ell_2_maintext} 
\end{align} 
\end{cor}

See Table~\ref{table:main} for the gate counts in various models and Section~\ref{chap:non_random} for the explicit dependence on the failure probability hidden in~\eqref{eq:ell_2_maintext}. When the Hamiltonian contains Fermionic terms or the input is restricted to a low-particle number subspace, see Proposition~\ref{prop:Trotter_non_random_low} and Proposition~\ref{prop:Trotter_non_random_Fermionic} for analogous results\footnote{This applies to the electronic structure Hamiltonian~\cite{babbush2018low,2021_Su_nearly_tight}, but there the error is dominated by single site terms (1-local Pauli $Z$ s), i.e., $\lnormp{\vH}{(local),2}\sim \lnormp{\vH}{(local),1}$. We only get improvement at lower order product formulas by $\lnormp{\vH}{(global),2}\ll \lnormp{\vH}{(global),1}$.}.

Regarding optimality (Section~\ref{sec:nonrandom_optimal}), we construct a Hamiltonian that demonstrates a separation between the worst case and the typical case bounds: its Schatten $p$-norm saturates our estimates, while the operator norm saturates the state-of-the-art bound~\cite{thy_trotter_error}. Namely, our 1-norm to 2-norm improvement is due to asking a qualitatively different question (Figure~\ref{fig:worst_typical}). 

\begin{prop}[A model with different $p$-norms and spectral norm] Consider a 2-local Hamiltonian on three subsystems of qubits $\CH =\CH_{S_1}\otimes \CH_{S_2}\otimes \CH_{S_3}$
\begin{align}
    \vH = \sum_{s_1\in S_1, s_2\in S_2} \vsigma^z_{s_1}\vsigma^x_{s_2} + \sum_{s_2\in S_2, s_3\in S_3} \vsigma^y_{s_2}\vsigma^z_{s_3}.
\end{align}
Then, at large subsystem sizes $\labs{S_1}=\labs{S_2}=\labs{S_3}\rightarrow \infty$, the first and second-order Trotter at short enough times match the p-norm estimates in Theorem~\ref{thm:Trotter_non_random} and also the spectral norm estimates~\cite{thy_trotter_error} (up to constant factors).
\end{prop}

Note that the dependence on the number of terms $\Gamma$ is not optimal when the terms in the Hamiltonian have non-uniform strengths; we can use a truncation argument~\cite{thy_trotter_error} to improve the gate complexity at early times (Appendix~\ref{sec:trunc_H}). Interestingly, the error due to truncation also enjoys concentration (using Hypercontractivity directly). 
\begin{table}[]\label{table:main}
\centering
\begin{tabular}{l|l|lll}
\hline
qDRIFT~\cite{campbell2019random}      \ \ \ \ \ \   & qubitization~\cite{Low_2019_qubitize}            & higher-order Suzuki \ \ \      &  first-order Trotter      &                \\\hline
$\lnormp{\vH}{(global),1}^2 t^2/\epsilon $ &  $\Gamma'\lnormp{\vH}{(global),1} t$  &$\Gamma\lnormp{\vH}{(local),1} t$   &$\Gamma \lnormp{\vH}{(global),1} \lnormp{\vH}{(local),1} t^2/\epsilon$\ \ &spectral norm~\cite{thy_trotter_error} \\ 
 & & \color{brown}$\Gamma\lnormp{\vH}{(local),2} t$   &\color{brown}$\Gamma \lnormp{\vH}{(global),2} \lnormp{\vH}{(local),2} t^2/\epsilon$\ \ &\color{brown}typical inputs (Theorem~\ref{thm:Trotter_non_random_maintext}) \\\hline
$n^{2}t^2/\epsilon$& $n^2t$ &$nt$  & $n^{2}t^2/\epsilon $ & spatially-local~\cite{Childs2019NearlyOL} \\
 & &\color{brown}$n t$   & \color{brown}$n^{\frac{3}{2}}t^2/\epsilon $ & \\ \hline
$n^{k+1}t^2/\epsilon$& $n^{\frac{3k+1}{2}}t$ &$n^{\frac{3k-1}{2}}t$  & $n^{2k}t^2/\epsilon $ & $k$-local, all-to-all  \\
 & &\color{brown}$n^{k}t$   &  \color{brown}$n^{k+\frac{1}{2}}t^2/\epsilon $&  ($\sum_\gamma {\norm{\vH_\gamma}^2} = \CO(n) $) \\ \hline
  $n^{4-2\alpha/d}t^2/\epsilon$ & $n^{4-\frac{\alpha}{d}}t$ &$n^{3-\frac{\alpha}{d}}t$ & $n^{5-2\frac{\alpha}{d}}t^2/\epsilon $& Power-law 2-local $d/2\le \alpha\le d$\\
  & & \color{brown}$n^{2}t$   & $\color{brown}n^{2+\frac{1}{2}}t^2/\epsilon $&  ($\norm{\vH_{xy}}\le \labs{x-y}^{-\alpha}$) 
\end{tabular}
\caption{Comparison of gate complexities for non-random Hamiltonians for simulation time $t$ and system size $n$ with new results in brown. For the higher-order formulas, we drop asymptotically vanishing dependence $o(1/\ell)$. \textbf{The spatially local} models at higher orders have no average-case speedup because the model has constant connectivity, and the local 1-norm and 2-norm are both independent of system size $\normp{\vH}{(1),1} = const. \normp{\vH}{(1),2}$. We only obtain speed-up at low orders due to the global norm $\normp{\vH}{(0),1} \propto \sqrt{n}\normp{\vH}{(0),2}$.
\textbf{The $k$-local} models have uniform weights for each term, with an SYK-like normalization $\sum_\gamma {\norm{\vH_\gamma}^2} = \CO(n) $). Its large connectivity yields substantial typical-case improvement over the worst-case results. \textbf{The 2-body power-law interacting models} $\vH = \sum_{x,y} \vH_{xy}$ is defined on the d-dimensional lattice with decaying interaction strength $\norm{\vH_{xy}}\le \labs{x-y}^{-\alpha}$. We mainly focus on the regime $d/2 \le \alpha \le d$ with clear-cut improvements. In the qubitization gate counts, we plugged in the number of Hamiltonian terms $\Gamma'=\Gamma$ for comparison, but the terms could potentially be implemented using much fewer parameters $\Gamma'$, such as in quantum chemistry~\cite{2021_Microsoft_catalysis,THC_google}. 
}
\end{table}

Lastly, we present numerics complementing our Trotter error bounds. In particular, we study Trotter error for 2-body Hamiltonians with an on-site disorder, with all-to-all connectivity (Figure~\ref{fig:P2_all_XXYYZZ}, Figure~\ref{fig:time_P2_all_XXYYZZ}, Figure~\ref{fig:P1_P4_all_XXYYZZ}) or nearest-neighbor interactions (Figure~\ref{fig:P2_spatially_local}).\footnote{The Trotter error is dominated by the 2-body terms, which have non-random coefficients.} These models may capture many-body localization and glassy physics. The Trotter error is averaged over realizations of disorder to extract a smooth curve. The disorder also illustrates the robustness of our bounds. Our numerics appear to match the theoretical predictions regarding the dependence on the system size $n$ (Figure~\ref{fig:P2_all_XXYYZZ}, Figure~\ref{fig:P2_spatially_local}), the evolution time $t$ (Figure~\ref{fig:time_P2_all_XXYYZZ}), and the product formula order $\ell$ (Figure~\ref{fig:P1_P4_all_XXYYZZ}).

\begin{figure}[h]
    \centering
    \includegraphics[width=0.95\textwidth]{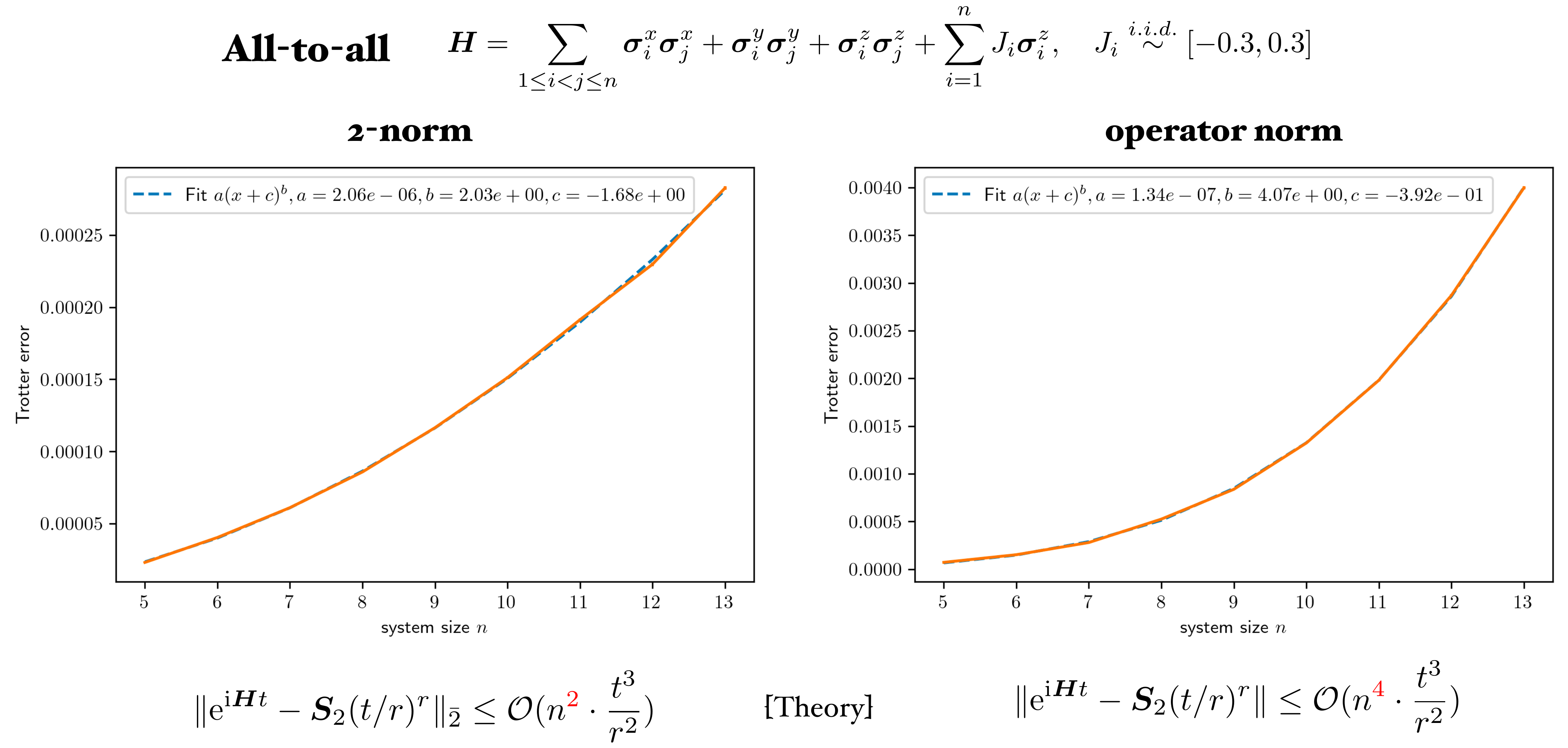}
    \caption{ Trotter error for the all-to-all interacting Heisenberg model for second-order Suzuki formulas $\vS_2(t/r)$. We fix time $t = 10$, repeats $r = 20000$, and change the system size $n = 5,\cdots, 13$. Each Trotter error is estimated by medium-of-mean: take the medium over $27$ bins, where each bin is an average over $32$ independent disorder realization. The fit $a(n+c)^b$ gives the system size dependence $b$. For average inputs (2-norm),  the empirical exponent reads $b=2.03\pm 0.03$, which matches the theoretical bound (Theorem~\ref{thm:Trotter_non_random_maintext}, $b = 2$). For worst inputs (operator norm), the empirical exponent is much larger, $b=4.07\pm 0.13$, which matches the theoretical bound (\cite{thy_trotter_error}, $b = 4$).
    }
    \label{fig:P2_all_XXYYZZ}
\end{figure}

\begin{figure}[h]
    \centering
    \includegraphics[width=0.6\textwidth]{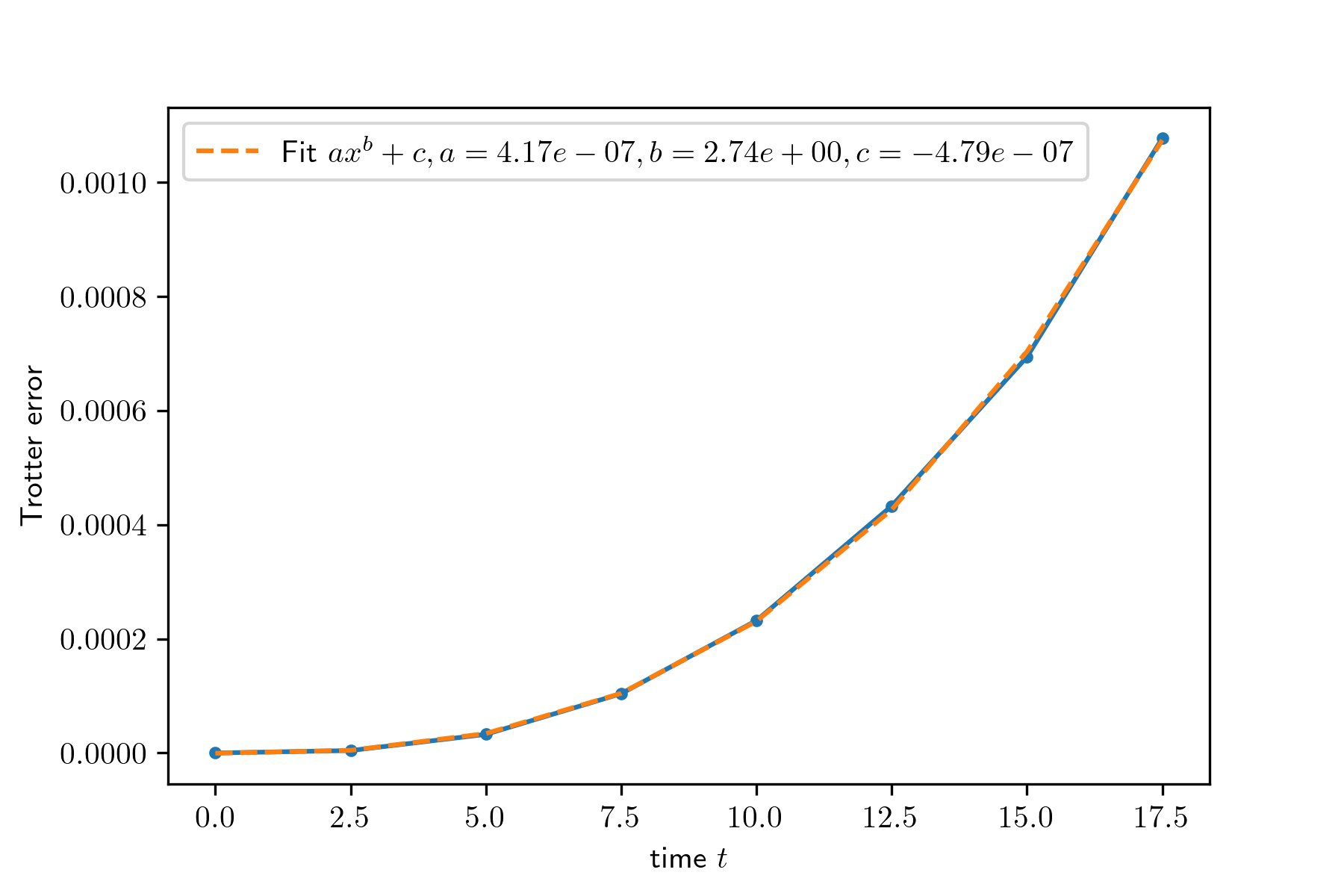}
    \caption{ Time dependence of the 2-norm Trotter error in Figure~\ref{fig:P2_all_XXYYZZ}. We fix repetition $r = 20000$, the system size $n = 12$, and change time $t = 0, \cdots 17.5$. 
    The fit $a t^b+c$ gives the time dependence exponent $b= 2.74 \pm 0.04$ (variance calculated by independent runs), which deviates slightly from the theoretical upper bounds (Theorem~\ref{thm:Trotter_non_random_maintext}, $b = 3$).
    }
    \label{fig:time_P2_all_XXYYZZ}
\end{figure}

\begin{figure}[h]
    \centering
    \includegraphics[width=0.95\textwidth]{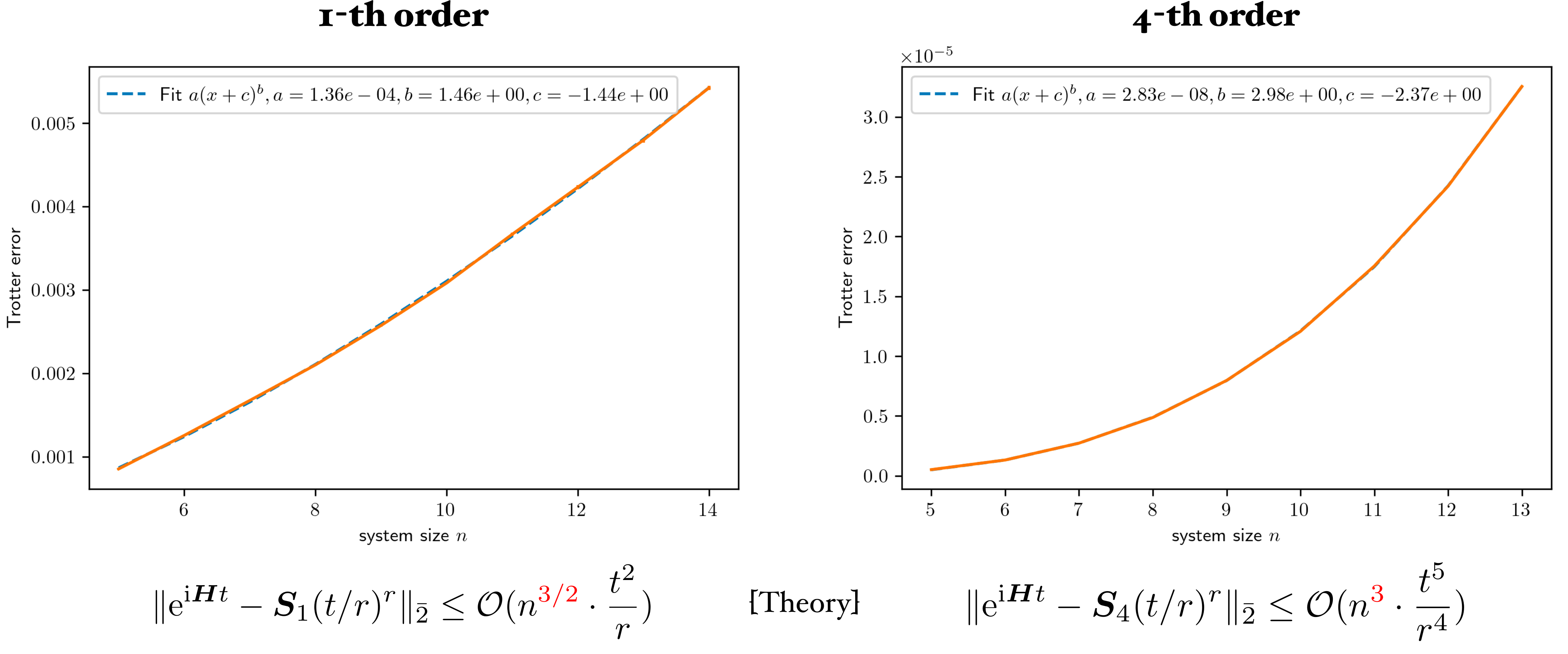}
    \caption{ 
    Different orders of Suzuki formulas for the all-to-all interacting Heisenberg model. For the first-order Lie-Trotter-Suzuki formula, the parameters are $t = 5, r = 200000, n = 5,\cdots, 14$. We take medium over 8 bins, each averaging over 12 runs. The fit $a(n+c)^b$ gives the empirical system size dependence $b=1.46\pm 0.03$, matching the theoretical bound (Theorem~\ref{thm:Trotter_non_random_maintext}, $b = 1.5$); the parameters for 4-th order formula are: $t=10, r =1000, n = 5,\cdots, 13$. We take medium over 32 bins, each averaging over 15 runs. The empirical exponent reads $b=2.98\pm 0.03$, matching the theoretical bound (Theorem~\ref{thm:Trotter_non_random_maintext}, $b = 3$).
    }
    \label{fig:P1_P4_all_XXYYZZ}
\end{figure}

\begin{figure}[h]
    \centering
    \includegraphics[width=0.95\textwidth]{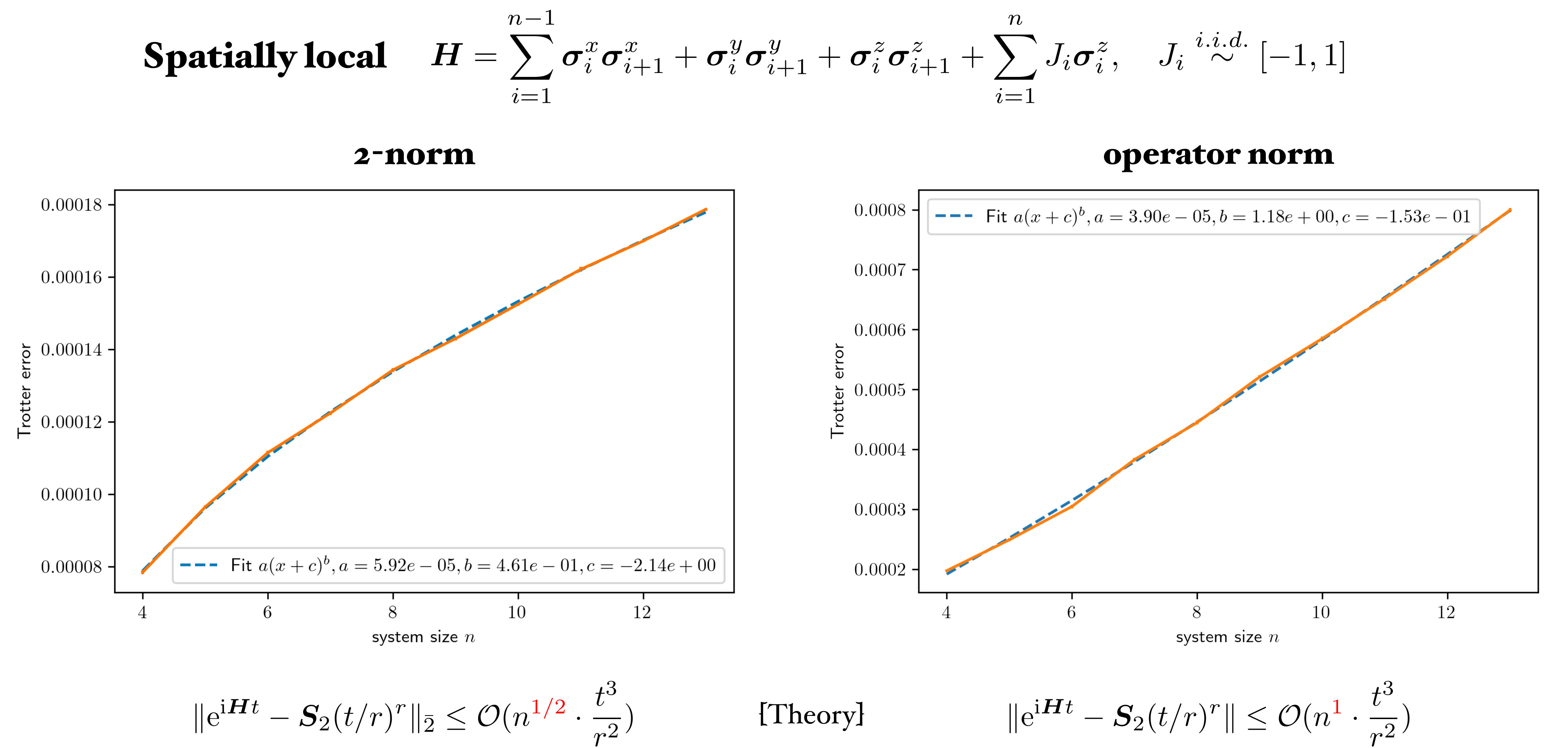}
    \caption{ Trotter error for the spatially-local Heisenberg model for second-order Suzuki formulas $\vS_2(t/r)$. We fix time $t = 50$, repeats $r = 40000$, and change the system size $n = 5,\cdots, 13$. We take medium over $15$ bins, where each bin is an average over $32$ independent disorder. 
    The fit $a(n+c)^b$ gives the system size dependence $b$. 
    For average inputs (2-norm),  the empirical exponent reads $b=0.46\pm 0.01$, which matches the theoretical bound (Theorem~\ref{thm:Trotter_non_random_maintext}, $b = 1$). For worst inputs (operator norm), the empirical exponent is much larger, $b=1.18\pm 0.02$, which is consistent with the theoretical bound (\cite{thy_trotter_error}, $b = 1$).
    }
    \label{fig:P2_spatially_local}
\end{figure}
\begin{figure}[h]
    \centering
    \includegraphics[width=0.95\textwidth]{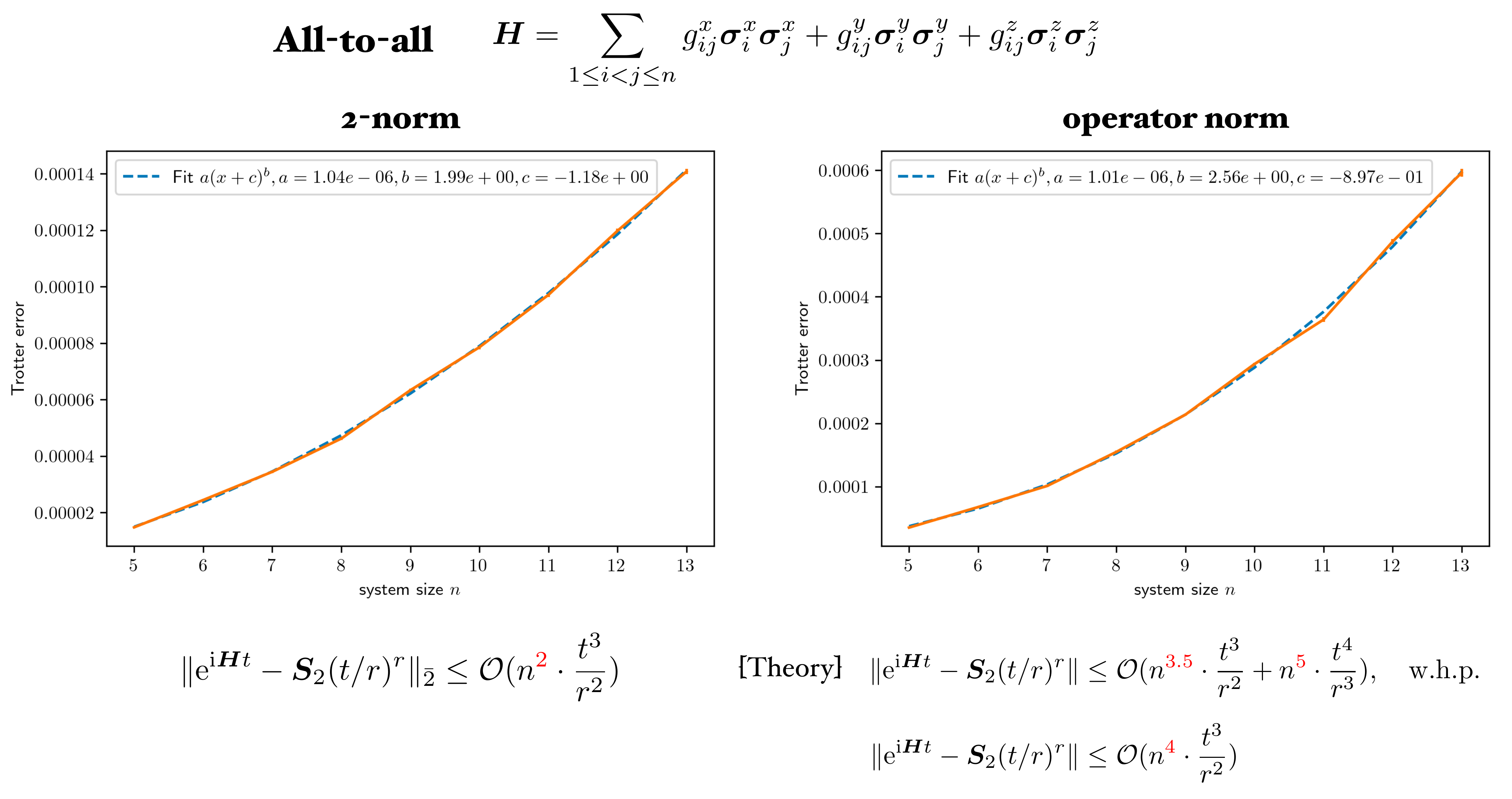}
    \caption{ Trotter error for the random all-to-all Heisenberg model for second-order Suzuki formulas $\vS_2(t/r)$. We fix time $t = 10$, repeats $r = 20000$, and change the system size $n = 5,\cdots, 13$. We take medium over $15$ bins, each averaging over $32$ independent disorder. 
    The fit $a(n+c)^b$ gives the system size dependence $b$. For worst inputs (operator norm), the empirical exponent is $b=2.56\pm 0.1$, which is smaller than the theoretical bound for random Hamiltonians (Theorem~\ref{thm:Trotter_random_H}, $b = 3.5$\footnote{The $b=5$ part should be suppressed at this value of repeats $r=20000$.} ) and non-random Hamiltonians (\cite{thy_trotter_error}, $b = 4$.). We are unable to numerically optimize the fixed input state for the norm $\normp{\cdot}{fix,2}$; we only present the numerics for average inputs (2-norm) for a comparison. 
    }
    \label{fig:P2_random_all}
\end{figure}

\subsubsection{Random Hamiltonians}\label{sec:main_random}
Sometimes, we are interested in an ensemble of Hamiltonians, most notably the Sachdev-Ye-Kitaev~\cite{Sachdev_1993,maldacena2016remarks} models with random coefficients. The intrinsic randomness of the Hamiltonian allows us to obtain similar but stronger results.   
More precisely, we consider random Hamiltonians $
   \vH = \sum_{\gamma=1}^\Gamma \vH_\gamma =   \sum_{\gamma=1}^\Gamma g_\gamma \vK_\gamma, 
$
where the coefficients $g_\gamma$ are i.i.d. standard Gaussian $\BE[g_\gamma^2]=1$, and the matrices $\vK_\gamma$ are deterministic. The local quantities here are defined by dropping Gaussians
\begin{align}
 \lnormp{\vH}{(global),2} := \sqrt{\sum_{\gamma} b_\gamma^2} \quad\text{and}\quad
 \lnormp{\vH}{(local),2} &:= \max_{\text{site }i}  \sqrt{\sum_{\gamma: i \subset \gamma } b^2_\gamma}, \quad\text{where}\quad b_{\gamma}:= \norm{\vK_{\gamma}}.
\end{align}

\begin{thm}[(informal) Trotter error in random models] \label{thm:Trotter_random_H_maintext}
Simulating random $k$-local models with Gaussian coefficients via higher-order ($\ell\rightarrow \infty$) Suzuki formulas, the asymptotic gate count 
\begin{align}
     G&\approx 
    \Gamma \lnormp{\vH}{(local),2} t\sqrt{n} \quad\textrm{ensures}\quad\lnorm{e^{\iunit \vH t}- \vec{S}(t/r)^r} \le \epsilon &(\text{all inputs}),\label{eq:all_input_maintext}\\
    G&\approx \Gamma\lnormp{\vH}{(local),2} t \quad\textrm{ensures}\quad \lnormp{(\e^{\iunit \vH t}- \vS(t/r)^r )\ket{\psi} }{\ell_2}\le \epsilon &(\text{fixed input state}),\label{eq:fixed_input_maintext}
\end{align}
with high probability drawing from the random Hamiltonian ensemble. The fixed input state $\ket{\psi}$ can be arbitrary.
\end{thm}
 See Section~\ref{chap:random} for the complete theorem depending on the finite order $\ell\ne \infty$ and the failure probability and Theorem~\ref{thm:first_order_gate_count} for a precise gate count for the first-order Trotter formula. In other words, when the Hamiltonian is random, an arbitrary fixed input state exhibits 2-norm scaling of Trotter error. A slightly higher gate count (by a factor of the system size $\sqrt{n}$) would control the performance for the worst inputs that may correlate with the Hamiltonian (e.g., the Gibbs state or the ground state of the model).

\begin{prop}[Distinct Hamiltonians]
There exists a set of k-local Hamiltonians $\{\vH^{(i)}\}$ with cardinality $\e^{\Omega(\Gamma)}$ such that each of them satisfies 
\begin{align}
\vH^{(i)} = \sum_{\gamma=1}^\Gamma \vH^{(i)}_\gamma \quad \text{where} \quad \norm{\vH^{(i)}_\gamma} \le \CO(1/n^{\frac{k-1}{2}}),
\end{align}
but for early times $t=\Omega(1)$ they are pairwise distinct
\begin{align}
    \lnormp{\vH^{(i)}- \vH^{(j)}}{\infty}t \ge \Omega(\sqrt{n}) \quad \text{for each pair} \quad i\ne j.
\end{align}
\end{prop}
If we further assume the matrix exponentials are also distinct (which is believable but harder to prove) $ \lnormp{\e^{\vH^{(i)}t}- \e^{\vH^{(j)}t}}{\infty} \stackrel{?}{\ge} \Omega(1), $
this implies a counting circuit complexity lower bound\footnote{For SYK ($k=4$, Majorana, Gaussian coefficients),~\cite{SYK_Babbush} uses qubitization to obtain a gate count 
$
  G=  n^{7/2}t + n^{5/2}t \log(n/\epsilon), 
$
which is lower than our $n^{k}t$. However, their $n^4$ Gaussians coefficients are not independent and hence not controlled by our circuit lower bounds. Physically, it is not clear whether SYK models with pseudorandom coefficients mimic the original ones.
} 
$G=\Omega(\Gamma)=\Omega(n^k)$, which matches our gate complexity for fixed inputs \eqref{eq:fixed_input_maintext} and typical input~\eqref{eq:pnorm_maintext} at early times $t=\theta(1)$. See Section~\ref{sec:counting} for the proof.

The general optimality of our bounds for random Hamiltonians is less understood numerically. We present qualitative evidence (Figure~\ref{fig:P2_random_all}) suggesting that the Trotter error for random Hamiltonians, in the operator norm, could be much smaller than that of non-random Hamiltonians~\cite{thy_trotter_error}. At the scale of our numerics, the error seems even smaller than our theoretical estimates. Unfortunately, we are not able to numerically estimate the norm $\normp{\cdot}{fix,2}$ for fixed inputs.

\subsection{Proof ingredients}\label{sec:main_proof_ingredients}
The Trotter error is a complicated function of matrices. The leading order Trotter error is a commutator; for example, in the first-order product formula
\begin{align}
    \vS_1(t)-\e^{\sum^{\Gamma}_{\gamma=1}\ri\vH_\gamma t}
    &= \frac{t^2}{2}\sum^\Gamma_{\gamma'>\gamma\ge 1}[\ri\vH_{\gamma'},\ri\vH_\gamma] +O(t^3).
\end{align}
Analogously, the $\ell$-th order product formulas have leading order errors as a degree $\ell+1$ polynomial of commutators~\cite{thy_trotter_error}.

There are two main technical steps: First, how to take care of the infinite series of higher-order terms? Second, how to deliver concentration bounds for the commutator? 

\subsubsection{A Good Presentation of Error}
The Trotter error has a rather nasty higher-order dependence on time, and a good expansion simplifies the proof. Here we build upon the framework from~\cite{thy_trotter_error}. Denote the target Hamiltonian
$
   \vH = \sum_{\gamma=1}^\Gamma \vH_\gamma,  
$
with some labels $\gamma$ for the summand. We specify a product formula $\e^{\ri a_J\vH_{\gamma(J)} t}\cdots \e^{\ri a_1 \vH_{\gamma(1)} t}$ with $J$ exponentials by choosing an ordering $\gamma(j)$ and weights $a_j$.
In particular, we will focus on the Suzuki formulas~\eqref{eq:higher_order_suzuki}, which can be rewritten as 
\begin{align}
    \vS_{\ell}(t) = \prod_{j=1}^{J} \e^{\ri a_{j}\vH_{\gamma(j)} t} = \prod_{\nu = 1}^{\Upsilon} \prod_{i=1}^{\Gamma} \e^{\ri a_{\nu,j} \vH_{\gamma(i, \nu)} t} \quad \text{where} \quad \Upsilon = 2\cdot 5^{\ell/2-1} \quad \text{and} \quad \labs{a_{\nu,i}} \le 1.
\end{align}
For the first-order Lie-Trotter formula, each term appears once, so there is only one stage $\Upsilon =1$; the higher-order Suzuki formula has a total of $J = \Gamma \cdot \Upsilon$ exponentials and decomposes into $\Upsilon$ stages, where each stage goes through each Hamiltonian term $\vH_{\gamma}$ exactly once.

Following~\cite{thy_trotter_error}, the Trotter error can be captured in the time-ordered exponential form
\begin{align}
    \prod_{j=1}^{J} \e^{\ri a_{j}\vH_{\gamma(j)} t} = \exp_{\CT}\L(\ri\int \L(\vec{\CE} (\vH_1,\cdots,\vH_\Gamma, t)+ \vH\R)dt\R).
\end{align}

The error is now represented as a sum of nested commutators
\begin{align}
    \vec{\CE} (\vH_1,\cdots,\vH_\Gamma, t) :&= \sum^{J}_{j=1}\left( \prod_{k=j+1}^{J} \e^{a_k\CL_{\gamma(k)} t} [a_j\vH_{\gamma(j)}] -a_j\vH_{\gamma(j)}\right)\quad \text{where}\quad \CL_{\gamma}[O]:=\ri[\vH_\gamma,O].
\end{align}
In our proof, we will ``beat the nested commutator $\vCE$ to death''; do a Taylor expansion on time $t$ for the nested commutators, and each order will be a polynomial of matrices. (Fortunately, we will not need the details of the particular orderings of product formulas.) We then apply our matrix concentration tools and go through a complicated combinatorial bound (which is much more involved than obtaining the 1-norm quantity $\normp{\vH}{(1),1}$ in~\cite{thy_trotter_error}).

\subsubsection{Uniform Smoothness, Matrix Martingales, and Hypercontractivity}\label{sec:maintext_unif_smooth}
\begin{figure}[t]
    \centering
    \includegraphics[width=0.7\textwidth]{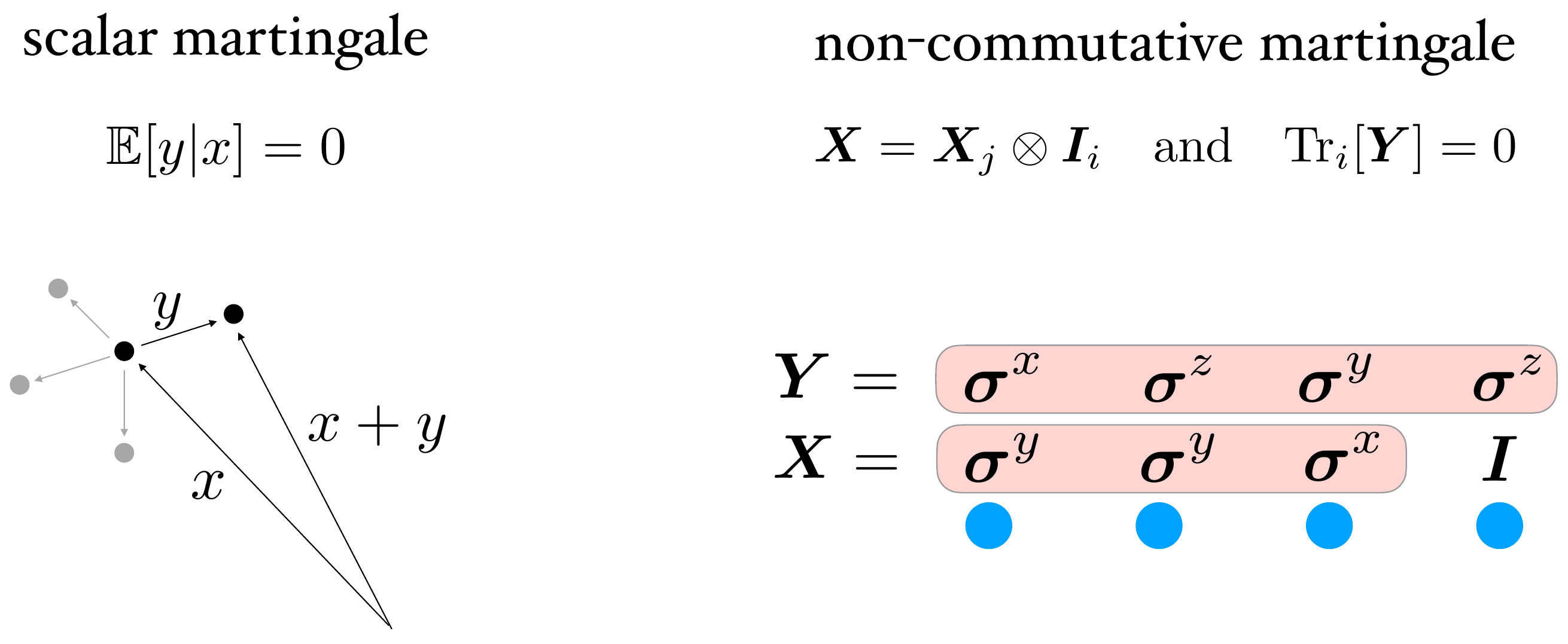}
    \caption{ For two scalar random variables satisfying the martingale condition, the variable $b$ has zero mean conditioned on variable $a$. In the non-commutative generalization, the matrix $\vB$ is partially traceless (``zero-mean'') on a subsystem where the matrix $\vA$ is trivial (``conditioned on $\vA$''). 
    }
    \label{fig:scalar_nc}
\end{figure}

To obtain quantitative control of complicated matrix functions, let us begin with an instructive example that captures the different perspectives. Consider a Hamiltonian as a sum of 1-local Pauli-Zs,
\begin{align}
    \vH = \vsigma^z_1 + \cdots +\vsigma^z_n, 
\end{align}
where each Pauli $\vsigma^z_i$ is supported on qubit $i$. How ``big'' is the sum?

(1) Take the spectral norm for the largest eigenvalue in magnitude
\begin{align}
    \lnorm{\vsigma^z_1 + \cdots +\vsigma^z_n} = n.
\end{align}

(2) Interpret the trace as an expectation, then its eigenvalue distribution is equivalent to a sum of independent random variables $ S_n:= x_1+\cdots+x_n$ each drawn from the Rademacher distribution $\Pr(x_i=1)=\Pr(x_i=-1)=1/2$. 
Now, we can use a \textit{concentration inequality} to describe how rarely the random variable deviates from its expectation
\begin{align}
\Pr( \lambda_i \ge \epsilon) \equiv   \Pr( S_n \ge \epsilon) \le \e^{-\epsilon^2/2n}  \  \ \ (\text{Hoeffding's inequality})\label{eq:hoeff}.
\end{align}
In other words, the \textit{typical} magnitude of eigenvalues $ \labs{\lambda}= \CO(\sqrt{n}) \ll n$ is much smaller than the largest eigenvalue. This simple example captures the overarching theme of this work: 
{  {\textit{Concentration is ubiquitous but often unspoken in the high dimensional setting.} }}.

To go beyond the above example, we rely on a family of recursive inequalities for their $p$-norms, which leads to concentration by Markov's inequality. We begin with reviewing the ancestral scalar version, often called the \textit{two-point inequality} or \textit{Bonami's inequality} (See, e.g.,~\cite{Garling2007InequalitiesAJ}).
\begin{fact}[The two-point inequality]\label{fact:two_point}
For real numbers $a, b$,
\begin{align} 
 \L(\frac{(a+b)^p+(a-b)^p}{2} \R)^{2/p} \le a^2+(p-1)b^2.
\end{align}
\end{fact}
This can be seen by expanding the binomial. This seemingly trivial inequality turns out to have far-reaching consequences, and its simplicity becomes its strength (See, e.g., Boolean analysis~\cite{odonnell2021analysis}).
The same form of inequality has an exact matrix analog, often called \textit{uniform smoothness}.
\begin{fact}[Uniform smoothness for Schatten Classes {\cite{Tomczak1974}}]\label{fact:unif_schatten} For matrices $\vX$ and $\vY$, 
\begin{align}
    \left[\frac{1}{2}(\lnormp{\vX+\vY}{p}^p+\lnormp{\vX-\vY}{p}^p)\right]^{2/p} \le \lnormp{\vX}{p}^2+ (p-1) \lnormp{\vY}{p}^2.
\end{align}
\end{fact}
The above form is not directly applicable, but its alternative forms with a martingale flavor streamline most of our proofs. For $k$-local operators (which are, in fact, closely related to \textit{non-commutative martingales}; see Figure~\ref{fig:scalar_nc}), we derive and make heavy usage of the following:
\begin{prop}[Uniform smoothness for subsystems]\label{prop:unif_subsystem_intro}
Consider matrices $\vX, \vY \in \CB(\CH_i\otimes\CH_j)$ that satisfy the non-commutative martingale condition $\tr_i(\vY) = 0$ and $\vX= \vX_j\otimes \vI_i$. For $p \ge 2$, 
\begin{equation}
\lVert \vX + \vY\rVert_{p}^2\le \lVert \vX \rVert_{p}^2  + (p-1)\lVert \vY\rVert_{p}^2.
\end{equation}
\end{prop}
In other words, uniform smoothness delivers \textit{sum-of-squares} behavior that contrasts with the triangle inequality, which is \textit{linear}
\begin{align}
    \norm{\vX+\vY} \le \norm{\vX}+\norm{\vY}.
\end{align}
This difference highlights the qualitative distinction between the worst case and the typical case, which is the starting point of all arguments in this work.

To illustrate its power, we apply to the 2-local operator (Figure~\ref{fig:XY}) 
\begin{align}
    \lnormp{\sum_{j<i}\vsigma^x_i\vsigma^y_j }{p}^2 &\le (p-1)\sum_{j} \lnormp{\sum_{i:j<i}\vsigma^x_i\vsigma^y_j }{p}^2 \\
    & \le (p-1)^2\sum_{j<i}\lnormp{ \vsigma^x_i\vsigma^y_j }{p}^2,
\end{align}
and more generally this gives concentration of $k$-local operators, or \textit{Hypercontractivity} (Section~\ref{sec:prelim_hyper}). 

\begin{figure}[t]
    \centering
    \includegraphics[width=0.3\textwidth]{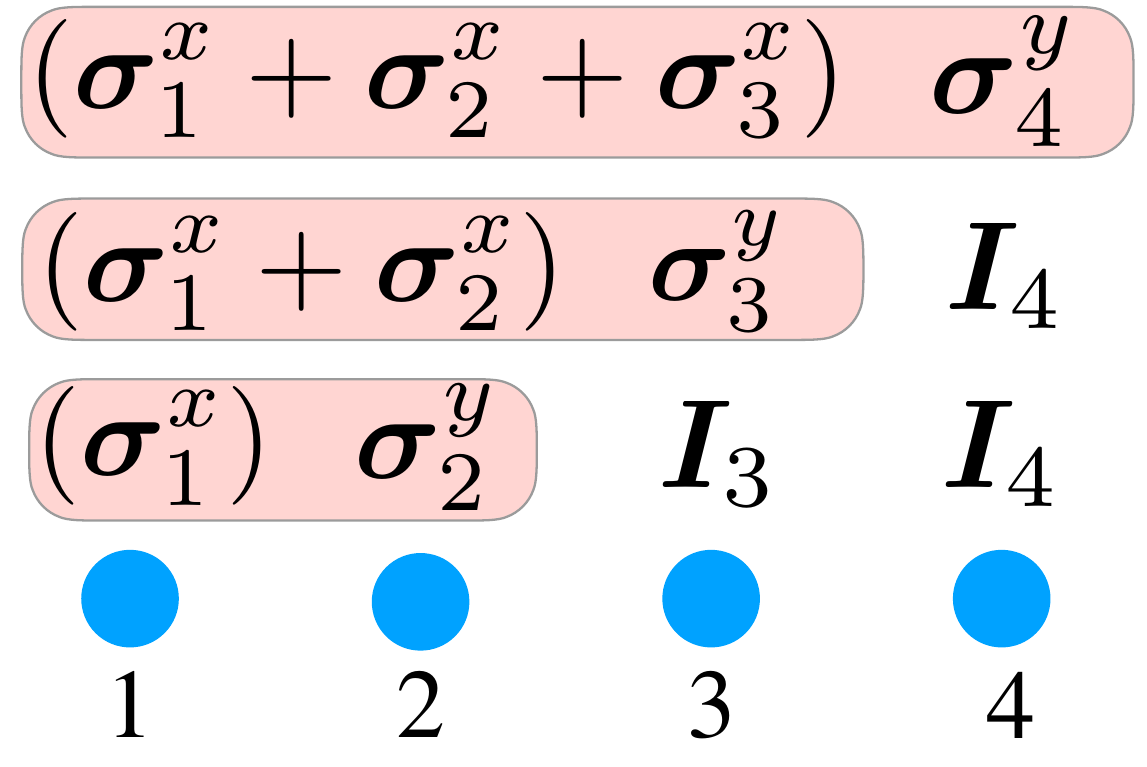}
    \caption{ Reorganizing the sum $\sum_{j<i}\vsigma^x_i \vsigma^y_j$ into a martingale, w.r.t. the index $j$. 
    }
    \label{fig:XY}
\end{figure}

For random Hamiltonians, the flavor of the problem changes slightly; we can think of adding Gaussian coefficients in our guiding example
\begin{align}
     \vH = g_1 \vsigma^z_1 + \cdots + g_n \vsigma^z_n.
\end{align}
The Gaussian coefficient (i.e., external randomness) requires the following version of uniform smoothness regarding the expected $p$-norm $\vertiii{\vX}_{p}:= (\BE[ \norm{\vX}_p^p ] )^{1/p}$ that will allow us to control the spectral norm, i.e., the worst input states. 
Initially, this featured in simple derivations of matrix concentration for martingales~\cite{naor_2012,HNTR20:Matrix-Product}.  
\begin{fact}[{Uniform smoothness for expected p-norm~\cite[Proposition~4.3]{HNTR20:Matrix-Product}}] \label{fact:sub_average_pq_maintext}
Consider random matrices $\vX, \vY$ of the same size that satisfy
$\BE[\vY|\vX] = 0$. When $2 \le p$,
\begin{equation}
\vertiii{\vX+\vY}_{p}^2 \le \vertiii{\vX}_{p}^2  + (p-1)\vertiii{\vY}_{p}^2 .
\end{equation}
\end{fact}
See Section~\ref{sec:martingales} for the relevant background and an alternative norm for arbitrary fixed input states.
Beyond the scope of this work, we emphasize these robust and straightforward martingale inequalities should find applications in versatile quantum information settings, whether by exploiting the tensor product structure of the Hilbert space or the randomness in matrix summands. See, e.g.,~\cite{chen2021optimal,chen2021concentration} for applications in operator growth and~\cite{chen2020quantum} in randomized quantum simulation.

\subsection{Discussion}
For many physical systems of interest (i.e., non-random $k$-local Hamiltonians), we present an average-case gate complexity that is qualitatively better than the worst-case. Without even changing the product formula, our analysis leads to a direct reduction of resources for quantum simulation applications. It is natural to hope that states appearing in practice (such as the ground state in quantum chemistry applications) behave like the typical states rather than the worst states, which would make product formulas very appealing for quantum simulation. It would be very interesting to carry out small-scale numerics. Our result holds with high probability for inputs drawn from any orthonormal basis.
Unfortunately, our current argument is probabilistic and does not label the exceptional states. Heuristically, the Trotter error is another $k$-local Hamiltonian that does not resemble the original Hamiltonian, so the states at low energy need not be ``aligned'' with the extremal states maximizing the Trotter error. We leave this as an open problem.

Another natural question is whether other quantum simulation methods (such as qubitization) or even other quantum algorithms enjoy an average-case speed-up. If true, it would greatly improve the feasibility of many quantum computing applications.

\subsection{Acknowledgments}
We thank Yuan Su and Mario Berta for helpful discussions and Joel Tropp and Andrew Lucas for related collaborations. CFC is especially thankful to Joel Tropp for introducing to him the subject of matrix concentration inequalities. We thank Sam McArdle for his helpful comments. We thank Ashley Milsted for setting up AWS EC2 computing for the numerics and AWS for computing credits. After this work was completed, we became aware of related work~\cite{Qi_2021_Hamiltonian_simulation_random} by Qi Zhao, You Zhou, Alexander F. Shaw, Tongyang Li, and Andrew M. Childs that also studied Hamiltonian simulation with random input states. We thank them for letting us know about their work. CFC is supported by the Caltech RA fellowship and the Eddleman Fellowship.

\section{Preliminary: $k$-Locality, Uniform Smoothness, and Hypercontractivity}\label{sec:prelim_hyper}
In quantum information, we often encounter a \textit{k-local} operator: its Pauli strings $\vF_{S}$ have lengths at most $k$. This is the quantum analog of a low-degree Boolean function~\cite{montanaro_q_boolean}. 
\begin{align}
    \vF &= \vI + \sum_{\labs{S}=1}\sum_{\alpha=1}^3 b^\alpha_{s_1} \vec{\sigma}^\alpha_{s_1}+ \sum_{\labs{S}=2} \sum_{\alpha=1}^3\sum_{\beta=1}^3 c^{\alpha\beta}_{s_1s_2} \vec{\sigma}^{\alpha}_{s_1}\vec{\sigma}^{\beta}_{s_2}+ \cdots
    := \sum_{\labs{S}\le k} \vF_S.
\end{align}

Given such a k-local operator, let us quantify its ``strength'' acting on states. One cautious choice is to worry about the worst possible state via the operator norm 
\begin{align}
    \norm{\vF}=\sup_{\psi}\lnormp{\vF\ket{\psi}}{\ell_2} \quad\text{(worst inputs)},
\end{align}
which maximizes the vector $\ell_2$-norm. 

In this work, we are instead interested in framing a ``typical case'' statement that might apply to states we encounter in practice. To be more precise, let us model the typical states by states drawn \textit{randomly} from a 1-design ensemble (e.g., an orthonormal basis). Practically, this can be a random computational basis state (which is more realistic than a Haar random state).

For such a random state $\ket{\psi}$, how large is the typical strength $\lnormp{\vF\ket{\psi}}{\ell_2}$? This question can be succinctly phrased in terms of a \textit{concentration inequality} that controls the probability of an undesirably large strength. 
\begin{prop}[Typical states and Schatten p-norms]\label{prop:typical_Schatten} For a pure state $\ket{\psi}$ drawn randomly from an ensemble 
$\BE_{\psi}[ \ket{\psi}\bra{\psi} ]= \vrho$, 
we have 
\begin{align}
    \Pr\L(  \lnormp{ \vF \ket{\psi}}{\ell_2}\ge \epsilon \R) \le \L(\frac{\lnormp{\vF\vrho^{1/p}}{p}}{\epsilon}\R)^p.\label{eq:pbar_norm_markov}
\end{align}
\end{prop}
In particular, for the maximally mixed state, we recover the normalized Schatten p-norm
\begin{align}
\lnormp{\vF\vrho^{1/p}}{p}= \lnormp{\vF}{\bar{p}}\quad\text{for}\quad \vrho = \vI/\tr[\vI].    
\end{align}
\begin{proof}[Proof of Proposition~\ref{prop:typical_Schatten}]
For illustration, we start with Chebychev's inequality and the variance $p=2$
\begin{align}
    \Pr\L(  \lnormp{ \vF \ket{\psi}}{\ell_2}\ge \epsilon \R) \le \BE_{\psi}\L[   \lnormp{\vF \ket{\psi}}{\ell_2}^2 \R] = \tr\L[ \vF \vrho\vF^\dagger  \R]. 
\end{align}
The last equality evaluates the expectation over states. To obtain sharper tail bounds via Markov's inequality, consider the p-th moment
\begin{align}
     \BE_{\psi}\L[   \lnormp{\vF \ket{\psi}}{\ell_2}^p \R] = 
     \BE_{\psi}\tr\L[\L( \vF\ket{\psi} \bra{\psi}\vF^\dagger \R)^{p/2} \R]  &\le \tr\L[ \L(\vF\vrho^{\frac{2}{p}}\vF^\dagger \R) ^{p/2}\R]. \label{eq:p-th_moment_of_strength}
\end{align}
The inequality applies a certain form of concavity (Fact~\ref{fact:poly_concavity}). \end{proof}

In other words, the (weighted) 2-norm coincides with the variance of the typical strength; the (weighted) p-norm governs the tail bounds. Indeed, the p-norms can be expressed in terms of the eigenvalues of the operator $\vF$;~\eqref{eq:p-th_moment_of_strength} will be an equality if we draw the state $\ket{\psi}$ from the eigenbasis of the operator $\vF$. Conveniently, with other choices of basis, the concavity tells us we still retain an inequality in~\eqref{eq:p-th_moment_of_strength}. The rest of our discussion boils down to estimating the $p$-norm of $k$-local operators. This is the content of \textit{Hypercontractivity}. 
\begin{fact}[Hypercontractivity for Paulis{~\cite[Theorem~46]{montanaro_q_boolean}}]\label{fact:nc_hyper}
For an operator acting on qubits $\vF\in \CB(\CH(2^n))$, $p \ge 2$, and $C_p:=p-1$,
\begin{align}
\lnormp{\vF}{\bar{p}}^2 \le \lnormp{\sum_S \sqrt{C_p}^{\labs{S}}\vF_{S}}{\bar{2}}^2 = \sum_S C_p^{\labs{S}} \lnormp{\vF_{S}}{\bar{2}}^2.
\end{align}
\end{fact}
Indeed, for $p=2$, it holds with equality, and the 2-norm gives an \textit{incoherent} sum over subsets $S$. (The Pauli strings are orthogonal in the Hilbert-Schmidt inner product.) We take squares to emphasize this sum-of-square nature. For general $p\ge 2$, Hypercontractivity gives an analogous incoherent sum over subsets but with enlarged coefficients $C_p^{\labs{S}}$. In other words, so long as the sets $S$ are small (e.g., the operator is $k$-local for a fixed $k$), we obtain all p-norm estimates from a 2-norm calculation. See also~\cite{Montanaro_application_hypercontract} for applications of Hypercontractivity, and note the equivalent inverted form also appears 
 \begin{align}
\lnormp{\sum_S \frac{1}{\sqrt{C_p}^{\labs{S}}} \vF_{S}}{\bar{p}}^2 \le \lnormp{\vF}{\bar{2}}^2.     
 \end{align}

Historically, this zoo of closely-related ideas starts from the Boolean cases (see, e.g.,~\cite{odonnell2021analysis} and Section~\ref{sec:scalar_hyper}) and extends to the non-commutative cases, including Paulis~\cite{montanaro_q_boolean}, Fermions~\cite{Carlen_hyper_fermions}, and abstract von Neumann algebras~\cite{ricardXu16}. In various contexts, Hypercontractivity has been constantly revisited and found applications in classical~\cite{Ben_Aroya_2008,Arunachalam21,odonnell2021analysis} and quantum computing~\cite{Montanaro_application_hypercontract}. The goal of our discussion here is to put together a coherent and accessible review that illustrates the rather common phenomena with some problem-driven adaptations. We begin with the intuitive qudit case (with the maximally mixed state as the ensemble) and later swiftly generalize to several settings arising from quantum simulation. 

To understand Hypercontractivity, our main approach  is the following recursive inequality called \textit{uniform smoothness}. 
\begin{prop}[Uniform smoothness for subsystems]\label{prop:unif_subsystem_recap}
Consider matrices $\vX, \vY \in \CB(\CH_i\otimes\CH_j)$ that satisfy
$\tr_i(\vY) = 0$ and $\vX= \vX_j\otimes \vI_i$. For $p \ge 2$,  $C_p=p-1$,
\begin{equation}
\lVert \vX + \vY\rVert_{p}^2\le \lVert \vX \rVert_{p}^2  + C_p\lVert \vY\rVert_{p}^2.
\end{equation}
\end{prop}
Technically, the partially traceless assumption $\tr_i(\vY) = 0$ makes it a \textit{non-commutative martingale} where taking partial trace is a \textit{conditional expectation}
\begin{align}
E_i[\cdot]:= \vI_i\otimes \btr_i[\cdot] \quad \text{where} \quad \btr_i[\cdot]: = \frac{\tr_i[\cdot]}{\tr_i[\vI_i]}.
\end{align}
We can rewrite Proposition~\ref{prop:unif_subsystem_recap} more formally \footnote{$\CB(\CH_j)\otimes \vI_i \subset  \CB(\CH_i\otimes\CH_j)$ is the filtration of von Neumann algebras. We will mostly stick to the usual qubit picture. See Section~\ref{sec:subalgebras} for an analogous result entirely in terms of von Neumann algebras.}. For any matrix $\vK$,
\begin{align}
    \lnormp{ \vK }{p}^2\le \lnormp{ E_i[\vK] }{p}^2  + C_p\lnormp{(1-E_i)[\vK] }{p}^2.
\end{align}
This martingale condition naturally appears in subroutines of quantum information applications, while Hypercontractivity as a black box is more ``rigid''. Although these two ideas are intimately related, we emphasize uniform smoothness is a versatile\footnote{In particular the above Hypercontractivity (Fact~\ref{fact:nc_hyper}) is restricted to qubits, (\cite{king2012Hypercontractivity} generalizes to other unital noise operators, on qubits), while uniform smoothness generalizes fairly easily. } and transparent driving horse, which implies, among other consequences, Hypercontractivity (Corollary~\ref{cor:unif_implies_hyper}).
 
 \subsection{Uniform Smoothness for Subsystems}
Proposition~\ref{prop:unif_subsystem_recap} is a special case of~\cite{ricardXu16}. Here, we present an elementary proof by adapting the argument in~\cite[Prop~4.3]{HNTR20:Matrix-Product}\footnote{This observation was made during this work and other work~\cite{chen2021optimal}. We include the proof in both.}. We start with the primitive form of uniform smoothness as a black box. 
\begin{fact}[Uniform smoothness for Schatten Classes, recap~\cite{Tomczak1974}]\label{fact:unif_schatten_recap} For matrices $\vX$ and $\vY$, 
\begin{align}
    \left[\frac{1}{2}\L(\lnormp{\vX+\vY}{p}^p+\lnormp{\vX-\vY}{p}^p\R)\right]^{2/p} \le \lnormp{\vX}{p}^2+ C_p \lnormp{\vY}{p}^2.
\end{align}
\end{fact}
We also need the following fact.
\begin{fact}[Monotonicity of $p$-norm w.r.t partial trace]\label{fact:nc_convexity}
For matrices $\vX$ and $\vY$ satisfying $\tr_i(\vY) = 0$ and $\vX= \vX_j\otimes \vI_i$, $p \ge 2$,
\begin{align}
\lnormp{\vX}{p}\le \lnormp{\vX+\vY}{p}.
\end{align}
\end{fact}
This can be understood as the non-commutative analog of convexity 
 $
     \lnormp{\vX+\BE_{\vY}\vY}{p} \le \BE_{\vY}\lnormp{\vX+\vY}{p}. 
 $
\begin{proof}[Proof of Fact~\ref{fact:nc_convexity}]
Recall the variational expression~\cite[Sec~12.2.1]{Wilde_QShannon} for Schatten p-norms 
\begin{align}
    \lnormp{\vX_j}{p} = \sup_{\lnormp{\vB_j}{q}\le 1} \tr(\vX_j\vB_j^\dagger)\quad \text{for} \quad 1/p+1/q=1.
\end{align}
Then,
\begin{align}
    \lnormp{\vX+\vY}{p} = \sup_{\lnormp{\vB}{q}\le 1} \tr[(\vX+\vY)\vB^\dagger] \ge \tr\L[ (\vX+\vY) \vB_j^\dagger\otimes \frac{\vI_i}{\lnormp{\vI_i}{q}}\R] = \lnormp{\vX_j\otimes \vI_i}{p}.
\end{align}

The last equality uses the partially traceless condition $\tr_i \vY =0$ and that 
\begin{align}
\vX_j\otimes \vI_i \quad \text{has maximizer} \quad\vB_j\otimes \vI_i/{\lnormp{\vI_i}{q}}.
\end{align}

An alternative proof is by averaging over Haar unitary on subsystem $i$
\begin{align}
\lnormp{\vX}{p}=\lnormp{\vX+\BE_{\vec{U}}[ \vec{U}_i\vY\vec{U}_i^\dagger]}{p}\le \BE_{\vec{U}_i}\lnormp{\vX+ \vec{U}_i\vY\vec{U}_i^\dagger}{p} = \lnormp{\vX+\vY}{p} .
\end{align}
The first equality is Schur's lemma, then convexity, and lastly, we used unitary invariance of p-norms. 
\end{proof}
We can almost prove Proposition~\ref{prop:unif_subsystem_recap}.
\begin{align}
     \frac{\lnormp{\vX+\vY}{p}^2 + \lnormp{\vX}{p}^2}{2} &\le \frac{\lnormp{\vX+\vY}{p}^2 + \lnormp{\vX-\vY}{p}^2}{2} \\
     &\le
     \left(\frac{\lnormp{\vX+\vY}{p}^p+\lnormp{\vX-\vY}{p}^p}{2} \right)^{2/p} 
     \le \lnormp{\vX}{p}^2+ C_p \lnormp{\vY}{p}^2.
\end{align}
The last inequality is Lyapunov's and then Fact~\ref{fact:unif_schatten}. Rearranging terms yields a slightly worse constant $2C_p$. The advertised constant can be obtained via another elementary but insightful trick~\cite[Lemma~A.1]{HNTR20:Matrix-Product}, which we reproduce as follows.
\begin{proof}[Proof of Proposition~\ref{prop:unif_subsystem_recap}]
The proof considers a rescaling argument. Let $\vZ:=\frac{1}{n}\vY$. We have just obtained 
\begin{align}
\lVert \vX + \vZ\rVert_{p}^2- \lVert \vX \rVert_{p}^2 \le  2C_p\lVert \vZ\rVert_{p}^2. \label{eq:2Cp}
\end{align}
Rearranging Fact~\ref{fact:unif_schatten}, 
\begin{align}
\lVert \vX + k\vZ\rVert_{p}^2-\lVert \vX - (k-1)\vZ\rVert_{p}^2&\le  \bigg(\lVert \vX + (k-1)\vZ\rVert_{p}^2-\lVert \vX + (k-2)\vZ\rVert_{p}^2 \bigg)+ 2C_p\lVert \vZ\rVert_{p}^2\\
&\le 2C_pk\lVert \vZ\rVert_{p}^2.
\end{align}
 The last inequality recursively applies the first line for $n\ge k\ge 2$ and \eqref{eq:2Cp} at base case\footnote{The quadratic rescaling argument inherits its uniform smoothness constant from Fact~\ref{fact:unif_schatten}, and the dependence of the base case constant vanishes in the limit. We just need some constant at the base. $\lVert \vX + \vZ\rVert_{p}^2- \lVert \vX \rVert_{p}^2 \le  f(p)\lVert \vZ\rVert_{p}^2.$} $k=1$.
Therefore, 
\begin{align}
    \lVert \vX + \vY\rVert_{p}^2 = \sum_{k} \bigg(\lVert \vX + k\vZ\rVert_{p}^2-\lVert \vX - (k-1)\vZ\rVert_{p}^2\bigg) &\le \sum_{k}2C_pk \lVert \vZ\rVert_{p}^2\\
    &=C_p\frac{n+1}{n} \lVert \vY\rVert_{p}^2.
\end{align}
Take $n\rightarrow\infty$ to obtain the sharp constant.
\end{proof}

\subsubsection{Subalgebras}\label{sec:subalgebras}
Let us work out the analogous elementary derivation for a subalgebra $\CN \subset \CM $, which captures non-commutative martingales in full generality. This also provides a unifying perspective for the manipulations we are doing.
For subalgebras $\CN \subset \CM \subset B(\CH)$, let $E:\CM \rightarrow \CN $ be the \textit{projection to subalgebra} $\CN$ (or \textit{the trace-preserving conditional expectation}), with the defining properties: 
\begin{align}
    E^\dagger[\vI]=\vI \quad \text{and}\quad  \tr[\vZ \vX] =   \tr[E[\vZ] \vX]\quad \text{for any} \quad \vX \in \CN, \vZ \in \CM.
\end{align}
Intuitively, $E$ is the analog of normalized partial trace $\vI_j \frac{\tr_j[\cdot ]}{\tr[\vI_j]}$. Using the notation natural in this setting, we reproduce the monotonicity.
\begin{fact}[Monotonicity of $p$-norm w.r.t projection to subalgebra]\label{fact:nc_convexity_subalgebra}
Consider finite dimensional subalgebras $\CN \subset \CM \subset B(\CH)$ and the corresponding projection to subalgebra $E:\CM \rightarrow \CN $.  Then, for any $\vZ \subset \CM$ and $p \ge 2$, 
\begin{align}
\lnormp{E[\vZ]}{p}\le \lnormp{\vZ}{p}.
\end{align}
\end{fact}
\begin{proof}[Proof of Fact~\ref{fact:nc_convexity_subalgebra}]
Again, consider the variational expression
\begin{align}
    \lnormp{\vX}{p} = \sup_{\lnormp{\vB_j}{q}\le 1, \vB \in \CN } \tr(\vX\vB^\dagger) \quad\text{where}\quad 1/p+1/q=1.
\end{align}
Note that the maximum is attained in the same algebra $\vB \in \CN$ (This can be seen by the structure theorem of finite-dimensional von Neumann algebra. $\CN$ is a direct sum of subsystems.). 
Then 
\begin{align}
    \lnormp{\vZ}{p} = \sup_{\lnormp{\vB'}{q}\le 1} \tr[\vZ\vB'^\dagger ] \ge \tr\L[ \vZ \vB^\dagger \R] =\tr\L[ E[\vZ] \vB^\dagger \R] = \lnormp{E[\vZ]}{p},
\end{align}
which is the advertised result.
\end{proof}
Through the same arguments, we conclude the discussion for subalgebras by the following.
\begin{prop}[Uniform smoothness for subalgebras]\label{prop:unif_subalgebra}
Consider finite dimensional subalgebras $\CN \subset \CM \subset \CB(\CH)$ and the corresponding projection to subalgebra $E:\CM \rightarrow \CN $.  Then, for any $\vZ\in \CM$,
\begin{align}
\lnormp{\vZ}{p}^2 \le \lnormp{E[\vZ]}{p}^2 +C_p \lnormp{(1-E)[\vZ]}{p}^2.
\end{align}
\end{prop}
This result was first obtained in~\cite{ricardXu16} in a more technical setting. We hope the discussion here provides a simple interpretation.
\subsection{Deriving Hypercontractivity}
Uniform smoothness, through a recursion, implies Hypercontractivity-like global formulas. \footnote{After this work was presented in QIP 2022, we thank Haonan Zhang for pointing to us the work~\cite{ricardXu16} that made an analogous observation.}
\begin{prop}[Moment estimates for local operator]\label{prop:general_pauli_expansion}
For an operator $\vF \in \CB(\BC^{d\otimes n})$ on $n$-qudits, and $p\ge 2$, $C_p := p-1$,
\begin{align}
\lnormp{\vF}{p}^2 &\le \sum_{S\subset \{m,\cdots,1\}} (C_p)^{|S|} \lnormp{\vF_S}{p}^2 \quad \text{where} \quad \vF_{S}:=\prod_{s\in S}(1-E_s)\prod_{s'\in S^c} E_{s'}[\vF].
\end{align}
The super-operator $E_s[\cdot]:=\vI_{s}\otimes \btr_s[\cdot] $ is the conditional expectation associated with the partial trace, and the set $S^c$ is the complement of set $S$. 
\end{prop}

Intuitively, for each subset $S$, the product of conditional expectations selects the component $\vF_S$ that is non-trivial on set $S$ and trivial on the complement $S^c$. Indeed, for qubits, the summand $\vF_S$ corresponds to the Pauli strings non-trivial on set $S$. 

Let us grasp this formula with some examples. For single-site Paulis, this resembles the usual concentration inequality for bounded independent summand (e.g., Hoeffdings' inequality).
\begin{eg}[1-local Paulis]
\normalfont For $\vF:=\sum_i \alpha_i\vec{\sigma}^z_i$,
\begin{align}
    \lnormp{\vF}{\bar{p}}^2 \le C_p \sum_i \alpha_i^2\lnormp{ \vec{\sigma}^z_i }{\bar{p}}^2=(p-1) \sum_i \labs{\alpha_i}^2\norm{ \vec{\sigma}^z_i }. 
\end{align}
By Markov's inequality, we obtain concentration for pure states drawn from a fixed orthonormal basis $\BE[\ket{\psi}\bra{\psi}]=\vI/\tr[\vI]$.
\begin{align}
    \Pr\L(\lnormp{\vH\ket{\psi}}{\ell_2}\ge \epsilon\R) \le (\frac{\normp{\vH}{\bar{p}} }{\epsilon})^{p} \le \L(\frac{\sqrt{(p-1)\sum_i \labs{\alpha_i}^2}}{\epsilon}\R)^p \le \e\cdot\exp \L(-\frac{\epsilon^2}{\e\sum_i \labs{\alpha_i}^2}\R).
\end{align}
In other words, the strength is typically bounded by the variance $\sum_i \labs{\alpha_i}^2$. If we take the states to be the computational basis, the tail bound applies to its eigenvalue distribution.
\end{eg}
Moreover, we obtain a similar sum-of-squares behavior for $4$-local Paulis, albeit with heavier tails.
\begin{eg}[4-local Paulis] \normalfont For $\vF:=\sum_{i>j>k>\ell} \alpha_{ijk\ell}\vec{\sigma}^x_i\vec{\sigma}^z_j \vec{\sigma}^y_k\vec{\sigma}^x_\ell$, 
\begin{align}
    \lnormp{\vF}{\bar{p}}^2 \le (C_p)^4 \sum_{i>j>k>\ell} \labs{\alpha_{ijk\ell}}^2\lnorm{\vec{\sigma}^x_i\vec{\sigma}^z_j \vec{\sigma}^y_k\vec{\sigma}^x_\ell}^2 
\end{align}
By Markov's inequality, we obtain 
\begin{align}
    \Pr\L(\lnormp{\vH\ket{\psi}}{\ell_2}\ge \epsilon\R) \le \left(\frac{\sqrt{(p-1)^4\sum_{i>j>k>\ell} \labs{\alpha_{ijk\ell}}^2}}{\epsilon} \right)^p \le \exp \left(-\sqrt{\frac{\epsilon}{\e \sqrt{\sum_{i>j>k>\ell} \labs{\alpha_{ijk\ell}}^2}}} \right)
\end{align}
which does not have a Gaussian tail anymore but still decays super-polynomially. 
\end{eg}
Let us now present the elementary proof.
\begin{proof}[Proof of Proposition~\ref{prop:general_pauli_expansion}]
Apply uniform smoothness (Proposition~\ref{prop:unif_subsystem_recap}) for $s=1,\cdots, n$ recursively. 
\begin{align}
\lnormp{\vF}{p}^2 &= \lnormp{\prod_{s=1}^n\L((1-E_s)+E_s \R)[\vF]}{p}^2\\ 
&\le \sum_{S\subset \{n,\cdots,1\}} (C_p)^{|S|} \lnormp{\prod_{s\in S}(1-E_s)\prod_{s'\in S^c} E_{s'}[\vF]}{p}^2
\end{align}
Each application produces two branches, and the total $2^n$ branches are labeled by the subsets $S\subset \{n,\cdots, 1\}$. We may regroup the conditional expectations since they are just taking partial traces of disjoint subsystems. 
The power $(C_p)^{|S|}$ comes from the times the branch $(1-E_s)$ appears.
\end{proof}
To compare with the existing Hypercontractivity for qubits, it is worth bringing Proposition~\ref{prop:general_pauli_expansion} to the following form.
\begin{cor}[Non-commutative Hypercontractivity]\label{cor:unif_implies_hyper} In the setting of Proposition~\ref{prop:general_pauli_expansion},
\begin{align}
\lnormp{\vF}{\bar{p}}^2 &\le \sum_{S\subset \{m,\cdots,1\}} (3C_p)^{|S|} \lnormp{\vF_S}{\bar{2}}^2 = \lnormp{\sqrt{3C_p}^{\labs{S}}\vF_{S}}{\bar{2}}. 
\end{align}
\end{cor}
This is equivalent to the existing bound (Fact~\ref{fact:nc_hyper}) up to slightly worse constants. However, the martingale formulation streamlines a simple proof (Proposition~\ref{prop:unif_subsystem_recap}) and, more importantly, allows us to adapt to different settings in the subsequent sections. 
\begin{proof}[Proof of Corollary~\ref{cor:unif_implies_hyper}]
Bound the normalized p-norm $\lnormp{\vF}{\bar{p}}:=\frac{\lnormp{\vF}{p}}{\lnormp{\vI}{p}}$ by Pauli expansion $\sigma_S:=\{\vsigma^x,\vsigma^y,\vsigma^z\}^{\labs{S}}$
\begin{align}
      \lnormp{\vF_S}{\bar{p}}^2\le \left(\sum_{\sigma_S} \lnormp{\vF_{\sigma_S}}{\bar{p}} \right)^2=  \left(\sum_{\sigma_S} \lnormp{\vF_{\sigma_S}}{\bar{2}} \right)^2 \le (\sum_{\sigma_S} )\cdot \sum_{\sigma_S} \lnormp{\vF_{\sigma_S}}{\bar{2}}^2 = 3^{\labs{S}}\cdot \lnormp{\vF_{S}}{\bar{2}}^2,
\end{align}
which is the advertised result. Intuitively, the extra factor we pay is the number of distinct Pauli strings $3^{\labs{S}}$.
\end{proof}

\subsection{Product Background States}\label{sec:product_bg}
Our previous discussion focused on the unweighted p-norm. In this section, we discuss the \textit{weighted p-norms}. For $0\le s\le 1$, define
\begin{align}
    \lnormp{\vF}{p,\vrho,s}&: =\lnormp{\vrho^{\frac{1-s}{p} }\vF\vrho^{\frac{s}{p}}}{p},
\end{align}
where $s=1/2$~\cite{Beigi_2013} and $s=1$ are the notable cases
\begin{align}
    \lnormp{\vF}{p,\vrho,\frac{1}{2}}&: =\lnormp{\vrho^{\frac{1}{2p}}\vF\vrho^{\frac{1}{2p}}}{p}\quad \text{and}\quad 
    \lnormp{\vF}{p,\vrho,1}: =\lnormp{\vF\vrho^{\frac{1}{p}}}{p}.
\end{align}
The latter feeds into the concentration for typical input states drawn from an ensemble whose average is $\vrho$ (Proposition~\ref{prop:typical_Schatten}). Even though not applied elsewhere in this paper, we keep the general $0\le s\le 1$ expression in the following arguments. Uniform smoothness generalizes to the $\vrho$-weighted $p$-norm for factorized state $\vrho=\otimes_i \vrho_i$. The martingale condition now depends on the state $\vrho_i$.
\begin{prop}[Uniform smoothness for subsystem, weighted]\label{prop:unif_subsystem_weighted}
Consider product state $\vrho= \vrho_j\otimes \vrho_i$ and matrices $\vX, \vY \in \CB(\CH_i\otimes\CH_j)$ that satisfy the martingale condition
$\tr_i(\vrho_i\vY) = 0; \vX= \vX_j\otimes \vI_i$. For $p \ge 2$, $C_p=p-1$,
\begin{equation}
\lVert \vX + \vY\rVert_{p,\vrho,s}^2\le \lVert \vX \rVert_{p,\vrho,s}^2  + C_p\lVert \vY\rVert_{p,\vrho,s}^2.
\end{equation}
\end{prop}

In a similar proof, all we need is to modify monotonicity.
\begin{fact}[monotonicity w.r.t partial trace]\label{fact:nc_convexity_weighted}
For matrices $\vX$ and $\vY$ satisfying $\tr_i(\vrho_i\vY) = 0$ and $\vX= \vX_j\otimes \vI_i$, $p \ge 2$, 
\begin{align}
\lnormp{\vX}{p,\vrho,s}\le \lnormp{\vX+\vY}{p,\vrho,s}.
\end{align}
\end{fact}
\begin{proof}
Once again, plug in the variational expression
\begin{align}
    \lnormp{\vrho^{\frac{1-s}{2p}}_j\vX_j \vrho^{\frac{s}{2p}}_j}{p} = \sup_{\lnormp{\vB_j}{q}\le 1} \tr\L[\vrho^{\frac{1-s}{2p}}_j\vX_j \vrho^{\frac{s}{2p}}_j \vB^\dagger_j\R] \quad \text{for} \quad 1/p+1/q=1.
\end{align}
Suppose the maximum is attained at some $\vB_j$. Then by Proposition~\ref{fact:nc_convexity},
\begin{align}
    \lnormp{\vX+\vY}{p,\vrho,s} &= \sup_{\lnormp{\vB}{q}\le 1} \tr\L(\vrho^{\frac{1-s}{2p}}_j(\vX_j+\vY_j) \vrho^{\frac{s}{2p}}_j \vB^\dagger\R) \\
    &\ge \tr\L[ \vrho^{\frac{1-s}{2p}}_j(\vX_j+\vY_j) \vrho^{\frac{s}{2p}}_j\cdot \vB^\dagger_j\otimes \frac{\vrho_i^{\frac{1}{q}}}{\normp{\vrho_i^{\frac{1}{q}}}{q}}\R] = \lnormp{\vX_j\otimes \vI_i}{p,\vrho,s}.
\end{align}
In the last inequality, we used the partially traceless assumption $\tr_i [\vrho_i\vY] =0$ and that 
\begin{align}
 \vrho^{\frac{1-s}{2p}} (\vX_j\otimes \vI_i )\vrho^{\frac{s}{2p}} \quad \text{has maximizer} \quad   \vB_j\otimes \frac{\vrho_i^{{1}/{q}}}{\normp{\vrho_i^{{1}/{q}}}{q}}.
\end{align}
\end{proof}
Combining the monotonicity with Fact~\ref{fact:unif_schatten}, we obtain uniform smoothness (Proposition~\ref{prop:unif_subsystem_weighted}). Moreover, we automatically get a weighted version of a Hypercontractivity-like formula. 
Let us first define the appropriate operator re-centered w.r.t. the background 
\begin{align}
    \vO^\eta:= (1-\eta)\ket{1}\bra{1}-\eta\ket{0}\bra{0} \quad \text{such}\quad \tr[\vrho_i\vO_i^\eta]=0
\end{align}
as a ``shifted'' Pauli $\vsigma^{z}$. Accordingly, we shift the conditional expectation
\begin{align}
    E_s[\cdot] := \vI_s \otimes \tr[\vrho_s \cdot]. 
\end{align}

\begin{prop}[Moment estimates for local operator, $\vrho$-weighted] \label{prop:general_pauli_expansion_weighted}
For an operator $\vF \in \CB(\BC^{d\otimes n})$ on $n$-qudits, product state $\vrho = \otimes_i \vrho_i$, and $p\ge 2$, $C_p := p-1$,
\begin{align}
\lnormp{\vF}{p,\vrho,s}^2 &\le \sum_{S\subset \{m,\cdots,1\}} (C_p)^{|S|} \lnormp{\vF_S}{p,\vrho,s}^2 \quad \text{where} \quad \vF_{S}:=\prod_{s\in S}(1-E_s)\prod_{s'\in S^c} E_{s'}[\vF].
\end{align}
The set $S^c$ is the complement of set $S$.
\end{prop}

\subsubsection{Low-particle number subspace}
Why did we study product state as the background? Interestingly, it will tell us about concentration when restricting to \textit{low particle number subspaces}.  Consider the following two operators: the projector $\vP_m$ to the $m$-particle subspace of $n$-qubit Hilbert space 
\begin{align}
    \vP_m :=\sum_{\#(\ket{1}) =m} \L( \ket{0}\cdots \ket{1}\R) \L( \cdot\R)^\dagger = \sum_{\#(\ket{1}) =m} \ket{0}\bra{0}\cdots \ket{1}\bra{1}
\end{align}
and the product state
\begin{align}
    \vrho_{\eta}=\otimes_i \vrho_{i}=\bigotimes_i \bigg(\eta\ket{1}\bra{1}+(1-\eta)\ket{0}\bra{0}\bigg)_i\quad \text{where}\quad \eta:= \frac{m}{n}.
\end{align} 
We can control the p-norm with the low-particle subspace, which we care about, with the product state, which we can calculate.
\begin{prop}\label{prop:projector_to_product}For any operator $\vF$,
\begin{align}
   \displaystyle \lnormp{\vF}{p,\bar{\vP},s }  \le 
   \lnormp{\vF}{p,\vrho_{\eta},s} \cdot\left(\poly(n,m)\right)^{1/p}.
\end{align}
\end{prop}
Note the factor $\poly(n,m)$ is mild since they are suppressed as long as $p\gtrsim \log(\poly(n,m))$. 
\begin{proof}[Proof of Proposition~\ref{prop:projector_to_product}]
By Stirling's approximation, the operators obey positive semi-definite order
\begin{align}
    \bar{\vP_m}:= \frac{\vP_m}{\tr[\vP_m]} \le \vrho_{\eta}\cdot b(n,m)\quad \text{where} \quad b(n,m) = (\eta)^m(1-\eta)^{n-m}\binom{n}{m}=\CO(\poly(n,m)).
\end{align}
This gives the advertised result by Fact~\ref{fact:mono_weight}.
\end{proof}
Fact~\ref{fact:mono_weight}, proved below, is that weighted norms are monotone w.r.t the state. In our application for Trotter error, the Hamiltonian is often particle number preserving, and the following becomes trivial. But for potential applications in other contexts, we include a quick proof when the operator $\vF$ and state $\vrho$ are not commuting.

\begin{fact}[monotonicity of weight]\label{fact:mono_weight}
For positive semi-definite operators $\vrho \ge \vsigma \ge 0$ (presumably not normalized),
\begin{align}
    \lnormp{\vO}{p,\vrho,s}\ge \lnormp{\vO}{p,\vsigma,s}. 
\end{align}
\end{fact}
This is closely related to a polynomial version of Lieb's concavity. 
\begin{fact}[{\cite[Theorem~1.1]{Carlen_2008}} ] \label{fact:poly_concavity}
For operator $\vA\ge 0$, and $q\ge 1$, $r\le 1$, the function
\begin{align}
f(\vA):=\tr[(\vB^\dagger \vA^{\frac{1}{q}}\vB)^{rq}]    
\end{align} 
is concave (and hence monotone) in $\vA$.
\end{fact}
We can now quickly adapt to our settings to present a proof.
\begin{proof}[Proof of Fact~\ref{fact:mono_weight}]
\begin{align}
    \lnormp{\vO}{p,\vrho}^p = \tr\L[\left(\vrho^{\frac{s}{2p}}\vO^\dagger \vrho^{\frac{1-s}{p}}\vO \vrho^{\frac{s}{2p}}\right)^{\frac{p}{2}}\R] &\ge \tr\L[\left(\vrho^{\frac{s}{2p}}\vO^\dagger \vsigma^{\frac{1-s}{p}}\vO \vrho^{\frac{s}{2p}}\right)^{\frac{p}{2}}\R]\\
    &=\tr\L[\left(\vsigma^{\frac{1-s}{2p}}\vO\vrho^{\frac{s}{p}}\vO^\dagger \vsigma^{\frac{1-s}{2p}}\right)^{\frac{p}{2}}\R] \ge\tr\L[\left(\vsigma^{\frac{1-s}{2p}}\vO\vsigma^{\frac{s}{p}}\vO^\dagger \vsigma^{\frac{1-s}{2p}}\right)^{\frac{p}{2}}\R] = \lnormp{\vO}{p,\vsigma,s}^p\notag. 
\end{align}
Both inequalities use Fact~\ref{fact:poly_concavity} for $q=\frac{p}{1-s}\ge 1$ and $r = \frac{1-s}{2} \le 1$ and for $q=\frac{p}{s}\ge 1$ and $r = \frac{s}{2} \le 1$. The second equality is $\lnormp{\vX^\dagger \vX}{\frac{p}{2}}= \lnormp{\vX \vX^\dagger}{\frac{p}{2}}$. This is the advertised result.
\end{proof}

\subsection{Fermionic Operators}
Uniform smoothness and Hypercontractivity apply to Fermions. Consider the Jordan-Wigner transform
\begin{align}
    \va_s &:= -  \prod_{j=1}^{s-1} \vI_j \cdot \vec{\sigma}^-_s\cdot \prod_{i=s+1}^{n}\vec{\sigma}^z_i\quad \text{where}\quad \vec{\sigma}^-= \ket{0}\bra{1},\\
    \va^\dagger_s &:= -  \prod_{j=1}^{s-1} \vI_j \cdot \vec{\sigma}^+_s\cdot \prod_{i=s+1}^{n}\vec{\sigma}^z_i\quad \text{where}\quad \vec{\sigma}^+= \ket{1}\bra{0}.
\end{align}
These operators also linearly span the full algebra on $n$-qubits $\CB(\CH(2^n))$ by products $\prod_s(\va_s,\va^\dagger_s,\va_s\va^\dagger_s,\vI_s)$. In this form, Fermions are not local operators due to the Pauli-Z strings. Fortunately, all we need for uniform smoothness is the martingale property (conditionally zero-mean). We derive an analogous 2-norm-like bound with a minor tweak due to Jordan-Wigner strings. The following result was known in~\cite[Theorem~4]{Carlen_hyper_fermions}\footnote{It uses the primitive uniform smoothness (Fact~\ref{fact:unif_schatten}).} but we hope the presented derivation is more transparent. We will also extend it in Corollary~\ref{cor:a_z_uniform}.
\begin{cor}[Hypercontractivity for Fermions]\label{cor:a_only_uniform} On $n$-qubits, consider an operator without terms $\va_i\va^\dagger_i$.  Expand it $\vA = \sum_{S\subset \{n,\cdots,1\}} \vA_{S}$ by subsets $S$ indicated by Fermionic operators $\{\va^\dagger,\va\}$. Then, for $p\ge 2$, $C_p=p-1$,
\begin{align}
\lnormp{\vA}{p}^2 &\le \sum_{S\subset \{m,\cdots,1\}} (C_p)^{|S|} \lnormp{\vA_S}{p}^2.
\end{align}
\end{cor}
\begin{proof}
 WLG, assume the Fermionic operators are ordered such that the larger index appears on the right (e.g.$\va_1\va_3\va_n$).
\begin{align}
    \lnormp{\vA}{p}^2= \lnormp{\va_1\vA_{>1}+\va^\dagger_1\vA'_{>1}+\vI_1\otimes \vB_{>1}}{p}^2 &\le \lnormp{\vI_1\otimes \vB_{>1}}{p}^2 +C_p\lnormp{\va_1\vA_{>1}+\va^\dagger_1\vA'_{>1}}{p}^2.
\end{align}
To complete the induction as in~\ref{prop:general_pauli_expansion}, apply a gauge transformation to change the Jordan-Winger string such that only $a_2$ is nontrivial on site $2$. Then we can repeat the above inequality. Note that the background $\vrho_{\eta}$ is invariant under gauge transformations, and the Pauli strings of $\vsigma^z$ do not blow up the weighted p-norm.
\end{proof}
\begin{eg}[2-local Fermionic operators]
\normalfont
\begin{align}
    \lnormp{\sum_{i<j} \alpha_{ij} \va_j\va_i}{p}^2 &\le \sum_{i<j}(C_p)^{2} \labs{\alpha_{ij}}^2\lnormp{\va_j\va_i}{p}^2\le \sum_{i<j}(C_p)^{2} \labs{\alpha_{ij}}^2\norm{\va_i\va_j}\lnormp{\vI}{p}^2
\end{align}
\end{eg}

However, when multiplying Fermion operators we may get even powers $\va^\dagger_i \va_i= (\vI+\vsigma_i)/2, \va_i\va^\dagger_i = (\vI-\vsigma_i)/2$ where the Pauli string $\vsigma^z$ cancels. Let us quickly extend to the cases with the presence of $\vsigma^z_i$ terms (perhaps with weighted background). Let us formally define the conditional expectation
\begin{align}
    &E_s: \CB(\{\va^\dagger_i,\va_i\}_{i=1,\cdots,n}) \rightarrow \CB(\{\va^\dagger_i,\va_i\}_{i=1,\cdots,n, i\ne s})=:\CN\\
    &\text{such that}\quad E_s[\vO_{-s}\va^\dagger_s] = E_s[\vO_{-s}\va_s] = E_s[\vO_{-s}\vO^{\eta}] = 0\quad \text{and}\quad E_s[\vO_{-s}] = \vO_{-s} \quad \text{for}\quad \vO_{-s} \in \CN.
\end{align}
The conditional expectation maps the full algebra to the subalgebra generated by all but one fermions. Intuitively, it removes terms that have non-trivial terms $\va_s,\va_s^\dagger, \vO^{\eta}$ on site $s$. 
\begin{cor}[Hypercontractivity for Fermions and $\vO^{\eta}$] \label{cor:a_z_uniform} On $n$ qubits, consider a product state diagonal in the computational basis $\vrho_{\eta}=\otimes_i \vrho_{i}=\otimes_i (\eta\ket{1}\bra{1}+(1-\eta)\ket{0}\bra{0})_i$. 
Then, for $p\ge 2$, $C_p =p-1$,
\begin{align}
\lnormp{\vA}{p,\vrho_\eta}^2 
&\le \sum_{ S \subset \{m,\cdots,1\}} (C_p)^{\labs{S}} \lnormp{\vA_S}{p,\vrho_\eta}^2 \quad \text{where}\quad \vA_{S}:=\prod_{s\in S}(1-E_s)\prod_{s'\in S^c} E_{s'}[\vA].
\end{align}
\end{cor}

The proof is also elementary.
\begin{proof} 
\begin{align}
    \lnormp{\vA}{p,\vrho_\eta}^2&= \lnormp{\va_1\vA_{>1}+\va^\dagger_1\vA'_{>1}+\vO^\eta_1\otimes \vC_{>1} +\vI_1\otimes \vB_{>1}}{p,\vrho_\eta}^2 \\
    &\le \lnormp{\vI_1\otimes \vB_{>1}}{p,\vrho_\eta}^2 +C_p\lnormp{\va_1\vA_{>1}+\va^\dagger_1\vA'_{>1}+\vO^\eta_1\otimes \vC_{>1}}{p,\vrho_\eta}^2.
\end{align}
The rest gauge transformation argument follows Corollary~\ref{cor:a_only_uniform}. Note that $\vO^{\eta}$ is invariant under gauge transformations. 
Alternatively, we can take the formal route by manipulating the conditional expectations as in Proposition~\ref{prop:general_pauli_expansion}.  
\end{proof}

\section{Non-Random $k$-Local Hamiltonians}
\label{chap:non_random}

This section presents the main result of this work. We evaluate Hypercontractivity (Section~\ref{sec:prelim_hyper}) for Trotter error of non-random Hamiltonians. 
\begin{thm}[Trotter error in $k$-local models] \label{thm:Trotter_non_random}
To simulate a $k$-local Hamiltonian using $\ell$-th order Suzuki formula, the gate complexity 
\begin{align}
G =\Omega\L( \left(\frac{p^{\frac{k}{2}} \lnormp{\vH}{(global),2} t}{\epsilon} \right)^{1/\ell} \Gamma p^{\frac{k-1}{2}} \lnormp{\vH}{(local),2} t\R) \ \ &\textrm{ensures}\ \  \lnormp{\e^{\iunit \vH t}- \vec{S}_{\ell}(t/r)^r}{\bar{p}} \le \epsilon.
\end{align} 
\end{thm}
The $p$-norm estimate and Proposition~\ref{prop:typical_Schatten} imply concentration for typical input states via Markov's inequality.
\begin{cor} 
Draw $\ket{\psi}$ from an orthonormal basis $\BE[ \ket{\psi}\bra{\psi}] =\vI/\tr[\vI]$, then 
\begin{align*}
    G =\Omega \L( \left(\frac{\sqrt{\log(1/\delta)}^{k}\lnormp{\vH}{(global),2} t}{\epsilon} \right)^{1/\ell} \sqrt{\log(1/\delta)}^{k-1} \Gamma \lnormp{\vH}{(local),2} t \R) \ \ &\textrm{ensures}\ \  \Pr \L( \lnormp{\e^{\iunit \vH t}- \vec{S}(t/r)^r\ket{\psi}}{\ell_2}  \ge \epsilon \R) \le \delta. 
\end{align*} 
\end{cor}
This quickly converts to the trace distance between the pure states 
\begin{align}
    \lnormp{\ket{\psi_1}\bra{\psi_1}-\ket{\psi_2}\bra{\psi_2}}{1} &\le \lnormp{\ket{\psi_1}\big(\bra{\psi_1}-\bra{\psi_2}\big)+\big(\ket{\psi_1}-\ket{\psi_2}\big)\bra{\psi_2}}{1}\\
    & \le 2\normp{\ket{\psi_1}-\ket{\psi_2}}{\ell_2}.
\end{align}
We begin with an instructive example that illustrates the combinatorics (Section~\ref{sec:heuristic_first_order}). We sketch the proof in Section~\ref{sec:sketch}. In Section~\ref{sec:nonrandom_g_th} and Section~\ref{sec:g'_th_non_random}, we combine the estimates and conclude the proof with explicit constants in Section~\ref{sec:proof_non_random}. See Section~\ref{sec:proof_fermion_non_random} for the analogous result for Fermions.

\subsection{An instructive example}\label{sec:heuristic_first_order}
Consider a 2-local Hamiltonian on three subsystems of qubits $\CH = \CH_{I_1}\otimes \CH_{I_2}\otimes \CH_{I_3}$ of equal subsystem sizes $n/3$.
\begin{align}
    \vH = \sum_{\gamma = 1}^{\Gamma} \vH_{\gamma} = \sum_{i_1\in I_1, i_2\in I_2} \vsigma^x_{i_1}\vsigma^z_{i_2} + \sum_{i_1\in I_1, i_2\in I_2} \vsigma^x_{i_1}\vsigma^x_{i_2} + \sum_{i_2\in I_2, i_3\in I_3} \vsigma^x_{i_2}\vsigma^z_{i_3}.
\end{align}
Let us play around with the first-order product formula. Recall
\begin{align}
     \e^{\ri\vH_\Gamma t} \cdots \e^{\ri\vH_1}- \e^{\iunit \sum_{\gamma=1}^\Gamma \vH_\gamma} = \frac{t^2}{2}\sum_{\gamma_2>\gamma_1}[\vH_{\gamma_2},\vH_{\gamma_1}]+\CO(t^3). 
\end{align}
The leading order $\CO(t^2)$ Trotter error will be a sum of 3-local terms and 1-local terms
\begin{align}
    \sum_{\gamma_2>\gamma_1}[\vH_{\gamma_2},\vH_{\gamma_1}]& = \sum_{\labs{S_3}=3}\vF_{S_3} + \sum_{\labs{S_1}=1}\vF_{S_1}.
\end{align}
The 3-local terms are the ``greediest'' way to produce long Pauli strings
\begin{align}
    \sum_{\labs{S_3}=3} \vF_{S_3} = - 2\sum_{i_1,i_2,i_3}\vsigma^x_{i_1}\vsigma^y_{i_2}\vsigma^z_{i_3} + 2\sum_{i_1,i'_1,i_2} \vsigma^x_{i_1}\vsigma^x_{i'_1}\vsigma^y_{i_2} \quad \text{and}\quad  \lnormp{\sum_{\labs{S_3}=3}\vF_{S_3}}{\bar{2}} = \theta( n^{3/2}). \quad \text{(greedy)}
\end{align}
No cancellation nor collision occurs, and each term is supported on distinct subsets $\{i_1,i_2,i_3\}$ or $\{i_1,i_1',i_2\}$. These operators add incoherently (in the Hilbert-Schmidt norm for simplicity). 
The 1-local terms are more peculiar but turn out equally important. They come from terms that overlap on both sites 
\begin{align}
    \sum_{\labs{S_1}=1}\vF_{S_1} = \sum_{i_1, i_2} [\vsigma^x_{i_1}\vsigma^z_{i_2}, \vsigma^x_{i_1}\vsigma^x_{i_2}]=2(\sum_{i_1})\cdot \sum_{i_2}\vsigma^y_{i_2} \quad \text{and} \quad  \lnormp{\sum_{\labs{S_1}=1}\vF_{S_1}}{\bar{2}} = \theta( n^{3/2}). \quad \text{(colliding)}
\end{align}
The collision of the same Pauli $\vsigma^x_{i_1}$ leads to a ``constructive interference'' over site $i_1$. Consequently, it gives a comparable contribution to the Trotter error, although it has a single sum over $i_2$.
This is not a coincidence; both terms are formally controlled by the advertised quantity
\begin{align}
    \lnormp{\sum_{\labs{S_1}=1}\vF_{S_1}}{\bar{2}}, \lnormp{\sum_{\labs{S_3}=3}\vF_{S_3}}{\bar{2}}=\theta(\sqrt{n}\cdot n)=\theta\L( \normp{\vH}{(1),2}\cdot \normp{\vH}{(0),2}\R).
\end{align}
From this example, we can anticipate a formal proof would require (1)
extracting the local quantities $\normp{\vH}{(1),2}$ and $\normp{\vH}{(0),2}$ from the nested commutators and Hypercontractivity and (2) dealing with the higher-order time dependence.
\subsection{Proof Outline}\label{sec:sketch}

With the above example in mind, we sketch the proof strategy as follows. Recall for any product formula with ordering $\gamma(j)$, weights $a_j$, and number of $J$ exponentials,
the general Trotter error can be represented in a time-ordered exponential 
\begin{align}
    \prod_{j=1}^{J} \e^{\ri a_j\vH_{\gamma(j)} t} = \e^{\ri a_J\vH_{\gamma(J)} t}\cdots \e^{\ri a_1 \vH_{\gamma(1)} t}= \exp_{\CT}(\ri\int (\vec{\CE} + \vH)dt).
\end{align}
The error $\CE$ is time-dependent and takes the commutator form
\begin{align}
    \vCE=\vec{\CE} (\vH_1,\cdots,\vH_\Gamma, t) :&= \sum^{J}_{j=1}\left( \prod_{k=j+1}^{J} \e^{a_k\CL_{\gamma(k)} t} [a_j\vH_{\gamma(j)}] -a_j\vH_{\gamma(j)}\right) \quad \text{where}\quad \CL_{\gamma}[O]:=\ri[\vH_\gamma,O].\label{eq:error_comm}
\end{align}
The particular form depends on the choice of ordering and weights, but fortunately, the precise values of the coefficients $a_j$ will not matter. For $2\ell$-th order Suzuki formulas that we focus on, all we need is a crude uniform bound $\labs{a_k} \le 1$ and that the total formula consists of $\Upsilon=2\cdot 5^{\ell/2-1}$ stages for $J= \Upsilon \cdot \Gamma$. Our combinatorial argument takes norms everywhere and does not rely on delicate cancellations. Our proof will "beat the error $\vec{\CE}$ to death" by Taylor expansion (from right to left).
\begin{fact}[Taylor expansion{\cite[Theorem~10]{thy_trotter_error}}] For any order $g'$,
\begin{align*}
    \e^{\CL_J t}\cdots \e^{\CL_{j+1} t} &= \sum^{g'-1}_{g=1} \sum_{g_J+\cdots+g_{j+1}=g-1}\CL^{g_J}_J\cdots \CL^{g_{j+1}}_2 \frac{t^{g-1}}{g_J!\cdots g_{j+1}!}\\
    &+\sum_{m=j+1}^{J} \e^{\CL_J t}\cdots \e^{\CL_{m+1} t} \int_0^t dt_1 \sum_{g_m+\cdots+g_{j+1}=g'-1,g_m\ge 1}e^{\CL_{m} t_1} \CL^{g_m}_m\cdots \CL^{g_{j+1}}_{j+1} \frac{(t-t_1)^{g_m-1}t^{g'-g_m-1}}{(g_m-1)!\cdots g_{j+1}!}.
\end{align*}
\end{fact}
The $g-1$ exponent will be used consistently in the following. Setting $\CL_j :=a_j\CL_{\gamma(j)}$, Taylor expansion gives the formal expansion for the error in powers of time $t$ 
\begin{align}
    \vCE = \sum_{g=\ell+1}^{g'-1} \vCE_g + \vCE_{\ge g'} \quad  \text{where}\quad \vCE_g = \CO(t^{g-1})\quad \text{and}\quad \vCE_{\ge g'} = \CO(t^{g'-1}).
\end{align}
Each $g$-th order term $\vCE_g$ is a sum of nested commutators, and we control its $p$-norm (Section~\ref{sec:nonrandom_g_th}). We will evaluate Hypercontractivity through a rather involved combinatorics to extract the local quantities $\normp{\vH}{(1),2}$ and $\normp{\vH}{(0),2}$. Note that we will use the version we derived (Proposition~\ref{prop:general_pauli_expansion})
\begin{align}
\lnormp{\vF}{p}^2 &\leq \sum_{S\subset \{m,\cdots,1\}} (C_p)^{|S|} \lnormp{\vF_S}{p}^2.
\end{align}
This will straightforwardly generalize to the case of Fermions (Section~\ref{sec:proof_fermion_non_random}) and is not restricted to the case of qubits.  See Section~\ref{sec:proof_non_random} for comments on how much constant overhead improvement is possible using the other Hypercontractivity $\lnormp{\vF}{\bar{p}}^2 \le \sum_{S} C_p^{\labs{S} } \lnormp{\vF_S}{\bar{2}}^2$ (Proposition~\ref{fact:nc_hyper}).

We handle the edge case $g'$-th order term $\vCE_{g'}$ in Section~\ref{sec:g'_th_non_random}. Indeed, bounding the infinite series will give divergent results, so we must halt the expansion at an appropriate order $g'$. We combine the estimates and apply Markov's inequality in Section~\ref{sec:proof_non_random}.

\subsection{Bounds on the $g$-th Order}\label{sec:nonrandom_g_th}
We proceed by controlling each $g$-th order~\eqref{eq:error_comm} polynomial by Hypercontractivity (Proposition~\ref{prop:unif_subsystem_recap}). We begin with $\CL_j :=a_j\CL_{\gamma(j)}$ to ease notation
\begin{align}
    \lnormp{\vCE_g}{p}^2&= \lnormp{\sum^{J}_{j=1} \sum_{g_J+\cdots+g_{j+1}=g-1}\CL^{g_J}_J\cdots \CL^{g_{j+1}}_{j+1} [\vH_{j}]\frac{t^{g-1}}{g_J!\cdots g_{j+1}!}}{p}^2\\
    &\le \sum_{S\subset \{n,\cdots, 1\}} (C_p)^{|S|} \lnormp{\left[\sum^{J}_{j=1} \sum_{g_J+\cdots+g_{j+1}=g-1}\CL^{g_J}_J\cdots \CL^{g_{j+1}}_{j+1} [\vH_{j}]\frac{t^{g-1}}{g_J!\cdots g_{j+1}!}\right]_S}{p}^2\\
    &\le (C_p)^{g(k-1)+1}\Upsilon(t\Upsilon)^{2(g-1)} \sum_{S\subset \{n,\cdots, 1\}}  \left( \sum_{\gamma_{g-1}}^\Gamma\cdots\sum_{\gamma_0}^\Gamma\lnormp{\left[\CL_{\gamma_{g-1}}\CL_{\gamma_{g-2}}\cdots \CL_{\gamma_1} [\vH_{\gamma_0}]\right]_S}{p}\right)^2.\label{eq:g-order_LLL}
\end{align}

The second inequality uses a uniform bound on locality $\labs{S}\le g(k-1)+1$ and applies a brutal triangle inequality. The last inequality expresses $\CL_j$ by $a_j\CL_{\gamma(j)}$ and $\vH_j$ by $a_j\vH_{\gamma(j)}$ and uses that $\labs{a_j} \le 1$. We also symmetrize the sum over terms $\CL_{\gamma}$ by throwing in extra terms. This costs an extra factor of $(g-1)!$ (which cancels the factor $1/(g-1)!$ in the exponential) due to possible permutation of a $(g-1)$-th order term. For example, consider a particular term 
\begin{align}
    \e^{\CL_3t}\e^{\CL_2t}\e^{\CL_1t} &=  \cdots + \CL_3 \CL_2 \CL_1 t^3+\cdots\\
    \e^{(\CL_3+\CL_2+\CL_1)t} &= \cdots + \CL_3 \CL_2 \CL_1 \frac{t^3}{3!}+\cdots.
\end{align}
The number of stages $\Upsilon$ arise as each term $\CL_{\gamma}$ or $\vH_{\gamma}$ appears $\Upsilon$-times. 

The main lemma of this section is the following recursive estimate for one layer of commutators $\sum_{\gamma} \CL_\gamma$. This is effectively calculating certain ``2-2 norm'' for the commutator $\sum_{S_2} \CL_{S_2}$,
where the  ``2-norm'' is
$
   \sum_{S_1} \L( \sum_\alpha \lnormp{[\vO]^{\alpha}_{S_1}}{p}\R)^2.
$
We will keep this at an analogy level to avoid introducing extra notations.
\begin{lem}[Effective 2-2 norm of the commutator]\label{lem:22norm_with_alpha}
For any set of operators $\{\vO^{\alpha}\}_{\alpha}$,
\begin{align}
    \sum_{S\subset \{n,\cdots, 1\}} \left(\sum_{\gamma}\sum_\alpha\lnormp{\left[\CL_{\gamma}[\vO^\alpha]\right]_S}{p}\right)^2
    &\le \lambda(k)^2(\labs{S}_{max})^{2k} \cdot \sum_{S_1\subset \{n,\cdots, 1\}} \left(\sum_{\alpha}\lnormp{[\vO^\alpha]_{S_1}}{p}\right)^2
\end{align}
where $\CL_{\gamma}[\vO^\alpha]$ is at most $\labs{S_{max}}$-local and 
\begin{align}
    \lambda(k):=\frac{2^{k/2+1}}{(k-1)!} \sum_{k_{f}=1}^k \frac{2^{k_{f}/2}}{{(k-k_{f})!}}\vertiii{\vH}_{(k_{f}),2}.
\end{align}
\end{lem}
Assuming Lemma~\ref{lem:22norm_with_alpha}, iterating it for $(g-1)$-times gives the  estimate 
\begin{align}
    (cont.)\quad \lnormp{\vCE_g}{p}^2 &\le (C_p)^{g(k-1)+1}  (\Upsilon t)^{2(g-1)} \bigg((g(k-1)+1)\cdots (2(k-1)-1)\bigg)^{2k}\lambda(k)^{2(g-1)}\cdot \lnormp{\vH}{(global),2}^2\lnormp{\vI}{p}^2\\
     &\le (C_p)^{g(k-1)+1} g^{2gk}\left(c(k)\Upsilon t\right)^{2(g-1)}\cdot \lnormp{\vH}{(global),2}^2\lnormp{\vI}{p}^2.\label{eq:nonrandom_gthorder}
\end{align}
The first inequality also evaluates the last sum over the Hamiltonian terms $\vH_{\gamma_0}$ by
\begin{align}
    \sum_{S\subset \{n,\cdots, 1\}} \lnormp{\left[\vH_{\gamma_0}\right]_S}{p}^2 \le  \lnormp{\vH}{(global),2}\lnormp{\vI}{p}^2.
\end{align}
The last inequality uses $g(k-1)+1\le g k$ and hides constants depending only on $k$ in the value $c(k)$
\begin{align}
    c(k):=k^{k}\lambda(k).
\end{align}
The expression~\eqref{eq:nonrandom_gthorder} yields the desired estimate for the g-th order error term. 
Unfortunately, the power series is not summable due to the super-exponential factor $g^{2gk}$. We will later truncate the expansion at some properly chosen order $g'$ (Section~\ref{sec:g'_th_non_random}).

What remains in this section is to show Lemma~\ref{lem:22norm_with_alpha}. As hinted in the example (Section~\ref{sec:heuristic_first_order}), we need to systematically handle the cases that grow greedily and those with collisions. Let us identify how taking commutators may produce other sets $S$ (Figure~\ref{fig:S2S1}). 
\begin{align}
    \CL_{\gamma}[\vO_{S_1}] = \sum_S  \left[ \CL_{\gamma}[\vO_{S_1}]\right]_S.
\end{align}
Let $S_2(\gamma)$ be the support of $\gamma$.

\begin{itemize}
    \item If the sets $S_1$ and $S_2(\gamma)$ are disjoint,  the commutator vanishes.
    \item (I) If they overlap on a single site, there is no cancellation. The resulting set is the union $S=S_2(\gamma)\cup S_1$. This was the ``greedy'' term in the example.
    \item (II) If they overlap on more than 1 site, we may lose all but 1 site. The resulting set $S$ is a subset of the union $S\subset S_2(\gamma)\cup S_1$.\footnote{Since the commutator has vanishing partial trace $\tr_S[\vO_{AS},\vO_{SB}]=0$, for operators $\vO_{AS},\vO_{SB}$ partially traceless on $S$; we must have at least 1 Pauli left. This is not the case for Fermions. See~\ref{sec:proof_fermion_non_random}. }
\end{itemize}
To account for the above, we rewrite the sets $S,S_2,S_1$ in terms of the components
\begin{align}
    S   &= S_0\perp \hspace{0.8cm} S_{f} \perp S_+,\\
    S_2 &= \hspace{0.8cm} S_-\perp S_{f}\perp S_+,\\
    S_1 &= S_0\perp S_-\perp S_{f}.
\end{align}
where
\begin{itemize}
    \item $S_0:= S_1/S_2$ are the ``untouched'' sites.
    \item $S_-\subset S_1\cap S_2$ are the sites that got canceled due to collison
    \item $S_f$ are the sites that stayed in all sets $S_0, S_1,S_2$. We must have $\labs{S_f} \ge 1$
    \item $S_+:=S_2/S_1$ are the new sites.
\end{itemize}
We will constantly use this decomposition back and forth in the proof. 

\subsubsection{``Greedy growth'' :Overlapping at 1 site }
To get familiar with the manipulations and notations, we work out the simpler case when the sets overlap on a single site
\begin{align}
    \labs{S_1\cap \gamma} =1\quad \text{for}\quad \CL_{\gamma}[\vO_{S_1}]. 
\end{align}
We will see that the growth due to the commutator $\sum_{\gamma} \CL_{\gamma}$ is controlled by the succinct norm $\lnormp{\vH}{(local),2}^2$, multiplied by some function of the locality $\labs{S}$. To ease the notation, we will also overload the set $S_2(\gamma)$ by $\gamma$.
\begin{align}
    \sum_{S\subset \{n,\cdots, 1\}} \L(\sum_{\gamma} \sum_{\substack{\labs{S_1\cap \gamma}=1 \\ S_1\cup \gamma=S}} \lnormp{\left[ \CL_{\gamma}[\vO_{S_1}]\right]_S}{p}\R)^2 &= \sum_{S\subset \{n,\cdots, 1\}} \L( \sum_{\substack{\labs{S_1\cap S_2}=1 \\ S_1\cup S_2=S}} \sum_{\gamma\sim S_2} \lnormp{\left[ \CL_{\gamma}[\vO_{S_1}]\right]_S}{p}\R)^2\\
    &\le \sum_{S\subset \{n,\cdots, 1\}} \sum_{\substack{\labs{S_1\cap S_2}=1 \\ S_1\cup S_2=S}} \L(\sum_{\gamma\sim S_2}  \lnormp{\CL_{\gamma}[\vO_{S_1}]}{p}\R)^2\cdot \L(\sum_{\substack{\labs{S'_1\cap S'_2}=1 \\ S'_1\cup S'_2=S}}\R)\\
    &\le \sum_{S_1} \sum_{S_2: \substack{\labs{S_1\cap S_2}=1}} \L( \sum_{\gamma\sim S_2}2\norm{\vH_{\gamma}}\R)^2\lnormp{\vO_{S_1}}{p}^2\cdot \L(\binom{\labs{S}}{\labs{S_1}}\labs{S_1}\R)\\
    &\le 4 \cdot  \lnormp{\vH}{(local),2}^2 \cdot \max_{\labs{S},\labs{S_1}}\binom{S}{\labs{S_{1}}}\labs{S_{1}}^2\cdot \left(\sum_{S_1} \lnormp{\vO_{S_1}}{p}^2\right).
\end{align}
The first inequality is Cauchy-Schwartz. The second inequality rearranges the sum over $S_1,S_2$, uses Holder's inequality, and then evaluates the combinations that the two sets $S_1,S_2$ can give rise to the set $S$. In the last inequality, we make the 2-norm $\lnormp{\vH}{(local),2}$ explicit by
\begin{align}
    \sum_{S_2:\labs{S_1\cap S_2}=1} \norm{\vH_{S_2}}^2 &= \sum_{s_1\in S_1}\sum_{S_+\cap S_1 = \emptyset}\norm{\vH_{S_+\cup\{s_1\}}}^2 \\
    &\le \labs{S_1} \cdot \max_{s_1}\sum_{S_+\cap S_1 = \emptyset}\norm{\vH_{S_+\cup\{s_1\}}}^2= \labs{S_1} \cdot \lnormp{\vH}{(local),2}^2.
\end{align}
We also use a uniform upper-bound for the combinatorial function of the set sizes $\labs{S},\labs{S_1}$. 

It is instructive to compare with the $1-1$-norm calculation without invoking Hypercontractivity. 
\begin{align}
    \sum_{S\subset \{n,\cdots, 1\}} \sum_{\gamma} \sum_{\substack{\labs{S_1\cap \gamma}=1 \\ S_1\cup \gamma=S}} \lnorm{\left[ \CL_{\gamma}[\vO_{S_1}]\right]_S} &\le \sum_{S_1} \sum_{S_2: \substack{\labs{S_1\cap S_2}=1}} \sum_{\gamma\sim S_2}2\norm{\vH_{\gamma}}\lnorm{\vO_{S_1}}\\
    & \le 2 \lnormp{\vH}{(local),1} \max_{S_1} \labs{S_1}\cdot \sum_{S_1}\lnorm{\vO_{S_1}}.
\end{align}
This is the 1-norm local quantity that featured in the worst-case Trotter error~\cite{thy_trotter_error}.

\begin{figure}[t]
    \centering
    \includegraphics[width=0.95\textwidth]{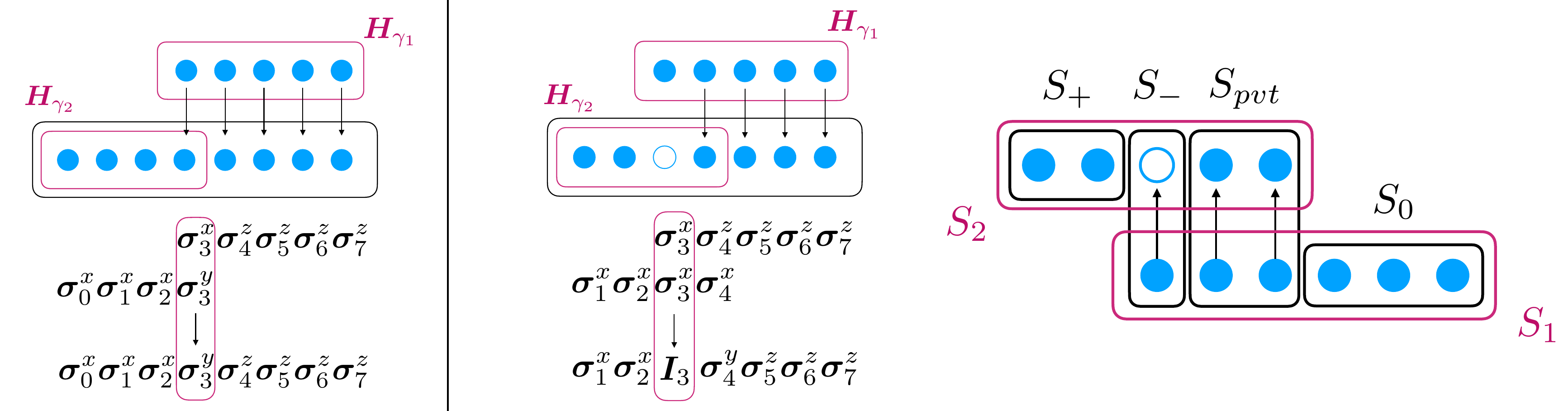}
    \caption{ 
    (Left) Intuitively, commuting with another operator produces new occupancy. (Middle) Unfortunately, occasion cancellations complicate the calculation. (Right) For bookkeeping, we label the possible sets that can be produced by acting the commutator $\CL_{S_2}$ on some operator $\vO_{S_1}$. In their intersection, some subset $S_-$ becomes the identity, and some subset $S_{f}$ remains occupied. For Pauli strings, the fixed subset $S_{f}$ must be non-empty; in the Fermionic case, the fixed subset $S_{f}$ may be empty.
    }
    \label{fig:S2S1}
\end{figure}
\subsubsection{Cancellation and collision due to larger overlap}
The case with cancellation requires delicate notations to handle. Suppose we lose some set in the overlap $ S_- \subset S_1\cap S_2$ due to collision and gain a new set $S_2/S_1=:S_+$ (Figure~\ref{fig:S2S1}). The combinatorics will be organized by the size of the fixed set $k_{f}:=\labs{S_{f}} = 1,\cdots, k$.
\begin{prop}[Fixed $k_f$]\label{prop:fixed_kf} For a value of $\labs{S_{f}}=k_{f} \in \{ 1,\cdots, k\}$,
\begin{align*}
    \sum_{S\subset \{n,\cdots, 1\}} \left(\sum_{\gamma} \sum_{\substack{ \labs{S_{f}}=k_{f} \\
    S_0\perp S_{f} \perp S_+ = S}} \indicator(S_2\sim \gamma) \lnormp{\left[\CL_{\gamma}[\vO_{S_1}]\right]_S}{p}\right)^2 &\le  2^{k+k_{f}+2} \cdot \left(\frac{(\labs{S}_{max})^k}{(k-k_{f})!(k-1)!}\right)^2 \cdot  \normp{\vH}{(k_{f}),2}^2 \cdot \left(\sum_{S_1}\lnormp{\vO_{S_1}}{p}^2\right),\end{align*}
    where the indicator $\indicator(S_2\sim\gamma)$ checks if the set $S_2$ coincides with the support of $\gamma$ and
    \begin{align}
    \lnormp{\vH}{(k_f),2} := \sqrt{ \max_{|S_f|=k_f } \sum_{S\supset S_f} b_{S}^2} \quad \text{where}\quad  b_{S}:=\sum_{\gamma\sim S} \norm{\vH_{\gamma}}.
    \end{align}
\end{prop}
To connect to our notation in the main text, for $k_f= 0$, this is what we defined as the global norm $\lnormp{\vH}{(0),2} = \lnormp{\vH}{(global),2}$; for $k_f= 1$, this is what we defined as the local norm $\lnormp{\vH}{(1),2} = \lnormp{\vH}{(local),2}$. The norms for $k_f\ge 2$ are more of a proof artifact. To be careful with the distinction between an operator $\vO$ and its local component $\vO_{S}$, we first note the following bound.
\begin{fact} \label{fact:project_2^S} For any set $S$ and operator $\vO$, we have 
\begin{align}
    \lnormp{ \prod_{s \in S} (1-E_s) [\vO]}{p} \le 2^{\labs{S}}\lnormp{\vO}{p} \quad \text{where}\quad E_s :=\vI_{s}\frac{\tr_{s}[\cdot]}{\tr[\vI_{s}]}.
\end{align}
\end{fact}
\begin{proof}[Proof of Fact~\ref{fact:project_2^S}]
Use monotonicity of partial trace (Fact~\ref{fact:nc_convexity}), i.e., the conditional expectation $E_s$ is norm non-increasing. The last factor $2^{\labs{S}}$ is due to a brutal triangle inequality.
\end{proof}
\begin{proof}[Proof of Proposition~\ref{prop:fixed_kf}]
\begin{align}
    \sum_{S\subset \{n,\cdots, 1\}} \left(\sum_{\substack{ \labs{S_{f}}=k_{f} \\S  = 
    S_0\perp S_{f} \perp S_+}} \sum_{\gamma\sim S_2}\lnormp{\left[\CL_{\gamma}[\vO_{S_1}]\right]_S}{p}\right)^2\label{eq:rest_of_effective_22norm}
    &\le\sum_{S\subset \{n,\cdots, 1\}} \left(2^{\labs{S_{f}}} \sum_{S_{f},S_+}\sum_{S_-}\sum_{\gamma\sim S_2} 2\norm{\vH_{\gamma}}\lnormp{\vO_{S_{f}S_-S_0}}{p}\right)^2\\ 
     &\le (4^{k_f+1})\cdot \sum_{S\subset \{n,\cdots, 1\}}  \sum_{S_{f},S_+}\left(\sum_{S_-}b_{S_2} \lnormp{\vO_{S_{f}S_-S_0}}{p}\right)^2\cdot \L(\sum_{S'_0\perp S'_{f} \perp S'_+=S} 1\R)\\ 
    &\le(\cdot)(\cdot)\sum_{S\subset \{n,\cdots, 1\}}  \sum_{S_{f},S_+}\left(\sum_{S_-}b_{S_fS_-S_+}^2\right)\left(\sum_{S'_-}\lnormp{\vO_{S_{f}S'_-S_0}}{p}^2\right)\\
    &= (\cdot)(\cdot) \sum_{S_{f}}\left(\sum_{S_-,S_+}b_{S_fS_-S_+}^2\right)\left(\sum_{S'_-,S_0}\lnormp{\vO_{S_{f}S'_-S_0}}{p}^2\right)\\
    &\le 2^{k+k_{f}+2} \cdot \left(\frac{(\labs{S}_{max})^k}{(k-k_{f})!(k-1)!}\right)^2 \cdot  \normp{\vH}{(k_{f}),2}^2 \cdot \left(\sum_{S_1}\lnormp{\vO_{S_1}}{p}^2\right).
\end{align}
The first inequality uses that $\normp{\left[\CL_{\gamma}[\vO_{S_1}]\right]_S}{p}= \normp{\prod_{s \in S_f} (1-E_s) \CL_{\gamma}[\vO_{S_1}]}{p}\le 2^{\labs{S_f}}\normp{\CL_{\gamma}[\vO_{S_1}]}{p}$ via Fact~\ref{fact:project_2^S}. 
The second inequality evaluates $\sum_{\gamma\sim S} \norm{\vH_{\gamma}}= b_{S}$ and uses Cauchy-Schwartz w.r.t to the sum over sets $S_{f}, S_+, S_0$ associated with a given set $S$. The third inequality uses Cauchy-Schwartz w.r.t the sum over sets $S_-$.
We also evaluate the elementary sum (the last inequality here uses that the largest term is attained at $\labs{S_+} = k- \labs{S_f}$. )
\begin{align}
    \sum_{S'_0\perp S'_{f} \perp S'_+=S} 1 = \sum_{\labs{S_+}=0}^{k-\labs{S_{f}}} \binom{S}{S_{f}} \binom{S-S_{f}}{S_+}
    &= \sum_{\labs{S_+}=0}^{k-\labs{S_{f}}} \frac{\labs{S}!}{\labs{S_{f}}!\labs{S_+}!(\labs{S}-\labs{S_+}-\labs{S_{f}})!}\\
    &\le \sum_{\labs{S_+}=0}^{k-\labs{S_{f}}} \frac{\labs{S}^{\labs{S_+}}}{\labs{S_{+}}!}\frac{\labs{S}^{\labs{S_{f}}}}{\labs{S_{f}}!}
     \le (k-k_{f}+1) \frac{\labs{S}^{k}}{(k-k_{f})!k!}.
\end{align}
The equality rearranges the sum. The last inequality uses the following estimates
\begin{align}
\max_{S_{f}} \left(\sum_{S_-,S_+} b_{S_fS_-S_+}^2\right) &=\max_{S_{f}} \left(\sum_{S'}\sum_{S_-\cup S_+=S'}b_{S_fS_-S_+}^2\right) \\
&\le 2^{k-k_{f}} \cdot \normp{\vH}{(k_{f}),2}^2,     
\end{align}
and
\begin{align}
   \sum_{S_{f},S'_-,S_0}\lnormp{\vO_{S_{f}S'_-S_0}}{p}^2\le \sum_{S_1} \sum_{S_{f}\perp S'_-\perp S_0=S_1}\lnormp{\vO_{S_1}}{p}^2&\le \sum_{\labs{S'_-}=0}^{k-k_{f}} \binom{\labs{S_1}}{S_{f}} \binom{\labs{S_1}-S_{f}}{S'_-}\cdot \sum_{S_1}\lnormp{\vO_{S_1}}{p}^2\\
   &\le (k-k_{f}+1) \frac{(\labs{S}_{max})^{k}}{(k-k_{f})!k!}\cdot \sum_{S_1}\lnormp{\vO_{S_1}}{p}^2.
\end{align}
These, together with the hidden constants $(\cdot)(\cdot)$, give the ultimate prefactors.
\end{proof}

We can now prove the main lemma by summing over the set sizes $k_f = 1,\cdots, k$.
\begin{proof}[Proof of Lemma~\ref{lem:22norm_with_alpha}]
\begin{align}
    \sum_{S\subset \{n,\cdots, 1\}} \left(\sum_{\gamma}\sum_\alpha\lnormp{\left[\CL_{\gamma}[\vO^\alpha]\right]_S}{p}\right)^2
     &= \sum_{S\subset \{n,\cdots, 1\}} \left(\sum_{S_2,S_1}\sum_{\gamma}\indicator(\gamma\sim S_2)\sum_\alpha\lnormp{\left[\CL_{\gamma}[\vO^\alpha_{S_1}]\right]_S}{p}\right)^2\\
    & = \sum_{S\subset \{n,\cdots, 1\}} \left(\sum_{k_f=1}^k\sum_{\substack{ \labs{S_{f}}=k_{f} \\S  = 
    S_0\perp S_{f} \perp S_+}}\sum_\alpha\lnormp{\left[\CL_{\gamma}[\vO^\alpha_{S_{f}S_-S_0}]\right]_S}{p}\right)^2\\
    & \le\left( \sum_{k_f=1}^k \sqrt{\sum_{S\subset \{n,\cdots, 1\}} \left( \sum_{\substack{ \labs{S_{f}}=k_{f} \\S  = 
    S_0\perp S_{f} \perp S_+}}\sum_\alpha\lnormp{\left[\CL_{\gamma}[\vO^\alpha_{S_{f}S_-S_0}]\right]_S}{p}\right)^2} \right)^2.
\end{align}
The second equality presents the sets $S_1,S_2$ by the decomposition $S_{f},S_-,S_+,S_0$ and isolates the sum over $k_f$. The last inequality might look intimidating, but it is actually a triangle inequality (over values $\labs{S_{f}}$ ) for certain 2-norm
\begin{align}
   \sqrt{\sum_{S} (\sum_{k_f} f(k_f,S))^2} \le \sum_{k_f} \sqrt{\sum_{S}  f(k_f,S)^2} \quad \text{for any function} \quad f(k_f,S). 
\end{align}
We may now use a variant of Proposition~\ref{prop:fixed_kf} with an additional sum over an abstract set $\sum_\alpha$. The derivation is analogous by keeping the sum at the innermost layer (sticking to the operator $\vO^\alpha$). with the replacement 
\begin{align}
\lnormp{\vO_{S_{f}S_-S_0}}{p}\rightarrow \left(\sum_{\alpha}\lnormp{\vO^\alpha_{S_{f}S_-S_0}}{p}\right).    
\end{align}
This is the advertised result.
\end{proof}

\subsection{Bounds for $g'$-th Order and Beyond.}\label{sec:g'_th_non_random}
The previous section evaluates the g-th order terms $\CE_g$. This section takes care of the last term in the Taylor expansion $\CE_{\ge g'}$. To ease notation, we set the dummy variable to be $g'\rightarrow g$. It has infinite-order dependence on time, so we have to tweak the calculations. Recall~\eqref{eq:error_comm},
\begin{align}
    \lnormp{\vec{\CE}_{\ge g}}{p}
    &= \lnormp{ \sum^{J}_{j=1} \sum_{m=j+1}^{J} \e^{\CL_J t}\cdots \e^{\CL_{m+1} t} \int_0^t dt_1 \sum_{g_m+\cdots+g_{j+1}=g-1,g_m\ge 1}e^{\CL_{m} t_1} \CL^{g_m}_m\cdots \CL^{g_{j+1}}_{j+1}[\vH_{j}] \frac{(t-t_1)^{g_m-1}t^{g-1-g_m}}{(g_m-1)!\cdots g_{j+1}!}}{p}\\
    &\le \sum_{m=2}^{J} \lnormp{ \sum^{J}_{j=m-1} \sum_{g_m+\cdots+g_{j+1}=g-1,g_m\ge 1}\CL^{g_m}_m\cdots \CL^{g_{j+1}}_{j+1}[\vH_{j}] \frac{t^{g-1}}{g_m!\cdots g_{j+1}!}}{p}\label{eq:apply_low_particle}\\
    &\le \sqrt{C_p}^{g(k-1)+1}\Upsilon (t\Upsilon)^{g-1} \sum_{\gamma_{g-1}}^\Gamma \sqrt{ \sum_{S\subset \{n,\cdots, 1\}}  \left( \sum_{\gamma_{g-2}}^\Gamma\cdots\sum_{\gamma_0}^\Gamma\lnormp{\CL_{\gamma_{g-1}}\left[\CL_{\gamma_{g-2}}\cdots \CL_{\gamma_1} [\vH_{\gamma_0}]\right]_S}{p}\right)^2 }\label{eq:g'-order_LLL}.
\end{align}
The first inequality exchanges the summation order, applies the triangle inequality, integrates over time, and removes the unitary conjugations by unitary invariance of p-norms. The second inequality is a similar calculation to~\eqref{eq:g-order_LLL}. We use Hypercontractivity, pull the p-norm inside the sum, and symmetrize the sum by completing the exponential for $\gamma_{g-2}\cdots \gamma_0$. 

The only difference from~\eqref{eq:g-order_LLL} is the outer-most sum \textit{outside} the square root.
\begin{lem}[Sum outside the square-root] \label{lem:22norm_with_alpha_beyond_g}
\begin{align}
    \sum_{\gamma} \sqrt{\sum_{S\subset \{n,\cdots, 1\}} \left(\sum_\alpha\lnormp{\left[\CL_{\gamma}[\vO^\alpha]\right]_S}{p}\right)^2 } &\le \lambda'(k) \cdot  \sqrt{\sum_{S_1\subset \{n,\cdots, 1\}} \left(\sum_{\alpha}\lnormp{[\vO^\alpha]_{S_1}}{p}\right)^2}
\end{align}
where 
\begin{align}
\lambda'(k)= 2\cdot \sum_{k'=1}^k \binom{k}{k'} {\sqrt{20}}^{k'} \sqrt{\frac{1}{k'!}\lnormp{\vH}{(k'),1}\lnormp{\vH}{(global),1}}.    
\end{align}
\end{lem}
We can evaluate the bound using Lemma~\ref{lem:22norm_with_alpha_beyond_g} for the outer-most sum $\sum_{\gamma_{g-1}}^{\Gamma}$ and Lemma~\ref{lem:22norm_with_alpha} for $\gamma_{g-2},\cdots \gamma_1$
\begin{align}
(cont.) &\le \Upsilon (\Upsilon t)^{g-1}   \sqrt{C_p}^{g(k-1)+1}   (g(k-1)+1)^{k/2}  \lambda'(k) \cdot \bigg( ((g-1)(k-1)+1)\cdots (2(k-1)-1)\bigg)^{k} \lambda(k)^{g-2}\cdot \lnormp{\vH}{(global),2}\lnormp{\vI}{p}\notag\\
    &\le  \sqrt{C_p}^{g(k-1)+1}   c'(k) \cdot g^{gk}\Upsilon\left(c(k)\Upsilon t\right)^{g-1}\lnormp{\vI}{p}.\label{eq:non_random_beyond_g}    
\end{align}
The last inequality absorbs constants into $c'(k) $ 
\begin{align}
    c'(k) &= \frac{\lambda'(k)}{\lambda(k)} \frac{1}{\sqrt{k}^k}\lnormp{\vH}{(global),2}.
\end{align}
 In other words, the higher-order time dependence forces us to apply triangle inequality for the outer layer sum; fortunately, we can still use Lemma~\ref{lem:22norm_with_alpha} for the inner sums. These give the different prefactor $c'(k)$.

\begin{proof}[Proof of Lemma~\ref{lem:22norm_with_alpha_beyond_g}]
The calculation is analogous to Lemma~\ref{lem:22norm_with_alpha}. We define a slightly different quantity
\begin{align}
    \labs{S_-}+\labs{S_{f}} := k'
\end{align}
that will organize the combinatorics (the analog of the number $k_f$ in Lemma~\ref{lem:22norm_with_alpha}). We first rearrange the expression in terms of the subsets $S_+,S_-,S_0,S_f$.

\begin{align}
    \sum_{\gamma} \sqrt{\sum_{S\subset \{n,\cdots, 1\}} \left(\sum_\alpha\lnormp{\left[\CL_{\gamma}[\vO^\alpha]\right]_S}{p}\right)^2 }&\le
    \sum_{S_2}\sum_{\gamma\sim S_2} \sqrt{ \sum_{S\subset \{n,\cdots, 1\}} \left(\sum_{S_1}\sum_\alpha\lnormp{\left[\CL_{\gamma}[\vO^\alpha_{S_1}]\right]_S}{p}\right)^2 }\\
    &\le  \sum_{S_2}\sum_{\gamma\sim S_2} \sqrt{ \sum_{S\subset \{n,\cdots, 1\}} \left(\sum_{k'=1}^k\sum_{\substack{S_-,S_{f}\subset S_2\\ \labs{S_-}+\labs{S_{f}} = k'}}2^{\labs{S_{f}}} \sum_\alpha\lnormp{\CL_{\gamma}[\vO^\alpha_{S_{f}S_-S_0}]}{p}\right)^2}\\
    &\le \sum_{k'=1}^k  \sum_{S_2} \sum_{\gamma\sim S_2}\sqrt{ \sum_{S\subset \{n,\cdots, 1\}} \left(\sum_{\substack{S_-,S_{f}\subset S_2\\ \labs{S_-}+\labs{S_{f}} = k'}}2^{\labs{S_{f}}} \sum_\alpha\lnormp{\CL_{\gamma}[\vO^\alpha_{S_{f}S_-S_0}]}{p}\right)^2}.
\end{align}
The first inequality parameterizes the sets $S_1 = S_{f}S_-S_0$ that could give rise to $S$ after taking the commutator $\CL_{\gamma}$. The factor $2^{\labs{S_{f}}}$ is due to Fact~\ref{fact:project_2^S}. The second inequality is a triangle inequality to postpone the sum over $k' = \labs{S_{f}}+\labs{S_-}$. 

Next, we use Cauchy-Schwartz to break the non-linear expression into individual pieces. This costs multiplicative constant overheads that depend only on $k$.
\begin{align}
    (cont.)&\le  \sum_{k'=1}^k\sum_{S_2} \sum_{\gamma\sim S_2} \sqrt{ \sum_{S_0} \sum_{\substack{S_-,S_{f}\subset S_2\\ \labs{S_-}+\labs{S_{f}} = k'}} \left( \sum_\alpha\lnormp{\CL_{\gamma}[\vO^\alpha_{S_{f}S_-S_0}]}{p}\right)^2 (\sum_{S'_{f},S'_-} 2^{2\labs{S'_{f}}}) }\\
    &\le  \sum_{k'=1}^k\sum_{\substack{S_+,S_-,S_{f}\\ \labs{S_-}+\labs{S_{f}} = k'}} 2b_{S_+S_-S_{f}} \sqrt{ \sum_{S_0} \left( \sum_\alpha\lnormp{\vO^\alpha_{S_{f}S_-S_0}}{p}\right)^2 (\cdot)}\\
    &\le \sum_{k'=1}^k 2\sqrt{\sum_{\substack{S_-,S_{f}\\ \labs{S_-}+\labs{S_{f}} = k'}} \L( \sum_{S_+} b_{S_+S_-S_{f}} \R)^2 } 
    \sqrt{ \sum_{\substack{S_-,S_{f},S_0 \\ \labs{S_-}+\labs{S_{f}} = k'}} \left( \sum_\alpha\lnormp{\vO^\alpha_{S_{f}S_-S_0}}{p}\right)^2 (\cdot)  }.
\end{align}
The first inequality is Cauchy-Schwartz, where the sum evaluates to
\begin{align}
     (\cdot)=\sum_{S'_{f},S'_-} 2^{2\labs{S'_{f}}} = \binom{k}{k'} \cdot \sum^{k'}_{k_{f}=1} \binom{k'}{k_{f}} 2^{2\labs{S'_{f}}} =  \binom{k}{k'} 5^{k'}.
\end{align}
The second inequality is a triangle inequality for the sum over subsets $S_-,S_{f}\subset S_2$, which then combines with the sum over $S_2$.  The fifth inequality is Cauchy-Schwartz's. 
Lastly, we evaluate the combinatorial factors for each term
\begin{align}
    \sum_{S_-,S_{f}} \L( \sum_{S_+} b_{S_+S_-S_{f}}\R)^2 &\le \sum_{S_-,S_{f}} \sum_{S_+} b_{S_+S_-S_{f}} \cdot \max_{S'_-,S'_{f}} \sum_{S'_+} b_{S'_+S'_-S'_{f}}\\
    &=  \binom{k}{k'} 2^{k'}\lnormp{\vH}{(global),1}\cdot \normp{\vH}{(k'),1},
\end{align}
and
\begin{align}
    \sum_{S_-,S_{f},S_0} \left( \sum_\alpha\lnormp{\vO^\alpha_{S_{f}S_-S_0}}{p}\right)^2  = \binom{\labs{S_{max}}}{k'} 2^{k'}\cdot  \sum_{S} \left( \sum_\alpha\lnormp{\vO^\alpha_{S}}{p}\right)^2.
\end{align}
These give the advertised result.
\end{proof}

\subsection{Proof of Theorem~\ref{thm:Trotter_non_random}}\label{sec:proof_non_random}
\begin{proof}
For a short time $\tau$, we arrange and perform the last integral using estimate $\int (\tau')^{g-1} d\tau'\le \tau'^{g}$
\begin{align}
    \frac{\lnormp{e^{\iunit \vH \tau}- \vec{S}_\ell(\tau)}{p}}{\lnormp{\vI}{p}} &\le \int^\tau_0 \frac{\lnormp{\vCE(\tau')}{p}}{\lnormp{\vI}{p}} d\tau' \\
    &\le\frac{\sqrt{C_p}}{c(k)} \lnormp{\vH}{(global),2} \cdot \sum_{g=\ell+1}^{g'-1} \left(g^{k}\sqrt{C_p}^{k-1}c(k)\Upsilon\tau \right)^{g}   
    + \sqrt{C_p}  \frac{c'(k)}{c(k)} \cdot \left(g'^{k}\sqrt{C_p}^{k-1}c(k)\Upsilon\tau \right)^{g'}\\
    &:= c'_{1,p} \sum_{g=\ell+1}^{g'-1} \left(g^{k}b_p\tau \right)^{g} + c'_{2,p} \left(g'^{k}b_p\tau \right)^{g'}\\
    &\le \frac{c'_{1,p}}{1-1/\e} \L( (\ell+1)^{k} b_p\tau \R)^{\ell+1} + c'_{2,p} \exp\L(-\frac{1}{\e(b_p\tau)^{1/k}} +1\R)\\
    &:= c_{1,p} ( b_p\tau )^{\ell+1} + c_{2,p} \lexp{-\frac{1}{\e(b_p\tau)^{1/k}}}.
\end{align}
In the second inequality we call the bounds for each $g-th$ order~\eqref{eq:nonrandom_gthorder} and the $g'$-th order~\eqref{eq:non_random_beyond_g} for a good value of 
\begin{align}
g' = \L\lfloor \frac{1}{\e(b_p\tau)^{1/k}} \R\rfloor.    
\end{align}
This is possible as long as the following holds.
\begin{constraint}\label{constr:non_random_1st}
$  (\frac{1}{b_p\tau})^{1/k}\ge \e (\ell+3).  $
\end{constraint}

Then, the total Trotter error at a long time $t = r \cdot \tau $ is bounded by a telescoping sum 
\begin{align}
    \frac{\lnormp{\vCE_{tot} }{p}}{\lnormp{\vI}{p}} := \frac{\lnormp{e^{\iunit \vH t}- \vec{S}_\ell(t/r)^r}{p}}{\lnormp{\vI}{p}} \le r \cdot \frac{\lnormp{e^{\iunit \vH t/r}- \vec{S}_\ell(t/r)}{p}}{\lnormp{\vI}{p}} &\le  c_{1,p} \frac{(b_pt)^{\ell+1}}{r^{\ell}} + r c_{2,p} \lexp{-\frac{1}{\e}(\frac{r}{b_pt})^{1/k} }\\
    &\le 2c_{1,p} \frac{(b_pt)^{\ell+1}}{r^{\ell}}\le p^\eta 2c_{1} \frac{(bt)^{\ell+1}}{r^{\ell}}.\\
    \text{where}\quad \eta:=\frac{(\ell+1)(k-1)+1}{2}.
\end{align}
At the second line we restrict to sufficiently large values of r that the first term dominates.\footnote{ In obtaining Constraint~\ref{constr:non_random_2nd}, note that factors of $p$ cancels out $c_{2,p}/c_{1,p} = c_{2}/c_{1}$.} 
\begin{constraint}\label{constr:non_random_2nd}
$(\frac{1}{b_p\tau})^{1/k} \ge \e\ln\L( \frac{c_2}{c_1} (\frac{1}{b_p\tau})^{\ell+1} \R)$.
\end{constraint}
The last inequality isolates the $p$-dependence and we use $C_p = p-1\le p$. 

Next, for each value of $r$,  
\begin{align}
\text{choose}\quad p= \L( \frac{\epsilon r^\ell}{2c_1(bt)^{\ell+1}}\R)^{1/\eta} \frac{1}{\e}.
\end{align}
Via Markov's inequality, this gives concentration for its singular values(or over any 1-design inputs)
\begin{align}
    \L(\frac{ \lnormp{\vec{\CE}_{tot}}{p}}{\epsilon\normp{\vI}{p}} \R)^p
    \le \lexp{- \frac{\eta}{\e } ( \frac{\epsilon r^{\ell}}{2c_1(bt)^{\ell+1}} )^{1/\eta}  } = \delta .
\end{align}
Choose 
\begin{align}
    r\quad \text{such that}\quad p\ge \max\L( 2, \log(1/\delta )/ \eta \R),
\end{align}
 which explicitly evaluates to
\begin{align}
    r\ge   \L(  \frac{2\sqrt{\e (2+\log(1/\delta ) / \eta)}}{\e-1} \L((\ell+1)\sqrt{\e(2+\log(1/\delta ) / \eta) }\R)^{(\ell+1)(k-1)} \cdot \frac{\lnormp{\vH}{(global),2}\Upsilon t }{\epsilon} \R)^{1/\ell} c(k)\Upsilon t.
\end{align}

We also need to comply with both Constraint~\ref{constr:non_random_1st} and Constraint~\ref{constr:non_random_2nd}, which summarize\footnote{We use $2\max(a,b) \ge a+b$ to simplify the constraints.} to 

\begin{align}
    \frac{1}{b_p\tau} \ge a \quad \text{where} \quad    a := \max \L[  \L(\e(\ell+3) \R)^{k}, \L(2\e \ln\L( \frac{c_2}{c_1} \R)  \R)^{k}, x \R].
\end{align}
The constant $x$ is the unique solution to the transcendental equation
\begin{align}
    x\quad \text{such that}\quad x = (2\e(\ell+1))^{k}\cdot \ln^{k}(x).
\end{align}
Rearrange to obtain
\begin{align}
    r \ge a^{2\eta/k} \L(  c(k)(\Upsilon t)^{1/k} \R) \L(\frac{1-1/\e}{2 \e^{\eta} (\ell+1)^{(\ell+1)k}    } \frac{\epsilon}{\lnormp{\vH}{(global),2}}\R)^{\frac{k-1}{k}},
\end{align}
And recall the explicit values
\begin{align}
    c_1 &= \frac{\norm{\vH}_{(0),2}}{c(k)(1-1/\e)}\\
    c_2 &=\e \frac{c'(k)}{c(k)}\\
    c(k)&=ak^{k}\lambda(k)\\
    c'(k)&= \frac{\lambda'(k)}{\lambda(k)} \frac{1}{\sqrt{k}^k}\lnormp{\vH}{(global),2}
\end{align}
and
\begin{align}
    \lambda'(k)&= 2\cdot \sum_{k'=1}^k \binom{k}{k'} {\sqrt{20}}^{k'} \sqrt{\frac{1}{k'!}\lnormp{\vH}{(k'),1}\lnormp{\vH}{(global),1}}\\
    \lambda(k)&=\frac{2^{k/2+1}}{(k-1)!} \sum_{k'=1}^k \frac{2^{k'/2}}{{(k-k')!}}\lnormp{\vH}{(k'),2}.
\end{align}
The above expressions for gate count are for numerical evaluation; for comprehension, use $\Omega(\cdot)$ to suppress functions of $k,\ell$ (such as the number of stages $\Upsilon$) and note the local norms are decreasing with $k'$. 

\begin{align}
     r &= \Omega\bigg[ \ln(\delta)^{\eta/\ell}\L(\frac{ \lnormp{\vH}{(global),2}}{\epsilon\lnormp{\vH}{(local),2}} \R)^{\frac{1}{\ell}} \L( \lnormp{\vH}{(local),2} t\R)^{1+\frac{1}{\ell}} \notag\\
     &\hspace{2cm}+(\lnormp{\vH}{(local),2}t)^{\frac{1}{k}} \L(\frac{\epsilon \lnormp{\vH}{(local),2}}{\lnormp{\vH}{(global),2}}\R)^{\frac{k-1}{k}} \L( \ln\L(  \frac{\sqrt{\lnormp{\vH}{(local),1}\lnormp{\vH}{(global),1}}}{\lnormp{\vH}{(local),2}}\R)   \R)^{2\eta} \bigg]\\
     &=\Omega\L[ \ln(\delta)^{\frac{k-1}{2}} \lnormp{\vH}{(local),2} t \cdot \L( \ln(\delta)^{k/2}\frac{ \lnormp{\vH}{(global),2} t}{\epsilon} \R)^{\frac{1}{\ell}} \R].
\end{align}
The first term dominates for large time, system size, and error (fixing the value of failure probability $\delta$). The gate complexity is given by $G = r\cdot \Upsilon \cdot \Gamma$. This is the advertised result.
\end{proof}
\subsubsection{Constant overhead improvement from another Hypercontractivity}\label{sec:another_hyper}
One may consider directly apply the existing Hypercontractivity $\lnormp{\vF}{\bar{p}}^2 \le \sum_{S} C_p^{\labs{S} } \lnormp{\vF_S}{\bar{2}}^2$ (Proposition~\ref{fact:nc_hyper}). However, one needs to go through the same combinatorial estimates, with minor constant overheads improvements by replacing $\lnormp{\vO_{S_1}}{p}^2\rightarrow \lnormp{\vO_{S_1}}{2}^2$ and discarding Fact~\ref{fact:project_2^S}. Unfortunately, what comes into the ultimate quantity $\lnormp{\vH}{(local),2}$ is the spectral norm $\norm{\vH_{\gamma}}$ coming from a Holder's inequality
\begin{align}
    \lnormp{\CL_{\gamma} [\vO_{S_1}]}{p} \le 2 \norm{\vH_{\gamma}} \lnormp{\vO_{S_1}}{p},
\end{align} 
and it requires more accounting to get better estimates.

\subsection{Spin Models at a Low Particle Number}\label{sec:spin_low_particle}
In many Hamiltonians, each term $\vH_{\gamma}$ preserves the particle number and the total Hilbert space decomposes into a direct sum of subspaces labeled by their particle number. The input state may have a known particle number. 

In this section, we will present an appropriate notion of concentration for input states drawn randomly from a fixed particle number subspace. Formally, denote the m-particle subspace by the orthogonal projector 
\begin{align}
    \vP_m :=\sum_{\#(\ket{1}) =m} \big( \ket{0}\cdots \ket{1}\big) (h.c.) = \sum_{\#(\ket{1}) =m} \ket{0}\bra{0}\cdots \ket{1}\bra{1},
\end{align}
then particle number preserving means
\begin{align}
    [\vP_{m'}, \vH_{\gamma}] = 0 \quad \text{for each} \quad m',\gamma.
\end{align}
We need to first define the appropriate $k$-locality in this case by expanding the Hamiltonian in the basis 
\begin{align}
\vH &= \sum_{S_-\perp S_+\perp S_z\subset \{m,\cdots,1\}} b_{S_+S_-S_z} \prod_{s_+\in S_+ } \vsigma^+_{s_+} \prod_{s_-\in S_- } \vsigma^-_{s_-}\prod_{s_z\in S_z }\vO^\eta_{s_z}\label{eq:k_local_low_particle}\\
&\text{where} \quad     \vsigma^+ :=\ket{1}\bra{0}, \vsigma^- :=\ket{0}\bra{1}, \quad\text{and}\quad     \vO^\eta:= (1-\eta)\ket{1}\bra{1}-\eta\ket{0}\bra{0}.
\end{align}
$k$-locality in this basis is defined by
\begin{align}
 k =\labs{S_-}+\labs{S_+}+\labs{S_z}.
\end{align}
Note that particle number preserving enforces the number of raising and lower operators match $\labs{S_+}=\labs{S_-}$. This expansion is motivated by an auxiliary product state $\vrho_{\frac{m}{n}}$ that closely relates to the normalized subspace projector $\bar{\vP}_m = \vP_m/\tr[\vP_m]$. Intuitively, the operator $\vO^{\eta}$ is the analog of Pauli $\vsigma^z$ in a biased background 
\begin{align}
\vrho_i := \eta\ket{1}\bra{1}+(1-\eta)\ket{0}\bra{0} \quad \text{such that}\quad  \tr[\vO^{\eta}\vrho_i] = 0.    
\end{align}
See Section~\ref{sec:product_bg} for the details on the construction. Here, we present the concentration result for Trotter error.
\begin{prop}[Trotter error in $k$-local models] \label{prop:Trotter_non_random_low}
To simulate a number preserving $k$-local Hamiltonian using the $\ell$-th order Suzuki formula on the m-particle subspace $\vP_m$, the gate complexity 
\begin{align}
        G =\Omega\L( \left(\poly(n,m)^{1/p}\frac{p^{\frac{k}{2}} \lnormp{\vH}{(global),2} t}{\epsilon} \right)^{1/\ell} \Gamma p^{\frac{k-1}{2}} \lnormp{\vH}{(local),2} t\R) \ \ &\textrm{ensures}\ \  \lnormp{\e^{\iunit \vH t}- \vec{S}_{\ell}(t/r)^r}{p,\bar{\vP}_m} \le \epsilon,
\end{align} 
where the quantities $\lnormp{\vH}{(global),2}$ and $\lnormp{\vH}{(local),2} $ are defined w.r.t to~\eqref{eq:k_local_low_particle}.
\end{prop}
Note that we have drop the parameter $s$ in $\lnormp{\cdot}{p,\bar{\vP}_m,s}$ since every term commutes with $\vP_m$ (and the auxiliary state $\vrho_{\frac{m}{n}}$).
\begin{proof}
The result quickly follows by converting to the p-norm w.r.t. the auxiliary product state 
$\vrho_{\eta}=\otimes_i \vrho_{i}=\otimes_i (\eta\ket{1}\bra{1}+(1-\eta)\ket{0}\bra{0})_i$ defined by the filling ratio $\eta = \frac{m}{n}$. For $\vF = [\vec{\CE} (\vH_1,\cdots,\vH_\Gamma, t)]_{g}$, 
\begin{align}
     \displaystyle \lnormp{\vF}{p,\bar{\vP}}  \le \lnormp{\vF}{p,\vrho_{\eta}} \cdot\left(\poly(n,m)\right)^{1/p} \le \sqrt{\sum_{S\subset \{m,\cdots,1\}} (C_p)^{|S|} \lnormp{\vF_S}{p,\vrho_{\eta}}^2 }\cdot\left(\poly(n,m)\right)^{1/p}.
\end{align}
Some technical notes: Holder's inequality still works for the $\vrho_{\eta}-$ weighted norms\footnote{Due to particle number preserving, i.e., $\vrho_{\eta}$ commutes with any term produced by the Hamiltonian. } $\lnormp{\vH_{\gamma} \vO}{p,\vrho_{\eta}}\le \lnormp{\vO}{p,\vrho_{\eta}} \norm{\vH_{\gamma}}$ (which needs not be true for general $\vrho$); if $\vO$ is particle number preserving, then $[\vO]_{S}$ is also particle number preserving \footnote{This can be seen by $\vO$ is the sum of terms that each has the same number of $\vsigma^+$ and $\vsigma^-$. Removing some of them by $[\vO]_S$ does not change this structure.}. 
\end{proof}

Via Markov's inequality (plug $\vrho =\bar{\vP}_m $ into Proposition~\ref{prop:typical_Schatten}), we obtain concentration.
\begin{cor} 
Draw $\ket{\psi}$ from a $m = \eta n$ - particle subspaces (i.e., $\BE[ \ket{\psi}\bra{\psi}] =\vP_{m}/\tr[\vP_{m}]$), then \begin{align*}
    G &=\Omega \L( \left(\frac{\sqrt{\log(\poly(n,m)/\delta)}^{k}\lnormp{\vH}{(global),2} t}{\epsilon} \right)^{1/\ell} \sqrt{\log(\poly(n,m)/\delta)}^{k-1} \Gamma \lnormp{\vH}{(local),2} t \R)\notag \\ &\hspace{4cm} \textrm{ensures}\ \  \Pr \L( \lnormp{\e^{\iunit \vH t}- \vec{S}(t/r)^r\ket{\psi}}{\ell_2}  \ge \epsilon \R) \le \delta. 
\end{align*} 
\end{cor}

\subsection{$k$-locality for Fermions}\label{sec:proof_fermion_non_random}
Analogously, we generalize to Hamiltonians with Fermionic terms. We begin with defining $k$-locality for Fermionic systems. Suppose the particle-number preserving Fermionic Hamiltonian can be written as 
\begin{align}
\vH = \sum_{S_-\perp S_+\perp S_z\subset \{m,\cdots,1\}} b_{S_+S_-S_z} \prod_{s_+\in S_+ } \va^\dagger_{s_+} \prod_{s_-\in S_- } \va_{s_-}\prod_{s_z\in S_z }\vO^\eta_{s_z}\label{eq:fermion_k_local}.
\end{align}
Again, particle number preserving enforces $\labs{S_+}=\labs{S_-}$.\footnote{Even worse, odd Fermionic terms, e.g., a single fermion $\va^\dagger_i$, anti-commute with each other even if they have no overlapping sites. }
Recall the second quantization commutation relations (following~\cite{thy_trotter_error})
\begin{align}
        [\va^\dagger,\vO^\eta] &= -\va^\dagger,\\
    [\va,\vO^\eta] &= \va,
\end{align}
and 
\begin{align}
    [\va^\dagger_j\va_k,\va^\dagger_\ell \va_m  ] & = \delta_{kl} \va^\dagger_j\va_m-  \delta_{jm} \va^\dagger_{\ell}\va_k,\label{eq:fermion_fermion}\\
    [\va^\dagger_j\va_k,\vO^{\eta}_\ell ] & = \delta_{kl} \va^\dagger_j\va_\ell- \delta_{j\ell} \va^\dagger_{\ell}\va_k,\\
    [\va^\dagger_j\va_k,\vO^{\eta}_\ell \vO^{\eta}_m ] & = \L(\delta_{kl} \va^\dagger_j\va_\ell- \delta_{j\ell} \va^\dagger_{\ell}\va_k \R)\vO^{\eta}_m+  \vO^{\eta}_{\ell}\L(\delta_{km} \va^\dagger_j\va_m- \delta_{jm} \va^\dagger_{m}\va_k \R).
\end{align}
Compared with $k$-local Paulis, the only difference for the Fermionic case is~\eqref{eq:fermion_fermion}: commuting two Fermionic operators on the same site $\ell$ can produce an identity $\vI_{\ell}$. This would add an extra term in our effective 2-2 norm calculation (Corollary~\ref{lem:22norm_with_alpha})
\begin{align}
    \lambda_{ferm}(k):=\frac{2^{k/2+1}}{(k-1)!} \sum_{k_{f}=1}^k \frac{2^{k_{f}/2}}{{(k-k_{f})!}}\lnormp{\vH}{(k_{f}),2} + \frac{2^{k/2+1}}{(k-1)!} \frac{1}{{k!}}\lnormp{\vH_{ferm}}{(0),2},
\end{align}
where the ``global'' 2-norm $\lnormp{\cdot}{(0),2}$ only contains Fermionic operators
\begin{align}
\vH_{ferm}:= \sum_{\labs{S_-}+\labs{S_+}\ne 0, S_-\perp S_+\perp S_z\subset \{m,\cdots,1\}} b_{S_+S_-S_z} \prod_{s_+\in S_+ } \va^\dagger_{s_+} \prod_{s_-\in S_- } \va_{s_-}\prod_{s_z\in S_z }\vO^\eta_{s_z}.
\end{align}
Intuitively, when identity is produced at the overlapping site, more terms may collide, i.e., add coherently. See Section~\ref{sec:fermion_example} for an example where this term is necessary. Otherwise, the rest of the calculation is identical ($\lambda'(k)$ remains the same). Note that we would use a Fermionic version of Fact~\ref{fact:project_2^S}, which can be shown by a gauge transformation argument. 
\begin{prop}[$k$-local Fermionic Hamiltonians]
\label{prop:Trotter_non_random_Fermionic} 
To simulate a $k$-local, particle number preserving Fermionic Hamiltonian using $\ell$-th order Suzuki formula on m-particle subspace $\vP_m$, the gate complexity 
\begin{align*}
    G =\Omega\L(\L(\frac{\poly(n,m)^{1/p} p^{k/2}\lnormp{\vH}{(global),2} t}{\epsilon}\R)^{1/\ell} \Gamma p^{k/2}(\lnormp{\vH}{(local),2}+\lnormp{\vH_{ferm}}{(0),2})  t\R) \ \ &\textrm{ensures}\ \  \lnormp{\e^{\iunit \vH t}- \vec{S}(t/r)^r}{\bar{p},\vP_{m}} \le \epsilon,
\end{align*} 
where $\lnormp{\vH}{(global),2}, \lnormp{\vH}{(local),2} $ is defined w.r.t to~\eqref{eq:fermion_k_local}.
\end{prop}
\begin{cor}
Draw $\ket{\psi}$ from a $m = \eta n$ - particle subspaces (i.e., $\BE[ \ket{\psi}\bra{\psi}] =\vP_{m}/\tr[\vP_{m}]$), then 
\begin{align*}
    &G =\tilde{\Omega} \L( \L(\frac{\log(\poly(n,m)/\delta)^{k/2}\lnormp{\vH}{(global),2} t}{\epsilon}\R)^{1/\ell} \log(\poly(n,m)/\delta)^{k/2} \Gamma \L(\lnormp{\vH}{(local),2}+\lnormp{\vH_{ferm}}{(0),2}\R) t \R) \\ 
    &\hspace{7cm} \textrm{ensures}\ \  \Pr \L( \lnormp{(\e^{\iunit \vH t}- \vec{S}(t/r)^r)\ket{\psi}}{\ell_2}  \ge \epsilon \R) \le \delta. 
\end{align*}
\end{cor}

\section{Optimality for First-order and Second-order Formulas}\label{sec:nonrandom_optimal}
We demonstrate the optimality of our p-norm estimates for a particular 2-local Hamiltonian, at short times, for the first and second-order Lie-Trotter-Suzuki formulas. The $k\ge 2$ cases can also be constructed analogously. 
Consider the Hamiltonian
\begin{align}
    \vH = \sum_{i>j} \alpha_{ij} \vsigma^z_i\vsigma^z_j + \sum_{i>j} \alpha_{ij} \vsigma^x_i\vsigma^x_j =: \vA+\vB
\end{align}
for the first-order Trotter formula
\begin{align}
    \e^{\ri (\vA+\vB )t} - \e^{\ri \vA} \e^{ \ri\vB t} &= \frac{1}{2}[\vA,\vB]t^2 +\CO(t^3).\\
\end{align}
We can exactly compute its 2-norm due to the orthogonality of Paulis
\begin{align}
    \lnormp{ [\vA,\vB]  }{2}^2&= \lnormp{ \L[\sum_{k>\ell} \alpha_{k\ell} \vsigma^z_k\vsigma^z_{\ell}, \sum_{i>j} \alpha_{ij} \vsigma^x_i\vsigma^x_j\R]}{2}^2\\
    &= \sum_{\{i,j,k\}} \lnormp{\alpha_{ij}\alpha_{jk} \vsigma^z_i \vsigma^y_j \vsigma^x_k }{2}^2 = \sum_{\{i,j,k\}} \alpha_{ij}^2\alpha_{jk}^2.
\end{align}
For our upper bounds~\eqref{eq:nonrandom_gthorder}, 
\begin{align}
\lnormp{\vH}{(global),2}^2 =\sum_{ij} 4\alpha_{ij}^2 \quad\text{and}\quad \lnormp{\vH}{(local),2}^2=\max_i \sum_{j} 4\alpha_{ij}^2,
\end{align}
which means when $\alpha_{ij}=1$ are equal strength,
\begin{align}
    \lnormp{\vH}{(global),2}^2\lnormp{\vH}{(local),2}^2 = \theta\L(  \lnormp{ [\vA,\vB]  }{2}^2\R).
\end{align}
It is less obvious how to calculate its p-norm or operator norm.

To obtain tight p-norm and spectral norm estimates, we construct another Hamiltonian on three set of qubits $\CH = \CH_{S_1}\otimes \CH_{S_2}\otimes \CH_{S_3}$
\begin{align}
    \vH = \sum_{s_1\in S_1, s_2\in S_2} \vsigma^z_{s_1}\vsigma^x_{s_2} + \sum_{s_2\in S_2, i_3\in S_3} \vsigma^y_{s_2}\vsigma^z_{s_3}:= \vA+\vB\label{eq:ZXYZ}.
\end{align}
The commutator evaluates to a factorized commuting sum 
\begin{align}
    [\vA,\vB] &=  \L[\sum_{s_1\in S_1, s_2\in S_2} \vsigma^z_{s_1}\vsigma^x_{s_2}, \sum_{s_2\in S_2, s_3\in S_3} \vsigma^y_{s_2}\vsigma^z_{s_3} \R]\\
    & = 2\sum_{s_1\in S_1, s_2\in S_2, s_2\in S_3} \vsigma^z_{s_1}\vsigma^z_{s_2} \vsigma^z_{s_3} =  2(\sum_{s_1\in S_1} \vsigma^z_{s_1}) \cdot (\sum_{s_2\in S_2} \vsigma^z_{s_2} )\cdot (\sum_{s_3\in S_3} \vsigma^z_{s_3}) .
\end{align}
Its p-norms can be obtain by central limit theorem at large $\labs{S_1}, \labs{S_2}, \labs{S_3}$

\begin{align}
    \lnormp{ \sum_{s_1\in S_1} \vsigma^z_{s_1} }{p}= \Omega( \sqrt{p\labs{S_1}} ) \lnormp{\vI}{p},
\end{align}
where we recall the p-th moment of standard Gaussian $\labs{g}_p=\theta(\sqrt{ p } )$.
Now, let $\labs{ S_1}=\labs{ S_2}=\labs{ S_3}=\theta(n)$, then it saturates our first-order p-norm upper bound~\eqref{eq:nonrandom_gthorder}.
\begin{align}
    \lnormp{ [\vA,\vB]}{p} &=\Omega(\sqrt{pn})^3 \lnormp{\vI}{p}\\
    \sqrt{C_p}^3\lnormp{\vH}{(global),2}\lnormp{\vH}{(local),2} \lnormp{\vI}{p} &= \CO\L( \sqrt{p}^3 \cdot \sqrt{ n^2} \cdot \sqrt{n}\R) \lnormp{\vI}{p}.
\end{align}
At the same time, its spectral norm 
\begin{align}
    \norm{ [\vA,\vB] } &= \theta( n^3 ) = \lnormp{\vH}{(global),1}\lnormp{\vH}{(local),1}
\end{align}
matches the triangle inequality bound in~\cite{thy_trotter_error}.

\subsection{Second-order Suzuki Formulas}
For the second-order Trotter error, recall the expansion~\cite[Appendix~L]{thy_trotter_error},
\begin{align}
     \e^{\ri (\vA+\vB )t} - \e^{\ri \vA t/2} \e^{ \ri\vB t} \e^{\ri \vA t/2}= -\frac{\ri}{12} \L([\vB,[\vB,\vA ]] - \frac{1}{2}[\vA,[\vA,\vB]]\R)t^3 +\CO(t^4)
\end{align}
with the same Hamiltonian~\eqref{eq:ZXYZ}. Due to the symmetry, we know $[\vB,[\vB,\vA ]]$ has the same p-norm as $[\vA,[\vA,\vB ]]$. Conveniently, the factor $\frac{1}{2}$ allows us to consider only one term (at most losing a constant overhead $\frac{1}{2}$)
\begin{align}
    [\vB,[\vB,\vA ]] &=  -4\sum_{s_1\in S_1, s_2\in S_2, s_3, s'_3\in S_3} \vsigma^z_{s_1}\vsigma^x_{s_2} \vsigma^z_{s_3} \vsigma^z_{s'_3}.
\end{align}
This converges to a function of three independent Gaussians (note that the $s_3$, $s'_3$ are two dummy indexes in the same set $S_3$)
\begin{align}
    \lnormp{[\vB,[\vB,\vA ]]}{p} &= \Omega( \labs{g_1g_2g_3^2}{p} ) \lnormp{\vI}{p} = \Omega(\sqrt{pn})^4 \lnormp{\vI}{p},\\
   \sqrt{C_p}^4\lnormp{\vH}{(global),2}\lnormp{\vH}{(local),2}^2 \lnormp{\vI}{p} &= \CO\L( \sqrt{p}^4 \cdot \sqrt{ n^2} \cdot \sqrt{n}^2\R) \lnormp{\vI}{p},
\end{align}
matching our p-norm bound.
The spectral norm
\begin{align}
    \norm{ [\vB,[\vB,\vA ]]} &= \theta( n^4 ) = \lnormp{\vH}{(global),1}\lnormp{\vH}{(local),1}^2
\end{align}
again matches the triangle inequality bounds in~\cite{thy_trotter_error}.

\subsection{Fermionic Hamiltonians}\label{sec:fermion_example}
To demonstrate the need for the extra term for Fermionic Hamiltonians $\lnormp{\vH_{ferm}}{(0),2}$, consider a Hamiltonian of the form 
\begin{align}
    \vH = \sum_{s_1\in S_1, s_2\in S_2} \va_{s_1}\va^\dagger_{s_2}+ \va^\dagger_{s_1}\va_{s_2} + \sum_{s_2\in S_2, s_3\in S_3} \va_{s_2}\va^\dagger_{s_3} +\va^\dagger_{s_2}\va_{s_3}:= \vA+\vB.
\end{align}
The commutator evaluates to 
\begin{align}
    [\vB,\vA] & = \sum_{s_1\in S_1, s_2\in S_2, s_2\in S_3}\va_{s_1} \va^\dagger_{s_3} -\va^\dagger_{s_1} \va_{s_3}\\ 
    \lnormp{[\vA,\vB]}{2} &= 2 \sqrt{\labs{S_1}} \labs{S_2} \sqrt{\labs{S_3}} = \theta( \lnormp{\vH}{(global),2}\lnormp{\vH}{(global),2}).
\end{align}
And for the second-order Suzuki,
\begin{align}
    [\vB,[\vB,\vA] ] & = -\labs{S_2}\cdot \sum_{s_1\in S_1, s_2\in S_2, s_2\in S_3} \va^\dagger_{s_2} \va_{s_1} +  \va_{s_1}\va^\dagger_{s_2}\\ 
    \lnormp{[ \vB,[\vB,\vA]]}{2} &= 2 \sqrt{\labs{S_1}} \labs{S_2}^2 \sqrt{\labs{S_3}} = \theta( \lnormp{\vH}{(global),2}^2\lnormp{\vH}{(global),2}). 
\end{align}

\section{Preliminary: matrix-valued martingales}\label{sec:martingales}
\textit{Concentration inequalities} are well known for an i.i.d. sum of random numbers. 
Unfortunately, the phenomena in the wild are rarely like a sum, identical, or independent yet nonetheless concentrate around the mean. Among the zoo of extensions that attempt to capture realistic randomness, a (scalar-valued) \textit{martingale} describes a random process that the future has zero mean conditioned on the past. Martingales constitute a class more flexible than i.i.d. sums that will serve our purpose.

 For a minimal technical introduction (following Tropp~\cite{tropp2011freedmans} and Huang et. al~\cite{HNTR20:Matrix-Product}), consider a filtration of the master sigma algebra $\CF_0\subset \CF_1 \subset \CF_2 \cdots \subset \CF_t \subset \cdots \CF$, where for each filtration $\CF_j$ we denote the conditional expectation $\BE_j$. Intuitively, we can think of the index $t$ as the 'time', where the associated filtration $\CF_t$ hosts possible events happening before time $t$. More precisely, a martingale is a sequence of random variable $Y_t$ adapted to the filtration $\CF_t$ such that
\begin{align}
    \sigma(Y_t) &\subset \CF_t &\textrm{(causality)},\\
    \BE_{t-1} Y_t &= Y_{t-1} &\textrm{(status quo)}.
\end{align}
In other words, the present depends on the past ('causality'), and tomorrow has the same expectation as today ('status quo'). For simplicity, we often subtract the mean to obtain a \textit{martingale difference} sequence 
\begin{align}
    \BE_{t-1} D_t = 0  \quad \text{where}\quad D_t:=Y_t-Y_{t-1}.
\end{align}

\subsection{Useful Norms and Recursive Bounds for Matrices}
In our case, our goal is to quantify the error between the ideal unitary $\vU=\e^{\iunit \vH t}$ and the product formula $\vec{S}$ where the Hamiltonian is drawn randomly. This can be framed as a \textit{matrix-valued} martingale
\begin{align}
    \sigma(\vY_t) &\subset \CF_t,\\
    \BE_{t-1} \vY_t &= \vY_{t-1},
\end{align}
where the conditional expectation $\BE_{t-1}$ acts \textit{entrywise}. In other words, the randomness here has both classical (the expectation $\BE$) and quantum (the trace $\tr[\cdot]$) sources. In comparison, our previous discussion on Paulis strings does not have the above extra layer of classical randomness. This will give slightly different flavors.

Historically, the earliest general results on matrix-valued martingales were established in~\cite{lust_piquard_86, lust_Pisier_91,Pisier_1997}, and more recent works and applications include~\cite{tropp2011freedmans,oliveira2010concentration,HNTR20:Matrix-Product, jungeZenq_nc15, chen2020quantum,chen2021concentration}. Throughout this work, our main driving horse is \textit{again} uniform smoothness (in a slightly different format from uniform smoothness for subsystems (Proposition~\ref{prop:unif_subsystem_recap})). It is not the tightest kind of martingale inequality but arguably the simplest and most robust when matrices are bounded (or with Gaussian coefficients via the central limit theorem). Analogously to Proposition~\ref{prop:unif_subsystem_recap}, these inequalities deliver sum-of-square (``incoherent'') estimates sharper than the triangle inequality, which is linear (``coherent'').

To study concentration of matrices, we first pick a suitable norm. The error between the ideal unitary and the product formula can be quantified in two ways with different operational meanings. For both norms, uniform smoothness streamlines our concentration results (Section~\ref{sec:first_order_random}, Section~\ref{sec:matrix_poly}). 

\subsubsection{The operator norm}
The operator norm quantifies the error for the worst input state
\begin{align}
 \lnorm{\vec{U}-\vec{S}}:=\sup_{\ket{\psi}}\lnormp{(\vec{U}-\vec{S})\ket{\psi}}{\ell_2}.
\end{align}
If we are interested in concentration of the operator norm, it suffices to control its moments by the \textit{expected Schatten p-norm} \begin{align}
    (\BE \lV \vY\rV^p )^{1/p}\le (\BE \lV \vY\rV^p_p)^{1/p}=:\vertiii{\vY}_p.
\end{align}
To bound the RHS, the driving horse is the following bound with only a martingale requirement (``conditionally zero-mean'').
\begin{fact}[{Uniform smoothness for Schatten classes~\cite[Proposition~4.3]{HNTR20:Matrix-Product}}] \label{fact:sub_average_pq}
Consider random matrices $\vX, \vY$ of the same size that satisfy
$\BE[\vY|\vX] = 0$. When $2 \le p$,
\begin{equation}
\vertiii{\vX+\vY}_{p}^2 \le \vertiii{\vX}_{p}^2  + C_p\vertiii{\vY}_{p}^2. 
\end{equation}
The constant $C_p = p - 1$ is the best possible.
\end{fact}
Uniform smoothness for Schatten classes in another form (Fact~\ref{fact:unif_schatten}) was proven by~\cite{Tomczak1974} with optimal constants determined by~\cite{BallCL_optimal_unif_smooth}. The above martingale form is due to~\cite{naor_2012,ricardXu16} and~\cite[Proposition~4.3]{HNTR20:Matrix-Product}. This can be alternatively seen as a special case of Proposition~\ref{prop:unif_subsystem_recap} by interpreting the classical expectation as a trace.\footnote{The condition $\BE[\vY|\vX] = 0$ is a special case of the general operator algebra formulation using the conditional expectation superoperator $E[\vY]=0, E[\vX]=\vX$.}

\subsubsection{Fixed input state}
Sometimes we only care about a \textit{fixed but arbitrary} input state $\ket{\psi}$. This deserves another error metric (following~\cite{chen2020quantum}) that differs from the spectral norm by an order of quantifier
\begin{align}
    \vertiii{\vY}_{\textrm{fix},p} :=\sup_{\ket{\psi}}\left(\BE  \lV \vY\ket{\psi}\rV_{\ell_2}^p \right)^{1/p} \ne \left(\BE \sup_{\ket{\psi}} \lV \vY\ket{\psi}\rV_{\ell_2}^p \right)^{1/p} = (\BE \lV \vY\rV^p )^{1/p}.
\end{align}

Uniform smoothness for this norm follows.
\begin{cor}[{Uniform smoothness, fixed input}~\cite{chen2021concentration}] \label{fact:sub_average_2q_DP}
Consider random matrices $\vX, \vY$ of the same size that satisfy
$\BE[\vY|\vX] = 0$. When $2 \le p$,
\begin{align}
\vertiii{\vX+\vY}_{\textrm{fix},p}^2 &= \vertiii{\vX}_{\textrm{fix},p}^2+C_p\vertiii{\vY}_{\textrm{fix},p}^2
\end{align}
with constant $C_p = p - 1$. 
\end{cor} 
\begin{proof}
This can be seen by rewriting the $\ell_2$-norm as a p-norm
\begin{align}
    \vertiii{\vX+\vY}_{\textrm{fix},p}^2 &= \sup_{rank(\vec{P})=1}\vertiii{\vX\vec{P}+\vY\vec{P}}_{p}^2\\
&\le\sup_{rank(\vec{P})=1}\left( \vertiii{\vX\vec{P}}_{p}^2  + C_p\vertiii{\vY\vec{P}}_{p}^2\right)\\
&\le \sup_{rank(\vec{P})=1}\vertiii{\vX\vec{P}}_{p}^2 + C_p\sup_{rank(\vec{P})=1}\vertiii{\vY\vec{P}}_{p}^2.
\end{align}
\end{proof}

Note that the pure inputs $\ket{\psi}$ capture general mixed inputs $\vrho$ by convexity
\begin{align}
    \frac{1}{2}\left(\sup_{\vrho}\BE \lnormp{ \vec{U}\vrho\vec{U}^\dagger -\vec{S}\vrho\vec{S}^\dagger }{1}^p \right)^{1/p}&\le \frac{1}{2}\left(\sup_{\ket{\psi}}\BE \lnormp{ \vec{U}\ket{\psi}\bra{\psi}\vec{U}^\dagger -\vec{S}\ket{\psi}\bra{\psi}\vec{S}^\dagger }{1}^p \right)^{1/p}\\
    &\le \sup_{\ket{\psi}} \left(\BE \lnormp{ (\vec{U}-\vec{S})\ket{\psi}\bra{\psi}}{1}^p \right)^{1/p} = \sup_{\ket{\psi}} \left(\BE \lnormp{ (\vec{U}-\vec{S})\ket{\psi}}{\ell_2}^p \right)^{1/p}.
\end{align}
The second inequality is a telescoping sum. The third equality uses that the operator norm equals to the 1-norm $\norm{\cdot}=\lnormp{\cdot}{1}$ for rank 1 matrices.

\subsection{Reminders of Useful Facts }
Before we turn to the proof, let us remind ourselves of the useful properties for the underlying norms $\vertiii{ \cdot}_{*}:=\vertiii{ \cdot}_{p,q}, \vertiii{ \cdot}_{\textrm{fix},p}$ for $p,q\ge 2$. They are largely inherited from the (non-random) Schatten p-norm. Following~\cite{chen2021concentration}, 
\begin{fact}[non-commutative Minkowski]\label{fact:non-commutative_mink}
Each of the expected moments satisfies the triangle inequality and thus is a valid norm. For any random matrix $\vX, \vY$
\begin{align}
\vertiii{\vX+\vY}_{*} \le \vertiii{\vX}_{*}+\vertiii{\vY}_{*}.
\end{align}
\end{fact}
\begin{fact}[operator ideal norms]\label{fact:operator ideal}
For operators $\vA$ deterministic and $\vec{X}$ random
\begin{align}
\vertiii{\vec{A}\vec{X}}_{*}, \vertiii{\vec{X}\vec{A}}_{*} \le \lV \vec{A}\rV \cdot \vertiii{\vec{X}}_{*}.
\end{align}
\end{fact}
\begin{fact}[unitary invariant norms]
For $\vec{U}, \vec{V}$ deterministic unitaries and random operator $\vX$
\begin{align}
\vertiii{\vec{U}\vec{X}\vec{V}}_{*} = \vertiii{\vec{X}}_{*}.
\end{align}
\end{fact}
Being operator ideal already implies unitary invariance, but we state it regardless. As the norm $\vertiii{\cdot}_{\textrm{fix},p}$ defined for low-rank input is somewhat non-standard, we include a proof as follows.
\begin{proof}[Proof of Fact~\ref{fact:operator ideal} for fixed inputs]
The case $\vertiii{\vA\vX}_{\textrm{fix},p}$ follows from the fact that p-norms are operator ideal. For the other ordering,
\begin{align}
\vertiii{\vX\vec{A}}_{\textrm{fix},p}&=\sup_{rank(\vec{P})=1}(\BE[\lV \vX\vec{A}\vec{P}\rV_p^p])^{1/p}\\
&=\sup_{rank(\vec{P})=1}(\BE[\lV \vX\vec{P}'\vec{A}'\rV_p^p])^{1/p}\\
&=\sup_{rank(\vec{P}')=1}(\BE[\lV \vX\vec{P}'\rV_p^p])^{1/p} \lV \vec{A}\rV\\
&=\vertiii{\vX}_{\textrm{fix},p} \lV \vec{A}\rV.
\end{align}
In the second line, we use the singular value decomposition 
\begin{align}
    \vec{A}\vec{P} = \vec{U}\vec{S}\vec{V}= \vec{U}\vec{S}_1 \vec{S}_{\vec{A}'}\vec{V} = \vec{U}\vec{S}_{1}\vec{U}^\dagger \cdot \vec{U}\vec{S}_{\vec{A}'}\vec{V} := \vec{P}'\vec{A}',
\end{align} where we rewrite the diagonal elements as product $\vec{S}=\vec{S}_{1}\vec{S}_{\vec{A}'}$, where $\vec{S}_1$ is a rank $1$ projector and $\lV \vec{S}_{\vec{A}'}\rV \le \lV \vec{S}\rV \le \lV \vec{A}\rV$. This is possible since $\vec{S}$ is bounded by $\lV \vec{S}\rV \le \lV \vec{P}\vec{A}\rV \le \lV \vec{A}\rV$. This is the advertised result.
\end{proof}

\section{First-order Trotter for random Hamiltonians}\label{sec:first_order_random}
In this section, we employ matrix martingales techniques on the first-order Lie-Trotter formula for random Hamiltonians.\footnote{The higher-order formulas require the more general analysis of random matrix polynomials; see Section~\ref{sec:matrix_poly}. } Recall, it suffices to control the Trotter error represented in the exponentiated form~\cite{thy_trotter_error}
\begin{align}
    \e^{\ri\vH_\Gamma t} \cdots \e^{\ri\vH_1} = \exp_{\CT}\L(\ri\int (\vec{\CE} (t)+ \vH)dt\R)\quad \text{where}\quad
    \vec{\CE} (t) := \sum^{\Gamma}_{k=2}\left( \prod_{\gamma=k-1}^{1} \e^{\CL_\gamma t} [\vH_k] -\vH_k\right).
\end{align}

\begin{thm}[First-order Trotter for random Hamiltonians] \label{thm:first_order_gate_count}
Consider a random Hamiltonian $ \vH= \sum_{\gamma=1}^\Gamma \vH_\gamma$ on $n$-qubits, where each term $\vH_\gamma$ is independent, zero mean,  and almost surely bounded 
    \begin{align}
        \BE \vH_\gamma =0\quad \text{and}\quad\norm{\vH_\gamma} \le b_\gamma.
    \end{align}
Then, the gate count
\begin{align*}
    G =  2\sqrt{2}\L(n\ln(2)+\log(\e^2/\delta)\R) \Gamma\lnormp{\vH}{(global),2} \lnormp{\vH}{(local),2}\frac{ t^2}{\epsilon}
    \ \ \textrm{ensures}\ \ \Pr\L(\norm{e^{\iunit \vH t}- \vec{S}_1(t/r)^r} \ge \epsilon\R) \le \delta. 
\end{align*}
For arbitrary fixed input state $\vrho$, the gate count
\begin{align*}
    G =  2\sqrt{2}\log(\e^2/\delta) \Gamma\lnormp{\vH}{(global),2} \lnormp{\vH}{(local),2}\frac{t^2}{\epsilon} 
    \ \ \textrm{ensures}\ \ \Pr\left(\frac{1}{2}\lnormp{(\e^{-\iunit \vH t}\vrho \e^{\iunit \vH t}- \vec{S}_1(t/r)^{\dagger r}\vrho \vec{S}_1(t/r)^r)}{1} \ge \epsilon\right) \le \delta.
\end{align*}
\end{thm}

We see that the gate counts depend on the 2-norm quantities $ \lnormp{\vH}{(global),2}:=\sqrt{\sum_{\gamma} b_\gamma^2}$ and $\lnormp{\vH}{(local),2}:=\max_{i}  \sqrt{\sum_{\gamma: i \subset \gamma } b^2_\gamma}$, but differ by the logarithm of the dimension $\log(d^n)$. 
Often, the Hamiltonian we encounter has gaussian coefficients. By a central limit theorem, we may quickly obtain an analogous result. 
\begin{cor}[Gaussian coefficients]\label{cor:first_order_gaussian}
Theorem~\ref{thm:first_order_gate_count} also holds for random Hamiltonian where each term $\vH_{\gamma}$ is a deterministic bounded matrix with i.i.d. standard Gaussian coefficients
    \begin{align}
        \vH_\gamma = g_\gamma \vK_\gamma\quad \text{and}\quad \norm{\vK_\gamma} \le b_\gamma.
    \end{align}
\end{cor}

For a concrete gate complexity, we evaluate Theorem~\ref{thm:first_order_gate_count} on all-to-all interacting (SYK-like) models on $n$-qudits,
\begin{align}
\vH = \sum_{\gamma} g_{\gamma} \vK_{\gamma} \quad \text{where}\quad \norm{\vK_{\gamma}} \le J\sqrt{\frac{(k-1)!}{kn^{k-1}} },\quad
    \lnormp{\vH}{(global),2}^2 \le  
    \frac{J^2n}{k^2} \quad  \text{and}\quad 
    \lnormp{\vH}{(local),2}^2 \le  
    J^2
\end{align}
 with $\Gamma\le n^k/k!.$ 
\begin{cor}[First-order Trotter for SYK models]
\begin{align*}
    G &=  \frac{2\sqrt{2}}{k\cdot k!}\L(n\ln(d)+\log(\e^2/\delta)\R) \frac{n^{k+1/2} (Jt)^2}{\epsilon} \quad &\text{(worst inputs)}\\
    G &=  \frac{2\sqrt{2}}{k\cdot k!}\log(\e^2/\delta) \frac{n^{k+1/2} (Jt)^2}{\epsilon}\quad &\text{(fixed input)}.
\end{align*}
\end{cor}
The proof of Theorem~\ref{thm:first_order_gate_count} is mainly controlling the integrand $\CE(t)$ via matrix martingale techniques, summarized in the following lemma.
\begin{lem}
\label{lem:first_Trotter_error} For both p-norms $\vertiii{\cdot}_*=\vertiii{\cdot}_p$ and $\vertiii{\cdot}_{\textrm{fix},p}$,
\begin{align*}
    \vertiii{\vCE(t)}_{*}^2
    &\le 2\vertiii{\vI}^2_*C_p\sum^{\Gamma-1}_{k=1}\left[ 4C_pt^2\sum^{k-1}_{n=1}\labs{\big\lV[\vH_n,\vH_k]\big\rV_\infty}_p^2+ \left(\sum^{k-1}_{n=1}  \frac{t^2}{2}\labs{\left\lV[\vH_n,[\vH_n,\vH_k]]\right\rV_\infty}_p \right)^2 \right].
\end{align*}
\end{lem}
Given such a bound, we may quickly convert to the advertised estimates.
\begin{proof}[Proof of Theorem~\ref{thm:first_order_gate_count}]
For a total evolution time $t$, repeat the Trotter formula for $r$ rounds with individual duration $\tau= t/r$.
Assuming Lemma~\ref{lem:first_Trotter_error}, each round has an error 
\begin{align}
    \vertiii{\vec{\CE} (\tau)}_*^2 &\le 2\vertiii{\vI}^2_*C_p\sum^{\Gamma-1}_{k=1}\left[ 16C_p b^2_k\tau^2\lnormp{\vH}{(local),2}^2+ \left(  b_k\frac{\tau^2}{2} \lnormp{\vH}{(local),2}^2 \right)^2 \right]\\
    &\le 32\vertiii{\vI}^2_*p^2 \tau^2 \lnormp{\vH}{(global),2}^2  \lnormp{\vH}{(local),2}^2. \label{eq:bound_high_order}
\end{align}
The last inequality simplifies the subleading term by the crude estimate $\frac{\tau}{4} \lnormp{\vH}{(local),2}\le 1$, which will be verified. To control the total Trotter error, integrate along time $\tau$ and invoke a telescoping sum,
\begin{align}
    \vertiii{e^{\iunit \vH t}- \vec{S}_1(t/r)^r}_* =\vertiii{\vec{\CE}_{tot}(t)}_*  &\le 2\sqrt{2}\vertiii{\vI}_*p  \frac{t^2}{r} \lnormp{\vH}{(global),2} \lnormp{\vH}{(local),2}=:\lambda  \vertiii{\vI}_*p.
\end{align}
To obtain concentration, it remains to optimize the moment $p$ for Markov's inequality.\\ 
(i) For the spectral norm, set $\vertiii{\cdot}_{*}= (\BE \norm{\cdot}_p^p)^{1/p}$
\begin{align}
    \Pr(\norm{\vec{\CE}_{tot}}\ge \epsilon)\le \frac{\BE \norm{\vec{\CE}_{tot}}^p}{\epsilon^p}&\le \frac{\BE \lnormp{\vec{\CE}_{tot}}{p}^p}{\epsilon^p} \\
    &\le D \left(p\frac{\lambda}{\epsilon}\right)^p\\
    &\le D \exp(-\frac{\epsilon}{\lambda}+2).
\end{align}
The factor of dimension $D=\norm{\vI}_p^p= 2^n$ is due to the trace and the offset $+2$ accounts for the constraint $p\ge 2$. To ensure the Trotter error is at most $\epsilon$ with failure probability $\delta$, we demand $\lambda \le \tfrac{\epsilon}{\log(\e^2 D/\delta)}$, which is
\begin{align}
    G = \Gamma r = 2\sqrt{2} \log(\e^2 D/\delta) \Gamma \frac{t^2}{\epsilon} \lnormp{\vH}{(global),2} \lnormp{\vH}{(local),2}.
\end{align}

(ii) For arbitrary fixed inputs, the factor $\log(D)$ disappears since $\norm{\vI}_{\textrm{fix},p}^p=\sup_{rank(\vP)=1}\norm{\vP}_{p}^p=1$. We arrive at the gate count
\begin{align}
    G =\Gamma r \ge 2\sqrt{2} \log(\e^2/\delta) \Gamma \frac{t^2}{\epsilon} \lnormp{\vH}{(global),2} \lnormp{\vH}{(local),2}
\end{align}
which already improves over qDRIFT~\cite{campbell2019random}. Lastly, for a consistency check, the choices of $r$ in both calculations (i) and (ii) guarantee that  
\begin{align}
    \tau^2 \lnormp{\vH}{(local),2} \lnormp{\vH}{(local),2} \le \frac{\epsilon}{2\sqrt{2} \log(\e^2/\delta)}   \le 16,
\end{align}
which is what we needed for \eqref{eq:bound_high_order}. 

\end{proof}
The above result for bounded random matrices quickly extends to those with Gaussian coefficients by the central limit theorem. 
\begin{proof}[Proof of Corollary~\ref{cor:first_order_gaussian}]
Representing Gaussian by a sum of i.i.d. Rademachers
\begin{align}
    \vH_\gamma=g_\gamma\vK_\gamma = (\lim_{N\rightarrow\infty}\sum_j^N \frac{\epsilon_{\gamma,j} }{\sqrt{N}}) \vK_\gamma:=\lim_{N\rightarrow\infty}\sum_j^N \vY_{\gamma,j},
\end{align}
we obtain a Hamiltonian as sum over bounded, zero mean summands
\begin{align}
   \vH = \sum_{\gamma=1}^\Gamma \vH_\gamma = \sum_{j=1}^N \sum_{\gamma=1}^\Gamma \vY_{\gamma,j}= \sum_{\gamma'=1}^{\Gamma'} \vY_{\gamma'}:= \vH'
\end{align}
where we use notation $ \vH'$ for the summand $\vec{Y_{\gamma'}}$. Plug in Theorem~\ref{thm:first_order_gate_count} and evaluate the 2-norm quantities
\begin{align}
    \norm{ \vH'}^2_{(0),2} &=\sqrt{\sum_{\gamma} b_\gamma^2} \\
    \norm{ \vH'}_{(1),2}^2 &= \max_{i}  \sqrt{\sum_{\gamma: i \subset \gamma} b^2_\gamma}.
\end{align}
This is the advertised result.
\end{proof}

\subsection{Proof of Lemma~\ref{lem:first_Trotter_error}}
It remains to prove Lemma~\ref{lem:first_Trotter_error} for the random Hamiltonian with bounded summand. We will use the martingale structure twice.
\begin{proof}
Recall
\begin{align}
    \vec{\CE} (t) = \sum^{\Gamma}_{k=2}\left( \prod_{\gamma=k-1}^{1} \e^{\CL_\gamma t} [\vH_k] -\vH_k\right)
\end{align}
 and observe the martingale property for summand\footnote{This evident martingale structure is unique to first-order Trotter. At higher orders, such a martingale structure is lost; we have to consider general polynomials.} 
\begin{align}
    \BE_{k-1}\L[ \prod_{\gamma=k-1}^{1} \e^{\CL_\gamma t} [\vH_k] -\vH_k\R] =0 \quad \text{for each}\quad k =1, \cdots, \Gamma.
\end{align}
Indeed, the terms $\gamma=k-1,\cdots, 1$ in the exponential are independent of $\vH_k$. By uniform smoothness, the martingale difference sequence is bounded by a sum-of-squares
\begin{align}
    \vertiii{\vec{\CE} (t)}^2_* &\le \vertiii{\sum^{\Gamma-1}_{k=2}\left( \prod_{\gamma=k-1}^{1} \e^{\CL_\gamma t} [\vH_k] -\vH_k \right)}^2 + C_p\vertiii{\prod_{\gamma=\Gamma-1}^{1} \e^{\CL_\gamma t} [\vH_\Gamma] -\vH_\Gamma}^2_{*}\\
    & \le C_p\sum^{\Gamma-1}_{k=1} \vertiii{ \prod_{\gamma=k-1}^{1} \e^{\CL_\gamma t} [\vH_k] -\vH_k }^2_*.
\end{align}

Next, we further massage each term to identify (yet another) martingale difference apart from the 'bias.' For each $k$, consider a telescoping sum
\begin{align}
    \prod_{\gamma=k-1}^{1} \e^{\CL_\gamma t} [\vH_k] -\vH_k &= \sum^{k-1}_{n=1} \prod_{\gamma=n-1}^{1} \e^{\CL_\gamma t}(\e^{\CL_{n} t}-I) [\vH_k]\\
    &= \sum^{k-1}_{n=1}\left( \prod_{\gamma=n-1}^{1} \e^{\CL_\gamma t}(\e^{\CL_{n} t}-I) [\vH_k] -  \prod_{\gamma=n-1}^{1} \e^{\CL_\gamma t}\BE_{n-1}(\e^{\CL_{n} t}-I) [\vH_k]\right) \quad& \text{(the difference)}\\
    &+\sum^{k-1}_{n=1} \prod_{\gamma=n-1}^{1} \e^{\CL_\gamma t}\BE_{n-1}(\e^{\CL_{n} t}-I) [\vH_k]\quad &\text{(the bias)}\\
    &:= \sum^{k-1}_{n=1}  \vec{D}_n +  \vec{B}_n\quad \text{where} \quad \BE_{n-1}  \vec{D}_n =0. 
\end{align}
The dominant source of error comes from the martingale difference sequence $\vec{D}_n$, which features the desired sum-of-squares behavior. The bias term is treated later.
\begin{align}
    \vertiii{ \prod_{\gamma=k-1}^{1} \e^{\CL_\gamma t} [\vH_k] -\vH_k }^2_* = \vertiii{\sum^{k-1}_{n=1}  \vec{D}_n +  \vec{B}_n }^2_* 
    & \le 2 \vertiii{\sum^{k-1}_{n=1}  \vec{D}_n }^2_* +2 \vertiii{\sum^{k-1}_{n=1}  \vec{B}_n }^2_*\\
    & \le 2 C_p \sum^{k-1}_{n=1}\vertiii{  \vec{D}_n }^2_* +2 \vertiii{\sum^{k-1}_{n=1}  \vec{B}_n }^2_*.
\end{align}
The first inequality is elementary $(a+b)^2 \le 2a^2+2b^2$ and the last inequality uses uniform smoothness. It remains to evaluate both terms. Compute the summand of the first term
\begin{align}
    \vertiii{ \vec{D}_n}^2_* &=  \vertiii{\prod_{\gamma=n-1}^{1} \e^{\CL_\gamma t}(\e^{\CL_{n} t}-I) [\vH_k] -  \prod_{\gamma=n-1}^{1} \e^{\CL_\gamma t}\BE_{n-1}(\e^{\CL_{n} t}-I) [\vH_k]}_*^2\\
    &\le \vertiii{\left((\e^{\CL_{n} t}-I)- \BE_{n-1}[e^{\CL_{n} t}-I]\right)[\vH_k]}_*^2\\
    &\le 4\vertiii{\left(\e^{\CL_{n} t}-I\right) [\vH_k]}_*^2\\
    &\le 4t^2 \labs{\big\lV[\vH_n,\vH_k]\big\rV_\infty}_p^2  \vertiii{\vI}_*^2
\end{align}
where the factor of $2^2$ is due to convexity of $ \vertiii{\cdot}_*^2$. The bias term cannot be treated as martingales; apply a crude triangle inequality
\begin{align}
\vertiii{\sum^{k-1}_{n=1} \vec{B}_n}^2_*&= \vertiii{\sum^{k-1}_{n=1} \prod_{\gamma=n-1}^{1} \e^{\CL_\gamma t}\BE_{n-1}(\e^{\CL_{n} t}-I) [\vH_k]}^2_*\\
&\le \left(\sum^{k-1}_{n=1}  \frac{t^2}{2}\labs{\left\lV[\vH_n,[\vH_n,\vH_k]]\right\rV_\infty}_p \right)^2\vertiii{\vI}^2_*.
 \end{align}
Fortunately, it is at high-orders $\CO(t^2)$ and thus subleading. Combining the two terms yields the advertised result.
\end{proof}

\section{Preliminary: Concentration for multivariate polynomials}\label{sec:matrix_poly}
This section develops concentration inequalities for multivariate polynomials of independent random matrices. This will prepares us for the proof of higher-order Trotter error for random Hamiltonians (Section~\ref{chap:random}).

\subsection{Scalars}\label{sec:scalar_hyper}
For a polynomial of independent scalars, the general results are relatively new and multifaceted~\cite{kim_vu,Lata_a_2006,schudy2012polyconcentration}. The problem is better understood for Rademachers and Gaussians, captured in the form of Hypercontractivity~\cite{odonnell2021analysis,janson_1997}.
As in Section~\ref{sec:prelim_hyper}, it relates the p-norm $\labs{f}_{p}:= ( \BE [ \labs{f}^p] )^{1/p}$ to the 2-norm, i.e., the typical fluctuation is well-captured by the variance.
\begin{fact}[Hypercontractivity for Rademacher polynomial~\cite{odonnell2021analysis}]
Consider a degree-r polynomial of Rademachers 
\begin{align}
f(z_m,\cdots,z_1)= \sum_{S\subset \{1,\cdots,m\}} f_{S} \prod_{s\in S} z_s \quad \text{where} \quad z_s =\pm 1.
\end{align}
For $p \ge 2$, 
\begin{align}
\labs{f}_{p} \le \labs{\sum_{S} \sqrt{C_p}^{\labs{S}} f_{S} \prod_{s\in S} z_s}_{2} \le \sqrt{C_p}^{r}\labs{f}_{2}.
\end{align}
\end{fact}
\begin{fact}[Hypercontractivity for a polynomial of independent Gaussians{~\cite[Theorem~6.12]{janson_1997}}]\label{fact:hyper_scalar_g}
Conisder a degree-r polynomial of i.i.d. Gaussian variables $f(g_m,\cdots,g_1)$. For $p \ge 2$, 
\begin{align}
    \labs{f}_{p} \le \sqrt{C_p}^{r}\labs{f}_{2}.
\end{align}
\end{fact} 
We do not present the intermediate bound for the Gaussian case because it requires expansion by the orthogonal Hermite polynomials, which complicates the picture. Note that we can WLG assumed the above to have zero mean.
\subsection{Matrices}
Unlike the scalar cases, concentration for a multivariate polynomial of matrices is relatively unexplored; even the i.i.d. cases are fairly modern (See, e.g.,~\cite{tropp2015introduction}) and there it remains what the appropriate matrix analog quantity (such as the variance) is. For multivariate polynomials, the problem seems too general in terms of how matrices may interact with each other and how randomness is involved. 

Nevertheless, we will derive concentration results that arguably match the best-known scalar results. What enables this is that we specialize in polynomials of \textit{bounded} matrices with Gaussian coefficients, motivated by concrete applications in physics and quantum information (e.g., Hamiltonian with Paulis strings with Gaussian coefficients). 

As we discussed in Section~\ref{sec:prelim_hyper}, we use the ``local'' uniform smoothness inequality recursively to derive the ``global'' concentration for multivariate matrix polynomials.
However, the external classical randomness will require slightly different arguments and presentations. We begin with a result being essentially the analog of Hypercontractivity we showed (Proposition~\ref{prop:general_pauli_expansion}). Unless otherwise noted, the norms in this section will be overloaded
\begin{align}
    \vertiii{\cdot}_{*} = \vertiii{\cdot}_{p} \quad \text{or}\quad \vertiii{\cdot}_{fix,p}.
\end{align}
Uniform smoothness holds for both norms (Fact\ref{fact:sub_average_pq}, Fact~\ref{fact:sub_average_2q_DP}).

\begin{prop}[Concentration for matrix function]\label{prop:general_matrix_function}
For a matrix-valued function $\vF(\vX_m,\cdots, \vX_1)$, with matrix-valued variables $\vX_i$,
\begin{align}
    \vertiii{\vF(\vX_m,\cdots, \vX_1)}^2_* \le \sum_{S\subset \{m,\cdots,1\}} (C_p)^{|S|}\vertiii{\prod_{s\in S}(1-\BE_s)\prod_{s'\in S^c}(\BE_{s'})\vF(\vX_m,\cdots, \vX_1)}^2_*.
\end{align}
The expectation $\BE_s$ is associated with random matrix $\vX_s$ and $S^c$ denotes the complement of set $S$.
\end{prop}
The proof is identical to Proposition~\ref{prop:general_pauli_expansion}. Note the expectation $\BE_s$ should not be confused with the conditional expectation.

To give a concrete example, we take $\vF$ to be a multi-linear function.
\begin{cor}[Multi-linear function of bounded matrices]\label{cor:bounded_matrix_multilinear}
Consider a degree $r$ multi-linear polynomial 
\begin{align}
\vF(\vX_m,\cdots, \vX_1)= \sum_{i_r,\cdots, i_i}T_{\vec{i}}\vX_{i_r}\cdots \vX_{i_1} = \sum_{S\subset \{m,\cdots,1\}} \sum_{\vec{i}\sim S} T_{\vec{i}}\vX_{i_r}\cdots \vX_{i_1}
\end{align}
where $\vec{i}=i_r,\cdots,i_1$ denotes the tuple and $\vec{i}\sim S$ indicates that the indices coincide (up to relabeling) with the set $S=\{s_r,\cdots,s_1\}$. Suppose each argument $\vX_i$ is an independent random matrix with zero mean $\BE \vX_i=0$ and bounded operator norm $\lV \vX_i\rV \le b_i$. Then 
\begin{align}
    \vertiii{\vF(\vX_m,\cdots, \vX_1)}^2_* &\le (C_p)^{r}\sum_{S\subset \{m,\cdots,1\}} \vertiii{\sum_{\vec{i}\sim S} T_{\vec{i}}\vX_{i_r}\cdots \vX_{i_1} }^2_* \\
    &\le (C_p)^r \sum_{S\subset \{m,\cdots,1\}} \left(b_{s_r}\cdots b_{s_1}\sum_{\vec{i}\sim S}\labs{T_{\vec{i}}}\right)^2 \vertiii{\vI}^2_*.
\end{align}

\end{cor}
Intuitively, the sum over different sets $S$ exhibits a sum-of-squares behavior. Within each set $S$, the reordering of the polynomial is summed via a triangle inequality ($\sum_{\vec{i}\sim S}\labs{T_{\vec{i}}}$), reflecting the fact that we are bounding the matrices $\vX_i$ by their scalar absolute bound $b_i$. This may seem wasteful to matrix concentration specialists but is a mild overhead for our applications. 
\begin{proof}
By Proposition~\ref{prop:general_matrix_function},
\begin{align}
    \vertiii{\sum_{\vec{i}}T_{\vec{i}}\vX_{i_r}\cdots \vX_{i_1}}^2_* &\le \sum_{S\subset \{m,\cdots,1\}} (C_p)^{|S|}\vertiii{\prod_{s\in S}(1-\BE_s)\prod_{s'\in S^c}(\BE_{s'})\sum_{\vec{i}}T_{\vec{i}}\vX_{i_r}\cdots \vX_{i_1}}^2_*\\
    &=\sum_{S\subset \{m,\cdots,1\}} (C_p)^{r}\vertiii{\sum_{\vec{i}}\indicator(\{i_r,\ldots, i_1\}=S)\ T_{\vec{i}}\vX_{i_r}\cdots \vX_{i_1}}^2_*.
\end{align}
The second line uses the multi-linearity that the expectation vanishes for $i\in S^c$ and converts to the indicator. Lastly, we use that the norm $\vertiii{\cdot}_*$ is operator ideal (Fact~\ref{fact:operator ideal}) to convert to the advertised result.
\end{proof}


\subsubsection{Deterministic matrix with Gaussian coefficients}
Thus far, we have shown for a polynomial of random bounded, zero mean matrices. In physics (such as the SYK model), randomness often comes in via adding Gaussian coefficients to a deterministic matrix. 
\begin{prop} \label{prop:Gaussian_coefficients}
Consider random matrices $\vX, \vY$ of the same size and a standard Gaussian $g$ independent of the matrices $\vY,\vX$. For $p\ge 2$, 
\begin{equation}
\vertiii{\vX+g\vY}_*^2 \le \vertiii{\vX}_*^2  + C_p\vertiii{\vY}_*^2 .
\end{equation}
\end{prop}
\begin{proof}
By the central limit theorem, present Gaussian as a sum over i.i.d. Rademachers 
\begin{align}
g = \lim_{N\rightarrow\infty}\sum_i^N \frac{\epsilon_i }{\sqrt{N}}\quad \text{where}\quad \epsilon_i = \pm 1.    
\end{align}
Then, apply uniform smoothness (Fact~\ref{fact:sub_average_pq},Fact~\ref{fact:sub_average_2q_DP}) to the $\epsilon_i \vY$ yields
\begin{align}
    \vertiii{\vX+g\vY}_*^2 &=  \vertiii{\vX+(\lim_{N\rightarrow\infty}\sum_i^N \frac{\epsilon_i }{\sqrt{N}})\vY}_*^2 = \lim_{N\rightarrow\infty}\vertiii{\vX+(\sum_i^N \frac{\epsilon_i }{\sqrt{N}})\vY}_*^2 \\
    &\le \vertiii{\vX}_*^2  + \lim_{N\rightarrow\infty}\sum_i^N\frac{1}{N}C_p\vertiii{\vY}_*^2=\vertiii{\vX}_*^2  + C_p\vertiii{\vY}_*^2. 
\end{align}
This is better than directly applying uniform smoothness 
\begin{align}
   \vertiii{\vX+g\vY}_*^2\le \vertiii{\vX}_*^2  + C_p\vertiii{g\vY}_*^2=\vertiii{\vX}_*^2  + C_p\CO(p)\vertiii{\vY}_*^2
\end{align}
where the Gaussian moments appear $\lV g\rV^2_p=\CO(p)$.
\end{proof}
It is tempting to guess that the coefficient $g$ only needs to be subgaussian, but it is not evident from the proof. At least, one still obtains comparable results if willing to sacrifice factors of $p$, i.e., heavier tails. Back to the discussion, as a corollary, we can upgrade the premise to allow Gaussian coefficients.
\begin{cor}[multi-linear function of matrices with Gaussian coefficients]\label{prop:Gaussian_matrix_multi}
Consider a degree $r$ multi-linear polynomial
\begin{align}
\vF(\vX_m,\cdots, \vX_1) = \sum_{S\subset \{m,\cdots,1\}} \sum_{\vec{i}\sim S} T_{\vec{i}}\vX_{i_r}\cdots \vX_{i_1}\quad \text{where}\quad \vX_i=g_i\vK_i.
\end{align}
The deterministic matrix $\lV \vK_i\rV \le \sigma_i$ are bounded, and the coefficients are i.i.d. standard Gaussians. Then, \begin{align}
    \vertiii{\sum_{\vec{i}}T_{\vec{i}}\vX_{i_r}\cdots \vX_{i_1} }^2_* 
    &\le (C_p)^{r}\sum_{S\subset \{m,\cdots,1\}} \vertiii{\sum_{\vec{i}\sim S}  T_{\vec{i}}\vK_{i_r}\cdots \vK_{i_1} }^2_* \\
    &\le (C_p)^r \sum_{S\subset \{m,\cdots,1\}} \left(\sigma_{s_r}\cdots \sigma_{s_1}\sum_{\vec{i}\sim S}\labs{T_{\vec{i}}}\right)^2 \vertiii{\vI}^2_*.
\end{align}
\end{cor}

This is immediate from Proposition~\ref{prop:Gaussian_coefficients}. For our later development, let us present another proof via the central limit theorem.

\begin{proof}
We can employ the central limit theorem mindset from the ground up. For each argument $\vX_i$, present each Gaussian via i.i.d. Rademachers $\epsilon_{i,j}$
\begin{align}
    \vX_i=g_i\vK_i = \L(\lim_{N\rightarrow\infty}\sum_j^N \frac{\epsilon_{i,j} }{\sqrt{N}}\R) \vK_i:=\lim_{N\rightarrow\infty}\sum_j^N \vY_{i,j}.
\end{align}
Then, the function
\begin{align}
    \sum_{\vec{i}}T_{\vec{i}}\vX_{i_r}\cdots \vX_{i_1} =   \sum_{j_r,\cdots,j_1}^N \sum_{\vec{i}}T_{\vec{i}} \vY_{i_r,j_r}\cdots \vY_{i_1,j_1} =h\L(\vY_{m,N},\cdots,\vY_{1,N},\cdots,\vY_{1,N},\cdots,\vY_{1,1} \R) 
\end{align}
 is again a multi-linear. By Corollary~\ref{cor:bounded_matrix_multilinear}, 
 \begin{align}
    \vertiii{\sum_{j_r,\cdots,j_1}^N \sum_{\vec{i}}T_{\vec{i}} \vY_{i_r,j_r}\cdots \vY_{i_1,j_1} }^2_* &\le (C_p)^{r}\sum_{S'\subset \{mN,\cdots,1\}}  \vertiii{ \sum_{\vec{j}\vec{i}}\indicator((\vec{i,j})\sim S')T_{\vec{i}} \vY_{i_r,j_r}\cdots \vY_{i_1,j_1} }^2_* \\
    &\le (C_p)^r \sum_{S\subset \{m,\cdots,1\}}\sum_{j_{s_r}}\cdots \sum_{j_{s_1}} \vertiii{\sum_{\vec{i}\sim S}T_{\vec{i}}\vY_{i_r,j_{i_r}}\cdots \vY_{i_1,j_{i_1}}}^2_* \\
    &\le(C_p)^r \sum_{S\subset \{m,\cdots,1\}} \left(\sigma_{s_r}\cdots \sigma_{s_1}\sum_{\vec{i}\sim S}\labs{T_{\vec{i}}}\right)^2 \vertiii{\vI}^2_*.
\end{align}
 The second inequality relabel the subset $S'$ by $S\subset \{m,\cdots,1\}$ and $j_{s_r}$ for each element $s_r$. Once fixing the pairs $(s_r,j_{s_r}),\cdots (s_1, j_{s_1})$, the index $j$ is a function of the index $i$ and hence we only need to look for reordering of indices $i_r,\cdots i_1$. Also note that the coefficients $T_{\vec{i}}$ do not depend on the indices $\vec{j}$. 
\end{proof}
\subsubsection{Beyond multi-linear function}
The story was clean and straightforward for multi-linear functions, but the function arising from Trotter error can well be non-multi-linear. With careful accounting, we will derive more general results for bounded matrices (Corollary~\ref{cor:bounded_matrix_general}) and matrices with Gaussian coefficients (Theorem~\ref{thm:Gaussian_matrix}). The bound is qualitatively the same as the multi-linear case but will require heavier notation to account for repeated terms (i.e., terms that show up more than twice). The bounds are analogous to the best-known scalar results~\cite[Theorem~1.4]{schudy2012polyconcentration}. 
\begin{cor}[Polynomial of bounded matrices]\label{cor:bounded_matrix_general}
Consider a multivariate polynomial (with potentially repeated indices $i_a=i_{a'}$) with zero mean, bounded, independent matrix arguments
\begin{align}
    \vF(\vX_m,\cdots, \vX_1)= \sum_{i_r,\cdots, i_i}T_{i_r,\cdots, i_1}\vX_{i_r}\cdots \vX_{i_1} \quad \text{where} \quad \BE \vX_i=0, \lV \vX_i\rV \le b_i.
\end{align} 
For $p \ge 2$, 
\begin{align*}
    &\vertiii{\vF(\vX_m,\cdots, \vX_1) }^2_* 
    &\le \sum_{S\subset \{m,\cdots,1\}} (4C_p)^{|S|}\left(\sum_{\Supp(\vec{u})=S}b_{1}^{u_1}\cdots b_{m}^{u_m}\sum_{\substack{\Supp(\vec{v})=S^c, \\ v_i\ne 1}} b_{1}^{v_1}\cdots b_{m}^{v_m}\sum_{\pi}|T_{\pi(\vec{u},\vec{v})}|  \right)^2\vertiii{\vI}^2_*
\end{align*}
where $\pi$ enumerates reordering of polynomial $\vX^{u_1+v_1}_{1}\cdots\vX^{u_m+u_1}_{m}$.
\end{cor}
In other words, as usual, we have sum-of-square behavior across different sets $S$. For each set $S$, \begin{enumerate}
    \item select the powers $u_1,\cdots u_m = 0,1,\cdots$ such that its support gives the set $S$,
    \item select the powers $v_1,\cdots, v_m  = 0,2,\cdots$ such that its support gives the complement $S^c$, and
    \item enumerate the reorderings $\pi$ of the non-commutative polynomial.
\end{enumerate}
The takeaway for this calculation is that (1) the larger set $|S|$ corresponds to a heavier tail $C_p^{|S|}$, and (2) the dominating contribution often comes from larger sets $S$ (if we fix the total degree and grow the number of summands). There, unevenly distributed values of powers $v_i\gg 2$ and $v_i\gg 1$ suppress the combinations of other possible $u_i,v_i$. 

\begin{proof}
By Proposition~\ref{prop:general_matrix_function},
\begin{align}
 & \vertiii{\sum_{\vec{i}}T_{\vec{i}}\vX_{i_r}\cdots \vX_{i_1}}^2_* \le \sum_{S\subset \{m,\cdots,1\}} (C_p)^{|S|}\vertiii{\prod_{s\in S}(1-\BE_s)\prod_{s'\in S^c}(\BE_{s'})\sum_{\vec{i}}T_{\vec{i}}\vX_{i_r}\cdots \vX_{i_1}}^2_p\\
 & = \sum_{S\subset \{m,\cdots,1\}} (C_p)^{|S|}\vertiii{\sum_{\vec{i}}T_{\vec{i}}\left(\prod_{s\in S}(1-\BE_s)\sum_{\Supp(\vec{u})=S}\right) \left(\prod_{s'\in S^c}(\BE_{s'})\sum_{\substack{\Supp(\vec{v})=S^c,\\ v_i\ne 1}} \right) \indicator(\vec{i}\sim (\vec{u},\vec{v}) )\vX_{i_r}\cdots \vX_{i_1}}^2_p\\
&=\sum_{S\subset \{m,\cdots,1\}} (C_p)^{|S|}\vertiii{\sum_{\Supp(\vec{u})=S} \sum_{\substack{\Supp(\vec{v})=S^c,\\ v_i\ne 1}}\sum_{\pi} T_{\pi(\vec{u},\vec{v})}\pi\bigg((\vX_{1}^{u_1}-\BE\vX_{1}^{u_1})\cdots (\vX_{m}^{u_m}-\BE \vX_{m}^{u_m})\BE\vX_{1}^{v_1}\cdots \BE \vX_{m}^{v_m}  \bigg) }^2_*.
\end{align}
The second equality inserts indicators for powers in polynomial (up to reordering) $\vX^{u_1}_{1}\cdots\vX^{u_m}_{m}\vX^{v_1}_{1}\cdots\vX^{v_m}_{m}$, with the array $\vec{v}$ for the powers of elements in $S$ and the array $\vec{u}$ for elements in the complement $S^c$. The third equality evaluates the expectations to give the constraint $v_i \ne 1$ and denote by $\pi$ the reordering of the non-commutative polynomial. Lastly, we use $\vertiii{\cdot}_*$ being operator ideal (Fact~\ref{fact:operator ideal}) to convert to bounds on the individual spectral norm. The factor $4^{|S|}$ is due to the crude estimate $\norm{\vX^u-\BE\vX^u}\le 2^{u} \norm{\vX^u}$. This is the advertised result.
\end{proof}

Next, we will use the central limit theorem to upgrade to Gaussian variables. 
\begin{thm}[Polynomial of matrices with Gaussian coefficients]\label{thm:Gaussian_matrix}
Consider a multivariate polynomial (with potentially repeated indices $i_a=i_{a'}$) where each argument is a bounded matrix with an i.i.d. standard Gaussian coefficient
\begin{align}
    \vF(\vX_m,\cdots, \vX_1)= \sum_{i_r,\cdots, i_i}T_{i_r,\cdots, i_1}\vX_{i_r}\cdots \vX_{i_1} \quad \text{where} \quad \vX_i=g_i\vK_i\quad  \text{and}\quad \lV \vK_i\rV \le \sigma_i.
\end{align} 
For $p \ge 2$,
\begin{align}
    &\vertiii{\vF(\vX_m,\cdots, \vX_1) }^2_* 
    \le \sum_{\vec{u}} (C_p)^{|\vec{u}|} \left(\sum_{\vec{v}} \sigma_{1}^{u_1+2v_1} \cdots \sigma_{m}^{u_m+2v_m} \sum_{\pi} |T_{\pi(\vec{u},\vec{v})}|w(\vec{u},\vec{v})\right)^2 \vertiii{\vI}^2_* \quad \text{where} \quad \labs{\vec{u}} := \sum_{i=1}^{m} u_i.
\end{align}
The $\pi(\vec{u},\vec{v})$ enumerates the reorderings of polynomial $\vX^{u_1+2v_1}_1\cdots \vX^{u_m+2v_m}_m$ and 
\begin{align}
     w(\vec{u},\vec{v})= \prod_{i=1}^{m} \frac{(u_i+2v_i-1)!!}{(u_i-1)!!}  \le \prod_{i=1}^{m} (u_i+2v_i)^{v_i}.
\end{align}
\end{thm}
Let us parse the expression.
\begin{enumerate}
    \item The powers $u_1,\cdots u_m = 0,1,\cdots$ are summed incoherently; the sum $\labs{\vec{u}}=\sum_i u_i$ determines the power $C_p^{\labs{\vec{u}}}$.
    \item In the square, fill in the remaining powers $v_1,\cdots, v_m  = 0,2,\cdots$ and sum over them.
    \item Enumerate the reorderings $\pi(\vec{u},\vec{v})$ of the non-commutative polynomial. The prefactor $w(\vec{u},\vec{v})$ comes from Wick contractions.
\end{enumerate}

We see a sum over squares $\sum_{\vec{u}}$, but unlike the multi-linear case, we also allow the same term to have a larger power $\sigma_1^{u_1}$. This is because of the Gaussian tails. One may be concerned about the sum $\sum_{\vec{v}}$ \textit{inside} the square, but looking more carefully, we see each term is already squared $\sigma_i^{2v_i}$. 

Intuitively, the RHS (without the coefficients $C_p$) can be interpreted as a \textit{variance} proxy.
\begin{prop}[Variance Proxy]\label{prop:var_proxy}
\begin{align}
    \sum_{\vec{u}}  \left(\sum_{\vec{v}} \sigma_{1}^{u_1+2v_1} \cdots \sigma_{i}^{u_i+2v_i} \sum_{\pi} |T_{\pi(\vec{u},\vec{v})}|w(\vec{u},\vec{v})\right)^2 = \BE\L[ \L(\sum_{\vec{i}}\labs{T_{\vec{i}}}x_{i_r}\cdots x_{i_1}\R)^2 \R] \label{eq:variance_proxy}
\end{align}
where $x_i (= g_i\sigma_i)$ are scalar Gaussians with variance $\sigma_i^2$.
\end{prop}
This essentially matches the scalar results (Fact~\ref{fact:hyper_scalar_g}) up to the reordering and absolute values $\labs{T_{\vec{i}}}$.\footnote{ The scalar result also replaces the coefficient $(C_p)^{|\vec{u}|}$ with the uniform upper bound $(C_p)^r$.}. Roughly speaking, from Gaussian scalars to matrices (for a bounded matrix $\norm{\vK_i}\le \sigma_i$)
\begin{align}
x_i=g_i\sigma_i \rightarrow \vX_i=g_i\vK_i,    
\end{align}
uniform smoothness tells us that the analogous bound holds; we only need to estimate the variance proxy. 

\begin{proof}[Proof of Theorem~\ref{thm:Gaussian_matrix}]
Let us painfully employ the central limit theorem as in the proof of Proposition~\ref{prop:Gaussian_matrix_multi}. Recall
\begin{align}
    \vX_i=g_i\vK_i = \L(\lim_{N\rightarrow\infty}\sum_j^N \frac{\epsilon_{i,j} }{\sqrt{N}}\R) \vK_i:=\lim_{N\rightarrow\infty}\sum_j^N \vY_{i,j}.
\end{align}
By Corollary~\ref{cor:bounded_matrix_general},
\begin{align}
    &\vertiii{\sum_{j_r,\cdots,j_1}^N \sum_{\vec{i}}T_{\vec{i}} \vY_{i_r,j_r}\cdots \vY_{i_1,j_1} }^2_*\notag\\
    &\stackrel{N\rightarrow \infty}{\le} \sum_{\vec{u}} (C_p)^{|\vec{u}|} \sum_{\vec{j}_1}\cdots \sum_{\vec{j}_m}\lvertiii \sum_{\pi,\vec{v}} \sum_{\vec{k}_1}\cdots \sum_{\vec{k}_m} T_{\pi((\vec{u},\vec{j}),(\vec{v},\vec{k}))}\pi\bigg((\vY_{1,j_{1,1}}\cdots \vY_{1,j_{1,u_1}}) \cdots (\vY_{m,j_{m,1}}\cdots \vY_{m,j_{m,u_m}})\notag\\ 
    & \hspace{9cm} (\BE\vY_{1,k_{1,1}}^{2}\cdots\BE\vY_{1,k_{1,v_m}}^{2})\cdots (\BE\vY_{m,k_{m,1}}^{2}\cdots\BE\vY_{m,k_{m,v_m}}^{2})  \bigg)\rvertiii^2_*  \notag
    \\
    &=: \sum_{\vec{u}} (C_p)^{|\vec{u}|} \sum_{\vec{j}_1}\cdots \sum_{\vec{j}_m} \vertiii{\sum_{\pi,\vec{v}} \sum_{\vec{k}_1}\cdots \sum_{\vec{k}_m} T_{\pi((\vec{u},\vec{j}),(\vec{v},\vec{k}))}\pi\left(\vY_{1}^{(u_1)} \cdots \vY_{m}^{(u_m)} (\BE\vY_{1}^{2})^{(v_1)}\cdots(\BE\vY_{m}^{2})^{(v_m)} \right)}^2_*\\
    &=\sum_{\vec{u}} (C_p)^{|\vec{u}|}  \vertiii{\sum_{\pi,\vec{v}} T_{\pi(\vec{u},\vec{v})}w(\vec{u},\vec{v})\pi\left(\vK_{1}^{u_1+2v_1} \cdots \vK_{m}^{u_m+2v_m} \right)}^2_*\\
    &\le\sum_{\vec{u}} (C_p)^{|\vec{u}|} \left( \sum_{\vec{v}}\sigma_{1}^{u_1+2v_1} \cdots \sigma_{m}^{u_m+2v_m} \sum_{\pi} |T_{\pi(\vec{u},\vec{v})}|w(\vec{u},\vec{v})\right)^2 \vertiii{\vI}^2_*.
 \end{align}
Importantly, in the large $N$ limit for the central limit theorem, the only possible contribution would be the linear terms (e.g., $\vY_{1,j_{1,1}}$) and the expected squares (e.g., $\BE\vY_{1,j_{1,1}} ^2$); any cubic term (e.g., $\vY_{1,j_{1,1}}^3$) is subleading in $1/N$. The array $\vec{u}=u_1,\cdots,u_m$ collects the number of duplicates $u_i$ of each argument $\vY_1,\cdots, \vY_m$ and contributes to $(C_p)^{\labs{\vec{u}}}$. Given the array $\vec{u}$, the sum $\sum_{\vec{j_1}}:=\sum_{j_{1,1}}\cdots \sum_{j_{1,u_1}}$ runs through the duplicates $\vY_{1,j_{1,1}}\cdots \vY_{1,j_{1,u_1}}$ of the term $\vY_1$.

 Second, we compress the notation by grouping duplicates into the exponent. Third, we get rid of duplicates by summing over indices $\vec{j},\vec{k}$ and also drop the Rademachers. The function $w(\vec{u},\vec{v})$ counts the number of ordering $\vY_1^{(u)}(\BE \vY_1^2)^{(v)}$ from $\vX^{u+2v}_1$
\begin{align}
    w(u,v)&:= \binom{u+2v}{u}\cdot (2v)!! \le 2^{u+2v}(2v)^{v}\\
    w(\vec{u},\vec{v})&:= w(u_1,v_1)\cdots w(u_m,v_m).
\end{align}
The binomial coefficient $\binom{u+2v}{u}$ counts the locations of $\vY_1^{(u)}$, and $(2v)!!$ are precisely the Wick contractions for $(\BE \vY_1^2)^{(v)}$. This is the advertised result.
\end{proof}

\begin{proof}[Proof of Proposition~\ref{prop:var_proxy}]
To see why the RHS is the advertised variance proxy~\eqref{eq:variance_proxy}, we apply the two-point inequality (Fact~\ref{fact:two_point}) to the variance proxy. At $p=2$ (which gives $C_p=1$) the two-point inequality becomes an equality, which implies equality for a scalar version of Proposition~\ref{prop:general_matrix_function}. Following the above chain of equations yields the desired expression.
\end{proof}

\section{Random $k$-local Hamiltonians}\label{chap:random}
In this section, we consider k-local Hamiltonians drawn from an ensemble (such as the SYK models) where the terms are bounded matrices with standard Gaussian coefficients $g_{\gamma}$
\begin{align}
   \vH = \sum_{\gamma=1}^\Gamma \vH_\gamma =   \sum_{\gamma=1}^\Gamma g_\gamma \vK_\gamma \quad \text{where} \quad  \BE[g_\gamma^2]=1 \quad \text{and} \quad  \norm{\vK_\gamma} \le b_\gamma.
\end{align}
We apply matrix concentration inequalities to its Trotter error (Section~\ref{sec:matrix_poly}). Recall the local quantities
\begin{align}
 \lnormp{\vH}{(global),2} := \sqrt{\sum_{\gamma} b_\gamma^2} \quad\text{and}\quad
 \lnormp{\vH}{(local),2} &:= \max_{i}  \sqrt{\sum_{\gamma: i \subset \gamma } b^2_\gamma} \quad\text{where}\quad b_{\gamma}:= \norm{\vK_{\gamma}}.
\end{align}

\begin{thm}[Trotter error for random Hamiltonians]\label{thm:Trotter_random_H}
Simulating random $k$-local models with Gaussian coefficients via $2\ell$-th order Suzuki formulas, the gate count 
\begin{align*}
     G &= \Omega\bigg[ 
      \lnormp{\vH}{(local),2} \sqrt{n+\log(1/\delta ) }t \cdot \L(\frac{\lnormp{\vH}{(global),2}^2\sqrt{n+\log(1/\delta ) } t}{\lnormp{\vH}{(local),2} \epsilon} \R)^{\frac{1}{\ell}} 
     \bigg]\quad \text{for}\quad \ell \text{ even}, \\
     G &= \Omega\bigg[ 
      \lnormp{\vH}{(local),2} \sqrt{n+\log(1/\delta ) }t \cdot \L( \L(\frac{\lnormp{\vH}{(global),2}\sqrt{n+\log(1/\delta ) } t}{ \epsilon} \R)^{\frac{1}{\ell}}+\L(\frac{\lnormp{\vH}{(global),2}^2\sqrt{n+\log(1/\delta ) } t}{\lnormp{\vH}{(local),2} \epsilon} \R)^{\frac{1}{\ell+1}}\R) 
     \bigg]\quad \text{for}\quad \ell \text{ odd}\\
    &\hspace{4cm}\textrm{ensures}\quad \Pr\L(\norm{e^{\iunit \vH t}- \vec{S}(t/r)^r} \ge \epsilon\R) \le \delta \quad(\text{worst inputs}).
\end{align*}
The probability $\Pr(\cdot)$ arises from the random Hamiltonian ensemble. For fixed but arbitrary input state $\vrho$, the gate count
\begin{align}
     G &= \Omega\bigg[ 
      \lnormp{\vH}{(local),2} \sqrt{\log(1/\delta ) }t \cdot \L(\frac{\lnormp{\vH}{(global),2}^2\sqrt{\log(1/\delta ) } t}{\lnormp{\vH}{(local),2} \epsilon} \R)^{\frac{1}{\ell}} 
     \bigg]\quad \text{for}\quad \ell \text{ even}, \\
     G &= \Omega\bigg[ 
      \lnormp{\vH}{(local),2} \sqrt{\log(1/\delta ) }t \cdot \L( \L(\frac{\lnormp{\vH}{(global),2}\sqrt{\log(1/\delta ) } t}{ \epsilon} \R)^{\frac{1}{\ell}}+\L(\frac{\lnormp{\vH}{(global),2}^2\sqrt{\log(1/\delta ) } t}{\lnormp{\vH}{(local),2} \epsilon} \R)^{\frac{1}{\ell+1}}\R) \\
    &\hspace{4cm}\textrm{ensures}\ \ \Pr\left(\frac{1}{2}\lnormp{\e^{-\iunit \vH t}\vrho \e^{\iunit \vH t}- \vec{S}(t/r)^{\dagger r}\vrho \vec{S}(t/r)^r}{1} \ge \epsilon\right) \le \delta &(\text{fixed inputs}).
\end{align}

\end{thm}
This is similar but different from the non-random $k$-local results (Theorem~\ref{thm:Trotter_non_random}): when the Hamiltonian is random, an arbitrary fixed input $\vrho$ already displays a 2-norm scaling; even the worst input states that may correlate with the Hamiltonian (the Gibbs state or the ground state) enjoys concentration with a price of $\sqrt{n}$. More carefully, the concentration here is stronger: the dependence on the failure probability $(\log(1/\delta))$ has a lower power because of the many independent Gaussians. However, the factor in $(\cdot)^{1/\ell}$ is slightly worse
(which will be suppressed at large $\ell$, anyway).

The proof strategy is the same Taylor expansion as in Section~\ref{sec:sketch} but with different norms. Recall the error 
\begin{align}
    \vCE = \sum_{g=\ell+1}^{g'-1} \vCE_g + \vCE_{\ge g'}.
\end{align}
The two error terms are combined in Section~\ref{sec:proof_random}. We also present an argument for lower bounds at short times (Section~\ref{sec:counting}).

\subsection{Bounds on the $g$-th Order}

We proceed by controlling each $g$-th order polynomial (for $ \ell < g < g'$) by Theorem~\ref{thm:Gaussian_matrix} for $\sigma_\gamma:= \norm{\vH_{\gamma}}$ 
\begin{align}
    &\vertiii{\vec{\CE}_g}_*^2 = \vertiii{\sum^{J}_{j=1} \sum_{g_J+\cdots+g_{j+1}=g-1}\CL^{g_J}_J\cdots \CL^{g_{j+1}}_{j+1} [\vH_{j+1}]\frac{t^{g-1}}{g_J!\cdots g_{j+1}!}}_*^2\\
    &\le  (2t)^{2(g-1)} \sum_{\vec{u}} (C_p)^{|\vec{u}|} \left(\sum_{\vec{v}} \sigma_{1}^{u_1+2v_1} \cdots \sigma_{\Gamma}^{u_\Gamma+2v_\Gamma}w(\vec{u},\vec{v}) \sum_{j=1}^{J}\sum_{\vec{g}} \frac{\indicator(\vec{g}\sim \vec{u}+2\vec{v})}{g_J!\cdots g_{j+1}!} \right)^2 \vertiii{\vI}^2_*. \label{eq:for_numerics_random_H}
\end{align}
The indicator $\indicator(\vec{g}\sim \vec{u}+2\vec{v})$ (1) keep track of the occurrences $\vec{u},\vec{v}$ of the Hamiltonian terms $\vH_\gamma$ of any ordering $\CL^{g_J}_J\cdots \CL^{g_{j+1}}_{j+1} [\vH_{j+1}]$ and (2) enforces the commutation constraint, i.e., it returns zero if some term $\CL$ commutes through all terms on its right. The factor $2$ is due to the commutator with coefficient $a_j\CL_j$ and the uniform bound $\labs{a_j} \le 1$. 

This scalar sum~\eqref{eq:for_numerics_random_H} can be numerically evaluated to get explicit gate counts for the particular systems. To get analytic estimates, we proceed with the combinatorics. First, as in~\eqref{eq:g-order_LLL}, throw in extra terms to count the symmetrized sum $\vec{\gamma}$ instead of the product-formula-dependent $\vec{g}$
\begin{align}
    \sum_{\vec{g}}\frac{1}{g_J!\cdots g_{j+1}!} \CL^{g_J}_J\cdots \CL^{g_{j+1}}_{j+1}[\vH_{j+1}] \rightarrow \sum_{\gamma_{g-1}=1}^\Gamma \cdots \sum_{\gamma_0=1}^\Gamma \CL_{\gamma_{g-1}} \cdots \CL_{\gamma_1}[\vH_{\gamma_0}]. 
\end{align}
This yields  
\begin{align}
    (cont.)\ &\le  \vertiii{\vI}^2_*(2\Upsilon)^2 (2\Upsilon t)^{2(g-1)} \sum_{\vec{u}} (C_p)^{|\vec{u}|} \left(\sum_{\vec{v}} \sigma_{1}^{u_1+2v_1} \cdots \sigma_{\Gamma}^{u_\Gamma+2v_\Gamma}w(\vec{u},\vec{v}) \sum_{\gamma_{g-1}=1}^\Gamma \cdots \sum_{\gamma_0=1}^\Gamma \indicator(\vec{\gamma}\sim \vec{u}+2\vec{v}) \right)^2 \\
    &\le (\cdot) \sum_{|\vec{u}|=0}^g \bigg[\sum_{\vec{u}} (C_p)^{|\vec{u}|} \sigma_{1}^{2u_1} \cdots \sigma_{\Gamma}^{2u_\Gamma}\sum_{\vec{v}} \sigma_{1}^{2v_1} \cdots \sigma_{\Gamma}^{2v_\Gamma} w(\vec{u},\vec{v}) \sum_{\vec{\gamma}} \indicator(\vec{\gamma}\sim \vec{u}+2\vec{v}) \label{eq:first_holder}\\
    &\hspace{4.5cm}\cdot \max_{\vec{u}} \left(\sum_{\vec{v'}} \sigma_{1}^{2v'_1} \cdots \sigma_{\Gamma}^{2v'_\Gamma} w(\vec{u},\vec{v'})\sum_{\vec{\gamma}} \indicator(\vec{\gamma}\sim \vec{u}+2\vec{v'}) \right)\bigg]\label{eq:second_holder}.
\end{align}
The second inequality is Holder's for the sum over $\vec{u}$. 
For the first term~\eqref{eq:first_holder},
\begin{align}
&\sum_{\vec{u}} (C_p)^{|\vec{u}|} \sigma_{1}^{2u_1} \cdots \sigma_{\Gamma}^{2u_\Gamma}\sum_{\vec{v}} \sigma_{1}^{2v_1} \cdots \sigma_{\Gamma}^{2v_\Gamma} w(\vec{u},\vec{v}) \sum_{\vec{\gamma}} \indicator(\vec{\gamma}\sim \vec{u}+2\vec{v}) \label{eq:g_order_gamma_u_v}\\
&\hspace{2cm}\le g^{|\vec{v}|}\cdot g^{|\vec{v}|}\cdot (C_p)^{g} \cdot \sum_{\vec{u}} \sigma_{1}^{2u_1} \cdots \sigma_{\Gamma}^{2u_\Gamma}\sum_{\vec{v}} \sigma_{1}^{2v_1} \cdots \sigma_{\Gamma}^{2v_\Gamma} \sum_{\vec{\gamma'}} \indicator(\vec{\gamma'}\sim \vec{u}+\vec{v}) \\
&\hspace{2cm}\le g^{g}\cdot (C_p)^{g} \left( gk \lnormp{\vH}{(local),2}^2\right)^{g-1-|\vec{v}|} \lnormp{\vH}{(global),2}^2. 
\end{align}
The first inequality passes $\vec{\gamma}$ to $\vec{\gamma'}$ by assigning pairings of $\vec{v}$ (which is bounded by $g^{|\vec{v}|}$); this divides the $\vec{v}$ occurences by two. The other factor comes from crude uniform estimates $w(\vec{u},\vec{v}) \le 2^g g^{|\vec{v}|}$. The second inequality uses the bound $2\labs{\vec{v}} \le g$ to combine the sum over $\vec{u}$ and $\vec{v}$ and evaluates this chain of commutators (as in~\cite{thy_trotter_error}, but here we uses the 2-norm $\lnormp{\vH}{(local),2}$ instead of 1-norm $\lnormp{\vH}{(local),1}$). We used a crude bound $( g(k-1)+1 )\le gk$ for the locality of commutators.

For the second term~\eqref{eq:second_holder},
\begin{align}
    \max_{\vec{u}} \left(\sum_{\vec{v'}} \sigma_{1}^{2v'_1} \cdots \sigma_{\Gamma}^{2v'_\Gamma} w(\vec{u},\vec{v'})\sum_{\vec{\gamma}} \indicator(\vec{\gamma}\sim \vec{u}+2\vec{v'}) \right)    &\le 2^g g^{|\vec{v}|}\cdot g^{|\vec{u}|} g^{|\vec{v}|}\cdot  \max_{\vec{u}} \left(\sum_{\vec{v'}} \sigma_{1}^{2v'_1} \cdots \sigma_{\Gamma}^{2v'_\Gamma}\sum_{\vec{\gamma''}} \indicator(\vec{\gamma''}\sim \vec{v'}) \right)\\
    &\le (2g)^g\cdot \left( gk\lnormp{\vH}{(local),2}^2 \right)^{|\vec{v}|}   \L(\frac{\lnormp{\vH}{(global),2}^2}{gk\lnormp{\vH}{(local),2}^2}\R)^{\indicator(|\vec{u}|=0)}.
\end{align}
Again, in passing $\vec{\gamma}$ to $\vec{\gamma''}$ we first select the locations $\vec{u}$, and then the pairings over $\vec{v}$. Given a non-empty set of terms in $\vec{u}$, each terms $\vec{v}$ needs to chain together, giving the factor $gk\lnormp{\vH}{(local),2}^2$. For the edge case $|\vec{u}|=0$, we obtain one additional factor $\lnormp{\vH}{(global),2}^2$ since we lose the chaining constraints from $\vec{u}$. Altogether, summing over $\sum_{\labs{u}}$ gives additional factor $g+1\le 2g$
\begin{align}
    \vertiii{[\vec{\CE} (\vH_1,\cdots,\vH_\Gamma,t)]_g}_* &\le \sqrt{C_p}^g \vertiii{\vI}_*2g^{3g/2}\cdot \L(4\sqrt{k} \Upsilon \lnormp{\vH}{(local),2} t \R)^{g} \cdot 
    \begin{cases}
    \displaystyle\frac{\lnormp{\vH}{(global),2}^2}{tk\lnormp{\vH}{(local),2}^2} \quad g\text{ even}\\
   \displaystyle\sqrt{g}\frac{\lnormp{\vH}{(global),2}}{t\sqrt{k}\lnormp{\vH}{(local),2}} \quad g\text{ odd} \label{eq:g_random}
    \end{cases}
\end{align}
for order $g\ge \ell+1 \ge 3$. 
On the other hand, for the first order Trotter ($\ell =1$) we have a sharper bound (Theorem~\ref{thm:first_order_gate_count}).

\subsection{Bounds for $g'$-th Order and Beyond}
In this section, we handle the edge case. Its presentation will suffer from adhoc prefactors, but fortunately, they will not impact the ultimate performance. The expression has an integral so we first remove it via unitary invariance of $\vertiii{\cdot}_*$ norm. Then, we may plug the polynomial into Theorem~\ref{thm:Gaussian_matrix}. 
\begin{align*}
    &\vertiii{[\vec{\CE} (\vH_1,\cdots,\vH_\Gamma,t)]_{\ge g}}_* \\
    &= \vertiii{ \sum^{J}_{j=1} \sum_{m=j+1}^{J} \e^{\CL_J t}\cdots \e^{\CL_{m+1} t} \int_0^t dt_1 \sum_{g_m+\cdots+g_{j+1}=g-1,g_m\ge 1}e^{\CL_{m} t_1} \CL^{g_m}_m\cdots \CL^{g_{j+1}}_{j+1}[\vH_{j}] \frac{(t-t_1)^{g_m-1}t^{g'-g_m-1}}{(g_m-1)!\cdots g_{j+1}!}}_*\\
    &\le \sum_{m=2}^{J} \vertiii{ \sum^{J}_{j=m-1} \sum_{g_m+\cdots+g_{j+1}=g-1,g_m\ge 1}\CL^{g_m}_m\cdots \CL^{g_{j+1}}_{j+1}[\vH_{j}] \frac{t^{g-1}}{g_m!\cdots g_{j+1}!}}_* \\
    &\le \vertiii{\vI}_*2\Upsilon (2\Upsilon t)^{g-1} \sum_{m=2}^{J}\sqrt{\sum_{\vec{u}=0}^{g} (C_p)^{|\vec{u}|} \left(\sum_{\vec{v}} \sigma_{1}^{u_1+2v_1} \cdots \sigma_{\Gamma}^{u_\Gamma+2v_\Gamma}w(\vec{u},\vec{v}) \sum_{\gamma_{g-1}=1}^\Gamma \cdots \sum_{\gamma_0=1}^\Gamma \indicator(\vec{\gamma}\sim \vec{u}+2\vec{v}, \gamma_{g-1}=\gamma(m)) \right)^2 }.
\end{align*}
We have a similar expression except for the constraint $\gamma_{g-1}=\gamma(m)$ coming from $g_m\ge 1$ and the sum $\sum_{m=2}^{J}$ outside the square root. This amounts to minor tweaks in the calculation.
For the first term,
\begin{align}
&\sum_{\vec{u}} (C_p)^{|\vec{u}|} \sigma_{1}^{2u_1} \cdots \sigma_{\Gamma}^{2u_\Gamma}\sum_{\vec{v}} \sigma_{1}^{2v_1} \cdots \sigma_{\Gamma}^{2v_\Gamma} w(\vec{u},\vec{v}) \sum_{\vec{\gamma}} \indicator(\vec{\gamma}\sim \vec{u}+2\vec{v},\gamma_{g-1}=\gamma(m)) \\
&\hspace{2cm}\le g^{|\vec{v}|}\cdot g^{|\vec{v}|}\cdot (C_p)^{g} \cdot \sum_{\vec{u}} \sigma_{1}^{2u_1} \cdots \sigma_{\Gamma}^{2u_\Gamma}\sum_{\vec{v}} \sigma_{1}^{2v_1} \cdots \sigma_{\Gamma}^{2v_\Gamma} \max_{0\le j\le g-1}\sum_{\vec{\gamma'}} \indicator(\vec{\gamma'}\sim \vec{u}+\vec{v}, \gamma'_{j}=\gamma(m)) \\
&\hspace{2cm}\le g^{g}\cdot (C_p)^{g} \left( gk \lnormp{\vH}{(local),2}^2\right)^{g-|\vec{v}|-1} \norm{\vH_{\gamma(m)}}^2. 
\end{align}
The only difference from~\eqref{eq:g_order_gamma_u_v} is that once a pairing of $\vec{v}$ is chosen, we lose one choice of $\gamma'$, but this could happen at any $\gamma'_j$. 

For the second term,
\begin{align}
\max_{\vec{u}} \sum_{\vec{v}} \sigma_{1}^{2v_1} \cdots \sigma_{\Gamma}^{2v_\Gamma} w(\vec{u},\vec{v}) \sum_{\vec{\gamma}} \indicator\L(\vec{\gamma}\sim \vec{u}+2\vec{v},\gamma_{g-1}=\gamma(m)\R) &\le 2^g g^{|\vec{v}|}g^{|\vec{u}|}\cdot g^{|\vec{v}|}\cdot \sum_{\vec{v}} \sigma_{1}^{2v_1} \cdots \sigma_{\Gamma}^{2v_\Gamma} \sum_{\vec{\gamma''}} \indicator(\vec{\gamma''}\sim \vec{v}) \\
&\le (2g)^{g}\left( gk \lnormp{\vH}{(local),2}^2\right)^{|\vec{v}|}.
\end{align}
For the first inequality, since the term $\gamma_{g-1}$ may or may not be in $\vec{v}$, we bound by the latter case.
Finally, sum over index $m$ to arrive at
\begin{align}
    \vertiii{[\vec{\CE} (\vH_1,\cdots,\vH_\Gamma,t)]_{>g}}_* 
    &\le  (C_p)^{g/2}\vertiii{\vI}_*\left(4\sqrt{k} \Upsilon \lnormp{\vH}{(local),2}t\right)^{g} 
    2g^{3g/2}\cdot\frac{\lnormp{\vH}{(global),1}}{t\sqrt{k} \lnormp{\vH}{(1,2)}}.
    \label{eq:g'_random}
\end{align}
\subsection{Proof of Theorem~\ref{thm:Trotter_random_H}}\label{sec:proof_random}
The proof is analogous to Section~\ref{sec:proof_non_random}, with minor changes. We only highlight the differences and hide the numerical factors.
\begin{proof}
From the bounds for each $g-th$ order~\eqref{eq:g_random} and the $g'$-th order~\eqref{eq:g'_random}, define
\begin{align}
    c(k)&:= 4\sqrt{k} \lnormp{\vH}{(local),2}.
\end{align}
For a short time $\tau = t$, we arrange and perform the last integral using estimate $\int (\tau')^{g-1} d\tau'\le \tau'^{g}/g$
\begin{align}
    \frac{\vertiii{e^{\iunit \vH \tau}- \vec{S}_\ell(\tau)}_{*,p}}{\vertiii{\vI}_{*,p}} &\le \int^\tau_0 \frac{\vertiii{\vCE(\tau')}_{*,p}}{\vertiii{\vI}{*,p}} d\tau' \\
    &\le \frac{2\lnormp{\vH}{(global),2}^2}{k\lnormp{\vH}{(local),2}^2} \cdot \sum_{g=\ell+1, even}^{g'-1} \left(g^{3/2}\sqrt{C_p}c(k)\Upsilon\tau \right)^{g}+\frac{2\lnormp{\vH}{(global),2}}{\sqrt{k}\lnormp{\vH}{(local),2}} \cdot \sum_{g=\ell+1, odd}^{g'-1} \left(g^{3/2}\sqrt{C_p}c(k)\Upsilon\tau \right)^{g}  \\  
    &+ \frac{\lnormp{\vH}{(global),1}}{t\sqrt{k} \lnormp{\vH}{(1,2)}} \cdot \left(g'^{3/2}\sqrt{C_p}c(k)\Upsilon\tau \right)^{g'}.
\end{align}

Now, we evaluate Markov's inequality to obtain concentration.

(I) For fixed input (suffice to consider pure states $\ket{\psi}$), 
\begin{align}
    \Pr( \lnormp{ \vec{\CE}_{tot} \ket{\psi} }{\ell_2}  \ge \epsilon)\le \frac{ \vertiii{\vec{\CE}_{tot}}_{fix, p}^p}{\epsilon^p} 
    \le 
    \delta .
\end{align}

Altogether, we obtain the rounds required (dropping $k,\ell$ dependences)
\begin{align}
     r &= \Omega\bigg[ 
      \lnormp{\vH}{(local),2} \sqrt{\log(1/\delta ) }t \cdot \L(\frac{\lnormp{\vH}{(global),2}^2\sqrt{\log(1/\delta ) } t}{\lnormp{\vH}{(local),2} \epsilon} \R)^{\frac{1}{\ell}} 
     \bigg]\quad \text{for}\quad \ell \text{ even}, \\
     r &= \Omega\bigg[ 
      \lnormp{\vH}{(local),2} \sqrt{\log(1/\delta ) }t \cdot \L( \L(\frac{\lnormp{\vH}{(global),2}\sqrt{\log(1/\delta ) } t}{ \epsilon} \R)^{\frac{1}{\ell}}+\L(\frac{\lnormp{\vH}{(global),2}^2\sqrt{\log(1/\delta ) } t}{\lnormp{\vH}{(local),2} \epsilon} \R)^{\frac{1}{\ell+1}}\R) 
     \bigg]\quad \text{for}\quad \ell \text{ odd}.
     \end{align}
Note that when the order $\ell$ is odd (i.e., $g= \ell+1$ is even), a larger cost may incur at order $g+1$.  

(II) For the spectral norm,
\begin{align}
    \hat{\Pr}( \norm{ \vec{\CE}_{tot} } \ge \epsilon)\le d\cdot\L(\frac{ \vertiii{\vec{\CE}_{tot}}_{p}}{\vertiii{\vI}_{p} \epsilon} \R)^p
    \le\delta.
\end{align}
The factor of Hilbert space dimension $d$ comes from the trace in Schatten p-norm. This would require a higher cost ($\log(\delta)\rightarrow \sqrt{n+\log(\delta)}$) to ensure the failure probability is small 
\begin{align}
     r &= \Omega\bigg[ 
      \lnormp{\vH}{(local),2} \sqrt{n+\log(1/\delta ) }t \cdot \L(\frac{\lnormp{\vH}{(global),2}^2\sqrt{n+\log(1/\delta ) } t}{\lnormp{\vH}{(local),2} \epsilon} \R)^{\frac{1}{\ell}} 
     \bigg]\quad \text{for}\quad \ell \text{ even}, \\
     r &= \Omega\bigg[ 
      \lnormp{\vH}{(local),2} \sqrt{n+\log(1/\delta ) }t \cdot \L( \L(\frac{\lnormp{\vH}{(global),2}\sqrt{n+\log(1/\delta ) } t}{ \epsilon} \R)^{\frac{1}{\ell}}+\L(\frac{\lnormp{\vH}{(global),2}^2\sqrt{n+\log(1/\delta ) } t}{\lnormp{\vH}{(local),2} \epsilon} \R)^{\frac{1}{\ell+1}}\R) 
     \bigg]\quad \text{else}.
     \end{align}
This is the advertised result.
\end{proof}
\section{Arguments for optimality}
\subsection{Comparing the random estimates with non-random estimates}
In this section, we compare the Trotter error of random Hamiltonians with that of non-random Hamiltonians. Recall the leading order expansion for both
\begin{align}
    \vertiii{\int_0^t [\vec{\CE} (\vH_1,\cdots,\vH_\Gamma,t)]_gdt}_* &\le f(k,g)\cdot \sqrt{p}^g \cdot \L(\lnormp{\vH}{(local),2}\R)^{g-2} \lnormp{\vH}{(global),2}^2 t^g \quad &\text{(random)}\\ 
    \lnormp{\int_0^t [\vec{\CE} (\vH_1,\cdots,\vH_\Gamma,t)]_gdt}{p} &\le f(k,g)\cdot \sqrt{p}^{g(k-1)+1} \cdot \L(\lnormp{\vH}{(local),2}\R)^{g-1} \lnormp{\vH}{(global),2} t^g \quad &\text{(non-random)}.
\end{align}
Indeed, we see similarities between the expressions as the derivation are qualitatively analogous. The random Hamiltonian has better p-dependences (i.e., better concentration) because of the ``external'' randomness from the gaussian coefficients; the non-random Hamiltonian relies solely on the intrinsic randomness of the k-local structure. 
However, counter-intuitively, the random case has a \textit{worse} dependence on system size. Shouldn't the random case be more ``incoherent'' than the non-random case? More precisely, can we not ignore the gaussian coefficents and plug in the non-random estimates?

Of course, both bounds are consistent because we use different norms. The norms $\normp{\cdot}{*}$ capture more stringent errors. For illustration, consider $p=2$ for the norm $\normp{\cdot}{fix,p}$ 
\begin{align}
\displaystyle\BE_{\vO} \lnormp{\vO}{\bar{2}}^2 = \BE_{i,\vO}\lnormp{\vO\ket{i}}{\ell_2}^2 &\le \sup_{\ket{\psi}} \BE_{\vO}\lnormp{\vO\ket{\psi}}{\ell_2}^2 =  \vertiii{\vO}_{fix,2}^2
\end{align}
where the vectors $\{\ket{i}\}$ are any orthonormal basis. In other words, the 2-norm captures the sizes of the \textit{typical} inputs $\ket{i}$; the fixed norm may optimize for \textit{arbitrary} \textit{fixed} inputs $\ket{\psi}$.  
To illustrate the subtle distinction, consider the following Hamiltonian
\begin{align}
    \vH = \sum_{i>j} g_{ij} (\vsigma^y_i\vsigma^z_j + \vsigma^z_i\vsigma^z_j) + g'_{ij}\vsigma^x_i\vsigma^y_j 
\end{align}
and a term in the commutator at order $\CO(t^4)$
\begin{align}
    t^4\sum_{k>i>j} g_{ij}^2 g^{'2}_{ki} \L[\vsigma^y_i\vsigma^z_j + \vsigma^z_i\vsigma^z_j, \L[\vsigma^x_k\vsigma^y_i,\L[\vsigma^x_k\vsigma^y_i, \vsigma^y_i\vsigma^z_j + \vsigma^z_i\vsigma^z_j\R]\R]\R]  \propto  t^4 \sum_{k>i>j} g_{ij}^2 g^{'2}_{ki} \vsigma^x_i.
\end{align}
Importantly, both two terms $g_{ij}(\vsigma^y_i\vsigma^z_j + \vsigma^z_i\vsigma^z_j)$ and $g'_{ki} \vsigma^x_k\vsigma^y_i$ show up twice that square the Gaussians, making them ``coherent''. We evaluate both norms 
\begin{align}
    \lnormp{ \sum_{k>i>j} g_{ij}^2 g^{'2}_{ki} \vsigma^x_i }{\bar{2}} &= \theta(n^2 \sqrt{n}) =\theta\L((\lnormp{\vH}{(local),2})^{3} \lnormp{\vH}{(global),2} \R) \\
    \lnormp{ \sum_{k>i>j} g_{ij}^2 g^{'2}_{ki} \vsigma^x_i }{fix,\bar{2}} &= \theta (n^2 n) = \theta\L((\lnormp{\vH}{(local),2})^{2} \lnormp{\vH}{(global),2}^2 \R)
\end{align}
and observe the latter is $\sqrt{n}$-larger because of optimizing over the fixed input 
\begin{align}
    \lnormp{\sum_i \vsigma^x_i \ket{+\cdots +} }{\ell_2} = n \ne \sqrt{n}= \lnormp{\sum_i \vsigma^x_i }{\bar{2}}. 
\end{align}
Nevertheless, the extra gate complexity cost due to the term $\lnormp{\vH}{(global),2}^2$ vanishes asymptotically at higher-order formulas. 
On the other hand, the optimality for the operator norm bounds (with extra factors of $\sqrt{n}$) is less understood. For example, a commutator term at order $\CO(t^2)$ reads
\begin{align}
   \sum_{k>i>j}g'_{ki}g_{ij}\L[\vsigma^x_k\vsigma^y_i, \vsigma^y_i\vsigma^z_j + \vsigma^z_i\vsigma^z_j\R] \propto \sum_{k>i>j} g'_{ki}g_{ij} \vsigma^x_k\vsigma^x_i \vsigma^z_j.
\end{align}
The expression resembles spin-glass Hamiltonians (e.g., the Sherrington-Kirkpatrick model~\cite{SK_model_75}), whose optimization is an entire research field. Even worse, the coefficients here are correlated, so we leave it as an open problem.

\subsection{Counting lower bounds at early times}\label{sec:counting}
In this section, we give a counting argument suggesting that our gate complexity for random $k$-local Hamiltonians with fixed inputs is optimal at early times. For a particular unitary $e^{\iunit \vH t}$, is it generally hard to rule out the existence of a shorter circuit. Fortunately, lower bounds do exist by a counting argument for a set of unitaries. 

We begin with reviewing the gate complexity for 1d-spatially-local models. There, the gate complexity $nt$ is known to be tight.
\begin{fact}[upper bounds, analog to digital~\cite{haah2020quantum,thy_trotter_error}]
For every piece-wise constant Hamiltonian
\begin{align}
    \vH([T,T+1]) = \vH_i(T), \lV \vH_i \rV\le 1,
\end{align}
product formula approximates it well using $\tilde{\CO}(nt)$ gates
\end{fact}
\begin{fact}[Lower bounds, digital to analog~\cite{haah2020quantum}]
 In the family of piece-wise constant Hamiltonian, there exists $\tilde{\CO}(nt)$ different instances of Boolean circuits, and hence require a circuit of size $\tilde{\Omega}(nt)$. 
\end{fact}

Now, for $k$-local random Hamiltonians, we have shown the unitary evolution $\e^{-\iunit \vH t}$ can be simulated with gate complexity of $\tilde{\CO}(n^k \lnormp{\vH}{(local),2} t)$ for a fixed input state. Is the factor of the number of Hamiltonian terms $\Gamma=n^k$ a feature or a bug? We conjecture it is the former.
\begin{hyp}
Simulation of a typical sample of random $k$-local (SYK normalization) Hamiltonian for time $t$ requires $\tilde{\Omega}(n^kt)$ gates.
\end{hyp}

We present a supportive early-time argument. Consider the random Hamiltonian drawn randomly where the $k$-local matrices are o.n.
\begin{equation}
    \vH^{}_{k,n} := \sum_{i_1< \ldots <i_k \le n} \vH^{}_{i_1\cdots i_k} =\sum_{i_1< \ldots <i_k \le n}  J_{i_1\cdots i_k} \vK^{}_{i_1\cdots i_k}\quad\text{where}\quad \tr(\vK_{i_1\cdots i_k}\vK_{i'_1\cdots i'_k})=d \delta_{\vec{i},\vec{i'}}
\end{equation}
 and $d$ is the dimension of the Hilbert space. Recall, the number of terms is $\Gamma=\binom{n}{k}=\CO(n^k)$. The coefficients are i.i.d. Gaussian\footnote{The argument generalizes to sub-Gaussian coefficients, e.g., bounded coefficients.} with variance
\begin{align}
    \BE[J^2_{\vec{i}}]=
    \CO(\frac{1}{n^{k-1}}).
\end{align}
For a counting argument, we count the number of different Hamiltonians by the size of epsilon net $N(\epsilon)$ via collision probability. 
Draw $N$ i.i.d. sample from the random Hamiltonian ensemble and take a union bound over the chance a pair of random samples collide
\begin{align}
  \Pr\left(\exists  \vH,\vH': \lnormp{ \vH- \vH'}{\infty}<\epsilon \right) \le \binom{N}{2}\Pr(\lnormp{ \vH- \vH'}{\infty}<\epsilon).
\end{align}
So long as RHS $< 1$, there must exist an epsilon-net of size $N$, i.e., $N(\epsilon)$ can be as large as
\begin{align}
    N(\epsilon) = \lfloor{\sqrt{2/\Pr(\lnormp{ \vH- \vH'}{\infty}<\epsilon)}}\rfloor.
\end{align}
To bound the RHS, we reduce to controlling the 2-norm
\begin{align}
    \Pr(\lnormp{ \vH- \vH'}{\infty}<\epsilon) \le \Pr(\lnormp{ \vH- \vH'}{2}<\epsilon\sqrt{d}),
\end{align}
where the dimension of Hilbert space $d$ will be canceled. The 2-norm calculation is a scalar concentration bound
\begin{align}
    \lnormp{ \vH- \vH'}{2}^2 = \lnormp{\sum_{\vec{i}}(J_{\vec{i}}-J'_{\vec{i}})\vK_{\vec{i}}}{2}^2 = \sum_{\vec{i}}(J_{\vec{i}}-J'_{\vec{i}})^2 d \simeq  2\sum_{\vec{i}}J_{\vec{i}}^2 d.
\end{align}
The last line uses that two i.i.d. Gaussians sum to another Gaussian. We can use Bernsteins' inequality for variables $x_{\vec{i}} := J^2_{\vec{i}}$
\begin{fact}[Bernstien's inequality]
For independent variables $x_i$ with variance $\sum_i (x_i-\BE[x_i])^2 = v$ and bound $\labs{x_i-\BE[x_i]} \stackrel{a.s.}{\le} L$, 
\begin{align}
    \Pr\left(|\sum_{i} x_{i} -\BE \sum_{\vec{i}} x_{i} |\ge \delta\right) &\le 2\exp\L(\frac{-\delta^2/2}{v +L\delta/3}\R).
\end{align}
\end{fact}
For our parameters,
\begin{align*}
        \Pr(\lnormp{ \vH- \vH'}{2}<\epsilon\sqrt{D})\le \Pr(\sum_{\vec{i}} J_{\vec{i}}^2\le \epsilon^2/2 )&\le  \Pr\left(|\sum_{\vec{i}} J_{\vec{i}}^2 -\Omega(n) |\ge \Omega(n)-\epsilon^2/2 \right)  \\
        &\lesssim \exp(-\Omega(n^k)) = \exp(-\Omega(\Gamma)).
\end{align*}
where we plugged in $\BE \sum_{\vec{i}} x_{\vec{i}} = \CO(n)$,\ $\delta= n-\epsilon/2=\CO( n)$, $L=\CO( 1/n^{k-1})$, and $v= \CO(n^k/n^{2(k-1)} )=\CO(1/n^{k-2})$.
To estimate the circuit lower bound, consider unitary evolution up to a short time (say $t_*\sim \theta(1)$\footnote{ The argument also works the slightly later scrambling time~\cite{Lashkari_2013,chen2021concentration,Sekino_2008} for $k$-local models $t_* \sim \Omega(\log(n))$ }) .
\begin{align}
    \Pr(\lnormp{e^{\iunit \vH t_*}-e^{\iunit \vH't_*}}{\infty}\le \epsilon ) \stackrel{?}{\lesssim} \Pr(\lnormp{( \vH- \vH')t^*}{\infty}<\epsilon) \le  \exp(\Omega(-\Gamma)).
\end{align}
Unfortunately, there is still a missing step from the Hamiltonian to the unitary evolution; the second inequality is rigorous. If this line holds, then a circuit of size 
\begin{align}
\Omega(\Gamma)=\Omega(n^k) 
\end{align}
is needed at early times $t=\theta(1)$, matching our Trotter bounds for non-random and random Hamiltonians (Theorem~\ref{thm:Trotter_non_random}, Theorem~\ref{thm:Trotter_random_H}).

\appendix
\section{Truncating the Hamiltonian}\label{sec:trunc_H}
When the Hamiltonian has many weak terms and fewer strong terms, directly calling Trotter costs a considerable price for the number of terms $\Gamma$. A common fix is to simulate a truncated Hamiltonian with much fewer terms~\cite{thy_trotter_error,2021_Microsoft_catalysis}
\begin{align}
    \vH = \vH_{tr}+ \delta\vH. 
\end{align}
When the Hamiltonian is $k$-local (including the Fermionic cases and/or in low particle number subspaces), Hypercontractivity (Proposition~\ref{prop:general_pauli_expansion}) quickly applies to the p-norm of the truncated part $\delta\vH$. 

For example, consider the power-law interacting Hamiltonian on a d-dimensional cube, with a total number of sites $n$ 
\begin{align}
    \vH = \sum_{x,y} \vH_{xy}\quad \text{where}\quad \norm{\vH_{xy}} \le \frac{1}{\labs{x-y}^{\alpha}},
\end{align}
and $\alpha <d$. We truncate the Hamiltonian for distance $\labs{x-y}$ larger than a tunable cut-off $\ell$.
Then 
\begin{align}
\lnormp{\e^{i\vH t}-\vS(t/r)^r}{p}  &= \lnormp{\e^{i\vH t} -\e^{i\vH_{tr} t}}{p} + \lnormp{\e^{i\vH_{tr} t}-\vS(t/r)^r }{p}     \\
& \le t \lnormp{\delta \vH}{p} + \lnormp{\e^{i\vH_{tr} t}-\vS(t/r)^r }{p}.
\end{align}
Set $\ell =\theta((\frac{nt^2}{\epsilon})^{\frac{1}{2\alpha-d}} )$ and drop all polynomial factors of $p$, then
\begin{align}
    t \lnormp{\delta \vH}{p} = t \sqrt{ n \ell^{d-2\alpha} } \lesssim \epsilon,
\end{align}
and the gate complexity for typical input states would be (dropping dependence of failure probability $\poly\ln(\delta)$)
\begin{align}
    G = \Omega( n\ell^d \lnormp{\vH_{tr}}{(1),2} ) = \Omega( nt \cdot (\frac{nt^2}{\epsilon} )^{\frac{d}{2\alpha-d}} ),
\end{align}
which is slightly better than the operator norm estimates~\cite{thy_trotter_error}.

\bibliography{ref}
\bibliographystyle{abbrv}

\end{document}